\documentclass{article}
\usepackage{amssymb, amsmath,bm,tikz,verbatim}
\usepackage{xcolor,colortbl}
\usepackage{amsthm}
\usepackage{amsfonts}
\usepackage[round]{natbib}
\usepackage[utf8]{inputenc}
\usepackage[T1]{fontenc}
\usepackage{babel}
\usepackage{hyperref}
\usepackage{amssymb}
\usepackage{longtable}
\usepackage{makecell, cellspace, caption}
\usepackage{tikz}
\usepackage{paralist}
\usepackage{ragged2e}
\newcommand{\todo}[1]{\textbf{\textcolor{red}{[#1]}}}
\usepackage{authblk}
\usetikzlibrary{positioning} 
 \usepackage{graphicx}
 \usepackage{caption}
 
\usepackage{array}
\newcolumntype{L}[1]{>{\raggedright\let\newline\\\arraybackslash\hspace{0pt}}m{#1}}
\newcolumntype{C}[1]{>{\centering\let\newline\\\arraybackslash\hspace{0pt}}m{#1}}
\newcolumntype{R}[1]{>{\raggedleft\let\newline\\\arraybackslash\hspace{0pt}}m{#1}}
\definecolor{tuftsblue}{RGB}{31,55,108}

\usepackage{titlesec,titletoc,bbm}
\titleformat*{\section}{\normalfont\Large\bfseries\color{tuftsblue}}
\titleformat*{\subsection}{\normalfont\large\bfseries\color{tuftsblue}}
\titleformat*{\subsubsection}{\normalfont\bfseries\color{tuftsblue}}
\titleformat*{\paragraph}{\normalfont\bfseries\color{tuftsblue}}
\titleformat{\paragraph}[runin]{\normalsize\bfseries\color{tuftsblue}}{\theparagraph}{1em}{}

\usepackage[noabbrev,capitalize]{cleveref}

\newcommand{\X}{\bm{X}}
\newcommand{\x}{\bm{x}}
\newcommand{\Z}{\bm{Z}} 
\newcommand{\Y}{\bm{Y}} 
\newcommand{\z}{\bm{z}} 
\newcommand{\U}{\bm{U}}
\newcommand{\D}{\bm{D}}
\renewcommand{\a}{\boldsymbol{a}} 
\renewcommand{\b}{\boldsymbol{b}} 
 
\renewcommand{\v}{\boldsymbol{v}}
\renewcommand{\u}{\boldsymbol{u}}

\newcommand{\y}{\boldsymbol{y}}
\renewcommand{\r}{\boldsymbol{r}}
\newcommand{\rb}{\bm{r}_{\x}} 
\newcommand{\ry}{\bm{r}_{\y}} 
\newcommand{\rone}{\bm{r}_{\x}^*}

\newcommand{\Xone}{ \bm{X}^* }
\newcommand{\xone}{ \bm{x}^* }
\newcommand{\boldone}{\bm{1}}
\newcommand{\bias}[1]{ \textrm{Bias} \left( #1 \right) } 
\newcommand{\betaxp}[1]{\beta_x^{#1}}

\newcommand{\Rtwo}[1]{\bm{C} \left( \theta_{#1} \right)}
\newcommand{\Rcovst}[1]{C \left( d, \theta_{#1} \right)}

\newtheorem{thm}{Theorem}
\newtheorem{definition}{Definition}
\newtheorem{notation}{Notational Convention}
\newtheorem{rmk}{Remark}

\newtheorem{lemma}{Lemma}
\newtheorem{cor}{Corollary}
\newtheorem{analysis}{Adjusted Spatial Analysis Method}

\newcommand{\Sigmah}{\hat{\bm{\Sigma}}}
\newcommand{\Sigmahinv}{\hat{\bm{\Sigma}}^{-1}}
\newcommand{\sigmainvnorm}[1]{ || #1 ||_{\bm{\Sigmah}^{-1}}}
\newcommand{\sigmainvnormx}[1]{ || #1 ||_{\bm{\Sigmah}_x^{-1}}}
\newcommand{\euclnorm}[1]{ || #1 ||}
\newcommand{\sigmainvdot}[2]{ \langle #1, #2 \rangle_{\bm{\Sigmah}^{-1}}}
\newcommand{\sigmainvdotx}[2]{ \langle #1, #2 \rangle_{\bm{\Sigmah}_x^{-1}}}
\newcommand{\euclmetric}[2]{ \langle #1, #2 \rangle}
\newcommand{\sigmahatgen}[1]{  \bm{\Sigmah}_{#1}^{-1}}
\newcommand{\euclangle}[2]{\theta_{#1,#2}}
\newcommand{\sigmainvangle}[2]{\phi_{#1,#2}}

\newcommand{\ISdot}[2]{ \langle #1, #2 \rangle_{(\bm{I} - \hat{\bm{S}}_y)}}

\newcommand{\cb}[1]{{\color{blue}CB note: #1}}

\title{Re-thinking Spatial Confounding in Spatial Linear Mixed Models}

\date{}

\author[1]{Kori Khan}
\author[2]{Candace Berrett}
\affil[1]{Department of Higher Education, State of Ohio}
\affil[2]{Department of Statistics, Brigham Young University}

\begin{document}

\maketitle

\begin{abstract}
\sloppypar{In the last two decades, considerable research has been devoted to a phenomenon known as spatial confounding. Spatial confounding is thought to occur when there is multicollinearity between a covariate and the random effect in a spatial regression model. This multicollinearity is considered highly problematic when the inferential goal is estimating regression coefficients and various methodologies have been proposed to attempt to alleviate it. Recently, it has become apparent that many of these methodologies are flawed, yet the field continues to expand. In this paper, we offer a novel perspective of synthesizing the work in the field of spatial confounding. We propose that at least two distinct phenomena are currently conflated with the term spatial confounding. We refer to these as the ``analysis model'' and the ``data generation'' types of spatial confounding. We show that these two issues can lead to contradicting conclusions about whether spatial confounding exists and whether methods to alleviate it will improve inference. Our results also illustrate that in most cases, traditional spatial linear mixed models \textit{do} help to improve inference on regression coefficients. Drawing on the insights gained, we offer a path forward for research in spatial confounding.}
\end{abstract}

\section{Introduction}
In myriad applications, standard regression models for spatially referenced data can result in spatial dependence in the residuals. For the better part of a century, the solution to this problem was to use a spatial regression model. In these models, a spatial random effect was introduced to account for the residual spatial dependence and thereby (theoretically) improve inference, whether the inferential goal was associational or predictive. 

\sloppypar{This practice continued unchallenged until about two decades ago. At that time, a phenomenon now known as spatial confounding was introduced by \citet{Reich} and \citet{Hodges_fixedeffects} \citep[see also,][]{Pac_spatialconf}. If, historically, spatial statisticians believed that incorporating spatial dependence with spatial regression models would improve inference, now those interested in spatial confounding suggest that incorporating spatial dependence with traditional models will distort inference. Originally focused on a setting where the estimation of individual covariate effects was important, interest in spatial confounding has since expanded to other inferential focuses \citep[e.g.,][]{page2017estimation,papadogeorgou2019adjusting}. Spatial confounding is typically described as occurring when there is multicollinearity between a spatially referenced covariate and a spatial random effect. It is thought to be quite problematic. For example, \citet{marques2022mitigating} states spatial confounding can lead to ``severely biased'' regression coefficients, \citet{Reich} asserts that it can lead to ``large changes'' in these estimates, and \citet{prates2019alleviating} argues that both the ``sign and relevance of covariates can change drastically'' in the face of spatial confounding.}

Although many of these claims are not empirically supported, research into spatial confounding and methods to alleviate it has exploded \citep{Hughes,Hanks,Hefley,Thaden,chiou2019adjusted,prates2019alleviating,Keller_spatialconf,adin2021alleviating,azevedo2021mspock,hui2021spatial,nobre2021effects,azevedo2022alleviating,dupont2022spatial+,marques2022mitigating}. A closer look at the body of work highlights inconsistencies in definitions of spatial confounding as well as the purported impact it can have on inference \citep{Hanks,khan2020restricted,nobre2021effects,zimmerman2021deconfounding,narc:2024}. Recently, many of the methods designed to alleviate spatial confounding have been shown to lead to counterintuitive results by \citet{khan2020restricted} and have even been classified as ``bad statistical practice'' \citep{zimmerman2021deconfounding}. Yet, efforts to study and alleviate spatial confounding continue without any attempt to address these observations, increasingly influencing new fields of study such as causal inference and even criminology \citep{reich2021review,kelling2021modeling}.

In this paper, we, 
\begin{enumerate}
    \item synthesize the existing body of work in spatial confounding, reviewing it in the context of historical practices from spatial statistics; 
    \item characterize two distinct, albeit related, phenomena currently conflated with the term spatial confounding; and 
    \item show, through theoretical and simulation results, that these two issues can lead to contradicting conclusions about whether spatial confounding exists and whether methods to alleviate it will improve inference. 
\end{enumerate}
Importantly, by examining spatial confounding in this way, these three key understandings show how ignoring the nuances of  ``spatial confounding'' can lead to methodologies that distort inferences in the very settings they are designed to be used.

The rest of this paper is organized as follows: We introduce the analytical setup in \cref{sec:background}. Using this setup, we provide an overview of spatial confounding in the broader context of spatial statistics. \cref{sec:nomenclature} provides a framework for understanding the two types of spatial confounding and illustrates how current (and past) research fits into this scheme. It also explores how efforts to mitigate spatial confounding can be organized into this framework. \cref{sec:mainresults} introduces theoretical results assessing the impact of both sources of spatial confounding on the bias for a regression coefficient.  In \cref{sec:simstudies}, we use simulation studies to explore settings identified in the literature as situations in which spatial confounding will lead to increased bias in regression coefficient settings. We illustrate that traditional spatial analysis models often outperform non-spatial models and models designed to alleviate spatial confounding in these cases. Finally, in \cref{sec:conclusion}, we propose a clear path toward resolving the contradictions explored in this paper.

\section{Background} \label{sec:background}
We begin by introducing the analytical setup that will be used throughout the rest of the paper (\cref{analytical}). We then use it to briefly explain, using the historical background of spatial modeling (\cref{subsec:spatregmodels}), how spatial confounding became a topic of concern in spatial statistics research and explore where it has gone since (\cref{subsec:histconf}).

\subsection{Analytical Set-Up} \label{analytical}
This paper distinguishes between a \textit{data generating} model and an \textit{analysis} model. The former is a model meant to approximate how the data likely arose, while the latter is used to analyze the observed data.

Spatial regression models are traditionally used when there is residual spatial dependence after accounting for measured variables. Residual spatial dependence is thought to be the result of either an unobserved, spatially varying variable or an unobserved spatial process \citep[][]{Waller}. To define a data-generating model, we focus on the former as this most closely matches the intuition motivating efforts to mitigate spatial confounding \citep[see e.g.,][]{Reich,Pac_spatialconf,page2017estimation,dupont2022spatial+}.

Specifically, we assume $y_i$ is observed at location $\bm{s}_i \in \mathbb{R}^2$ for $i= 1, \ldots, n$ and it can be modeled as follows:
\begin{flalign} \label{eq:model_0}
	\textrm{\textbf{Generating Model: }}    y_i (\bm{s}_i) = \beta_0 + \beta_x x_i (\bm{s}_i) + \beta_z z_i (\bm{s}_i) + \epsilon_i,
\end{flalign}
where $\x\left( \bm{s} \right) = \left( x_1(\bm{s}_1), \ldots, x_n (\bm{s}_n) \right)^T$ and $\z \left( \bm{s} \right)= \left( z_1(\bm{s}_1), \ldots, z_n (\bm{s}_n) \right)^T$ are each univariate variables, $\bm{\epsilon} = \left( \epsilon_1, \ldots, \epsilon_n \right)^T$ is the vector of errors with mean $\bm{0}$ and variance-covariance matrix $\sigma^2 \bm{I}$, and $\bm{\phi} = \left( \beta_0, \beta_x, \beta_z, \sigma^2 \right)^T$ are unknown parameters. 

Throughout this paper, we assume that $\x \left( \bm{s} \right)$ and $\y\left( \bm{s} \right)$ are observed and $\z\left( \bm{s} \right)$ is unobserved. We also assume that the primary inferential interest is on $\beta_x$. We consider three possible approaches to modeling the relationship between $\y \left( \bm{s} \right)$ and $\x \left( \bm{s} \right)$: 1) a non-spatial linear approach, 2) a ``traditional'' spatial approach, and 3) an ``adjusted'' spatial approach. Each framework is associated with one or more analysis models that can be fit to the observed $\y\left( \bm{s} \right)$ and $\x\left( \bm{s} \right)$. 
\begin{eqnarray}
	\textrm{ \small \textbf{Non-Spatial Analysis Model:}}&&  y_i (\bm{s}_i)= \beta_0 + \betaxp{NS} x_i (\bm{s_i})  + \epsilon_i  \label{eq:OLSmodel} \\
\textrm{\small \textbf{Spatial Analysis Model:}}&&  y_i (\bm{s}_i)= \beta_0 + \betaxp{S} x_i (\bm{s}_i) + g(
	\bm{s}_i) + \epsilon_i \label{eq:genericSpatial}\\
\textrm{ \small\textbf{Adj. Spatial Analysis Model:}}&&  \tilde{y}_i (\bm{s}_i)= \beta_0 + \betaxp{AS} \tilde{x}_i (\bm{s}_i) + h(
	\bm{s}_i) + \epsilon_i \label{eq:genericadjSpatial}
\end{eqnarray}
For \eqref{eq:OLSmodel}--\eqref{eq:genericadjSpatial}, $\epsilon_i$ are i.i.d with mean $0$ and unknown variance $\sigma^2$. The regression coefficients $\beta_0$, $\betaxp{NS}$, $\betaxp{S}$, and $\betaxp{AS}$ are unknown. We note that $\sigma^2$ and $\beta_0$ will vary based on the analysis model chosen. In other words, to be precise, we would use notation such as $\beta_0^{NS}$,$\beta_0^{S}$, and $\beta_0^{AS}$. As our primary interest is $\beta_x$, we refrain from doing so for the sake of simplicity. 

The spatial random effects $g(\bm{s})$ and $h(\bm{s})$ are assumed to have mean zero and unknown, positive-definite variance-covariance matrices. We note that models relying on Gaussian Markov Random Fields (GMRFs) can be considered as special cases of this if the variance-covariance matrices are defined to be pseudo-inverses of the singular precisions \citep{Pac_ICAR}. The tildes over $\y\left( \bm{s} \right)$ and $\x\left( \bm{s} \right)$ in \eqref{eq:genericadjSpatial} reflect that they may be functions of the originally observed $\y\left( \bm{s} \right)$ and $\x\left( \bm{s} \right)$ respectively. In future sections, we distinguish between a realization of the variables $\x\left( \bm{s} \right)$ and $\z\left( \bm{s} \right)$ and the stochastic processes that could have generated such realizations. We use capital letters (e.g., $\X(\bm{s})$ and $\Z(\bm{s})$ ) to refer to stochastic processes and lowercase letters to indicate a realization of the variables (e.g., $\x(\bm{s})$ and $\z(\bm{s})$ ). After this, we drop the notation indicating the dependence on spatial location unless it is needed for clarity. 

We take a moment here to distinguish between two different types of spatial data: geostatistical or point-referenced data and areal data.  For more details on the two types of data, see for example, \cite{Cressie93, banerjee2004, Waller}.  Geostatistical data are observed at specific points over a continuous field; for example, soil moisture measurements are observed at specific locations in a field and temperature measurements are obtained from weather stations at specific coordinates.  In contrast, areal data are observed (often an aggregation of measurements) over a partition within a region of interest.  For example, average income level or other demographic variables are generally reported for an entire county in the United States.  Geostatistical data traditionally model dependence by using the distance (Euclidean, great circle, or other) between the locations of interest and a correlation function.  Areal data often model dependence through a neighborhood matrix, where neighbors are defined by, for example, whether the partitions share a border or not.  Both types of modeling approaches account for Tobler's First Law of Geography: things closer together are more similar (correlated) than things farther apart \citep{tobler1970}.  However, the strength of that correlation and how fast or slow that correlation decays will inherently be different. For example, for areal data, the correlation decays more slowly when each partition has a larger number of neighbors; whereas for geostatistical data, the correlation decays more slowly due to the parameters of the correlation function.  That said, \cite{rue:2002} showed that the dependence structure of one form can approximate the dependence structure of another quite well for a fixed set of locations.  In future sections we examine these two types of data separately.  However, for now, we’ll simplify matters by focusing solely on the general precision matrix related to the spatial dependence.  

\subsection{Spatial Models} \label{subsec:spatregmodels}
When there is residual spatial dependence, the conventional wisdom in spatial statistics literature is that a model that accounts for this will offer better inference than a model that does not account for it \citep{Cressie93,Waller,bivand2008applied}. Historically, this view first appeared in the context of geostatistics and interpolation efforts.  There, the goal was to improve predictions for the values of a stochastic process at unobserved locations \citep[e.g.,][]{wikle2010low}. In other words, $\betaxp{S}$ in \eqref{eq:genericSpatial} was merely a tool to de-trend the data, and the primary interest was often estimating the variance-covariance matrix of the spatial random effect. The idea that accounting for spatial dependence improves inference later inspired many popular spatial models proposed for areal data. These models were often developed with the goal, either implicit or explicit, of ensuring that $\betaxp{S}$ in \eqref{eq:genericSpatial} was ``close'' to $\beta_x$ in \eqref{eq:model_0} \citep{besag1991bayesian,Hodges_fixedeffects}. In recent decades, the lines delineating methods for geostatistical data and areal data have become blurred with advancements in computing and the popular class of models proposed by \citet{Diggle}. However, across analysis goals and data types, the consensus continued that models accounting for spatial dependence should be preferred over models that did not account for spatial dependence.

Recently, however, this view has shifted. The challenge to the prevailing view arose in a line of research about a phenomenon now known as ``spatial confounding.''

\subsection{Spatial Confounding}\label{subsec:histconf}
\citet{Clayton} is often referenced as the first article to describe spatial confounding. These authors noticed what they referred to as ``confounding by location:''  estimates for regression coefficients changed when a spatial random effect was added to the analysis model. \citet{Clayton} interpreted this as a favorable change -- one in which the estimates of the association between a response and an observed covariate were adjusted to account for an unobserved spatially-varying confounder \citep[see also,][]{Hodges_fixedeffects}. The modern conceptualization of spatial confounding arose in work by \citet{Reich} and \citet{Hodges_fixedeffects}. These articles were the first to suggest that fitting spatial models could induce bias in the estimates of the regression coefficients and an ``over-inflation'' of the uncertainty associated with these estimates. These works have inspired a serious and active line of research into the phenomena of spatial confounding \citep{Pac_spatialconf,Hughes,Hefley,Thaden,prates2019alleviating,azevedo2021mspock,nobre2021effects,yang2021estimation,dupont2022spatial+,marques2022mitigating}.

\sloppypar{Spatial confounding is almost always introduced as an issue of multicollinearity between a spatially varying covariate and a spatial random effect in a spatial analysis model \citep{Reich,Hodges_fixedeffects,Hefley,Thaden,reich2021review,dupont2022spatial+}. This statement is often deemed sufficient to identify the phenomena of spatial confounding. However, there is no consensus on a formal definition for spatial confounding. While there have been two previous efforts to formalize spatial confounding, both were definitions considering special cases of a broader phenomenon \citep{Thaden,khan2020restricted}.}

Despite the ambiguity of spatial confounding as a concept, researchers using the term have developed shared expectations for the phenomenon. These expectations have, in turn, shaped multiple methods aimed at alleviating spatial confounding. Some researchers have noticed inconsistencies and contradictions in some of the conclusions reached by the spatial confounding literature. For example, \citet{Hanks} and \citet{nobre2021effects} have both observed that a distortion in inference for $\beta_x$ can occur in the absence of stochastic dependence between $\x$ and $\z$, contradicting some stated expectations for spatial confounding. These inconsistencies have largely remained unresolved even as research on spatial confounding has increasingly begun influencing other lines of  work, such as causal inference \citep[e.g.,][]{papadogeorgou2019adjusting,reich2021review}.

We propose that some of these contradictions arise because at least two distinct categories of issues are being studied by researchers in spatial confounding. Loosely speaking, we can think of these categories as encompassing a data generation phenomenon and an analysis model phenomenon. Importantly, once teased apart, these two issues can lead to different conclusions about whether spatial confounding is present and whether spatial analysis models should be adjusted.

\section{Types of Spatial Confounding} \label{sec:nomenclature}
As previously noted, spatial confounding is typically described as an issue of multicollinearity between a spatially varying covariate and a spatial random effect in a spatial analysis model. It appears, however, that researchers can disagree about the source of multicollinearity as well as what it means for a covariate to be spatially varying (in a problematic sense). In this section, we tease apart what we refer to as analysis model spatial confounding (\cref{subsec:analysis}) and data generation spatial confounding (\cref{subsec:datagen}).  In \cref{fig:organization}, we summarize how the problematic relationships that are thought to cause spatial confounding differ by these two types of spatial confounding, and we  elaborate on these relationships below. We emphasize that this framework is not currently in use. Instead, it is a novel attempt meant to help organize some of the existing conceptualizations of spatial confounding in the literature. Importantly, many articles can have references to both types of spatial confounding within them. In the following discussion, we sort works based on the article's primary focus. We frame the recommendations for alleviating each type of confounding in the context of these two sources in \cref{subsec:alleviating}.

\begin{figure}
\begin{center}
\begin{tikzpicture}[place/.style={rectangle,draw=black},place2/.style={rectangle,draw=white},
   transition/.style={rectangle,draw=black!50,fill=black!20,thick}]
  \node[place] (1) {\textbf{Spatial Confounding}};
  \node[place2] (3) [below=of 1] {\hspace*{.1in}};
  \node[place] (7) [left=of 3] {\textbf{Analysis Model}};
  \node[place] (2) [right=of 3] {\textbf{Data Generation}};
 \node[place, align=center] (10) [below=of 7, label=above:\emph{Problematic Relationship:}]{$\x$ and $\hat{\bm{\Sigma}}^{-1}$ in\\ Spatial Analysis Model (3)};
  \draw [thick,->, shorten >= .1in] (1) to (7);
  \draw [thick,->, shorten >= .1in] (1) to (2);
  \node[place, align=center] (8) [below=of 2, label=above:\emph{Problematic Relationship:}] {$\x$ and $\z$ (or $\X$ and $\Z$) in\\  Generating Model (1)};
   \draw [thick,->, shorten >= .25in] (2) to (8);
   \draw[thick,->, shorten >= .25in] (7) to (10);
\end{tikzpicture}
\captionof{figure}{A diagram to illustrate the primary sources of spatial confounding by type of spatial confounding (analysis model or data generation).}
\label{fig:organization}
\end{center}
\end{figure}
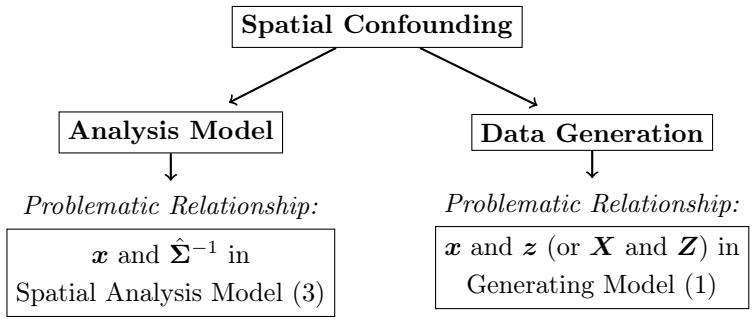

\subsection{Analysis Model Spatial Confounding} \label{subsec:analysis}
\citet{Reich} and \citet{Hodges_fixedeffects} are the works that introduced the modern conceptualization of spatial confounding. These papers, and many of the works they inspired, focused on what we will refer to as analysis model spatial confounding. Research motivated by the analysis model issue often does not consider how $\y$ or $\x$ were generated. In other words, these works do not assume there is a missing $\z$ or a data generation model of the form \eqref{eq:model_0} \citep{Hodges_fixedeffects}. Instead, this conceptualization of spatial confounding focuses on the relationship between an observed $\x$ and the spatial random effect in a spatial analysis model \citep{Reich,Hodges_fixedeffects,Hughes,Hanks,Hefley,prates2019alleviating,azevedo2021mspock,hui2021spatial}. 

In this line of work, identifying the problematic source of multicollinearity and defining what it means for $\x$ to be spatially varying both rely on the analysis model. More specifically, in the context of our analytical set-up, they typically rely on the  eigenvectors of the estimated precision matrix $\hat{\bm{\Sigma}}^{-1}$ of the spatial random effect $g(\bm{s})$ in \eqref{eq:genericadjSpatial} \citep{Reich,Hanks,Hefley,prates2019alleviating,azevedo2021mspock}. For example, statistics  developed to identify spatial confounding involve using both the observed $\x$ and the estimated precision matrix for a particular spatial analysis model \citep{Reich,Hefley,prates2019alleviating}. These statistics all identify, loosely, whether $\x$ is correlated with low-frequency eigenvectors of a decomposition of $\hat{\bm{\Sigma}}^{-1}$. Similarly, $\x$ is considered spatially varying (in a problematic sense) if it is highly correlated with such a low-frequency eigenvector of $\hat{\bm{\Sigma}}^{-1}$. We note that, in spatial confounding literature, no one has precisely defined what it means for an eigenvector to be low-frequency \citep[but, see][]{Reich_var}. However, when displayed graphically, they tend to show spatial patterns where nearby things are more similar than others. Thus, the problematic relationship which causes data generation spatial confounding is thought to be primarily between $\x$ and $\hat{\bm{\Sigma}}^{-1}$, as summarized in \cref{fig:organization}.

There are several common beliefs underlying work focused on this conceptualization of spatial confounding. First, spatial confounding occurs as a result of fitting a spatial analysis model. While a distortion to inference can be expected in any spatial analysis model \citep{Hodges_fixedeffects}, it is plausible that the degree of distortion may vary based on the particular analysis model chosen \citep[see e.g.,][]{Hefley}. Second, efforts should be taken to determine whether spatial confounding needs to be adjusted for in the analysis model. In this line of work, many authors acknowledge that it is not clear when spatial confounding needs to be accounted for \citep{Hanks,prates2019alleviating,hui2021spatial}. In other words, there is at least an implicit understanding that spatial analysis models may still be preferable over adjusted spatial analysis models at times. Finally, determining whether spatial confounding exists will involve studying characteristics of the observed data (in particular $\x$) along with properties of the chosen analysis model.

\subsection{ Data Generation Spatial Confounding } \label{subsec:datagen}
In work that focuses on data generation spatial confounding, researchers often do assume that $\y$ is generated from a model of the form \eqref{eq:model_0}. In the context of our analytical set-up, the interest is typically on how the relationship between $\X$ or $\Z$ (or alternatively $\x$ and $\z$) impacts inference on $\betaxp{S}$ when a spatial analysis model of the form \eqref{eq:genericSpatial} is used to fit the data \citep{Pac_spatialconf,page2017estimation,Thaden,nobre2021effects,dupont2022spatial+,narc:2024}. 

In this line of work, spatial confounding is often defined as an issue of multicollinearity (often defined as strong correlation) between $\x$ and $\z$ \citep[or $\X$ and $\Z$; see, for example,][]{dupont2022spatial+,Thaden}. However, the source of the multicollinearity and the definition of spatially varying (in the problematic sense) are not always clear. \citet{Pac_spatialconf} has shaped much of the current work focused on data generation spatial confounding, as well as many of the most recent methods designed to alleviate spatial confounding \citep[see e.g.,][]{page2017estimation,Thaden,keller2020selecting,dupont2022spatial+,marques2022mitigating,narc:2024}. In that article and the many that followed, researchers make assumptions about the variables ($\x$ and $\z$) or the stochastic processes that generated them ($\X$ and $\Z$).

Researchers who focus on $\X$ and $\Z$ often assume that these processes are generated from spatial random fields parameterized by some set of known parameters \citep{Pac_spatialconf,page2017estimation,nobre2021effects,narc:2024}. $\X$ and $\Z$ are typically assumed to be generated in such a way that $\X$ has two components of spatial structure: 1) one that is shared with $\Z$ (the confounded component), and 2) one that is not shared with $\Z$ (the unconfounded component). Based on characteristics of these assumed processes, theoretical results or observations have been used to identify when fitting a spatial analysis model of the form \eqref{eq:genericSpatial} will distort inference on $\beta_x$ \citep{Pac_spatialconf,nobre2021effects}. In other words, the problematic relationship is between $\X$ and $\Z$, as summarized in \cref{fig:organization}. 

Most of the theoretical results related to the data generation source of spatial confounding focus on the processes, $\X$ and $\Z$. However, when it comes to methods designed to alleviate spatial confounding, there can be assumptions made about the specific realizations $\x$ and $\z$. For example, \citet{Thaden}  and \citet{dupont2022spatial+} assume that $\x$ is a linear combination of $\z$ and Gaussian noise. In these cases, $\z$ is either chosen to have a fixed spatial structure or is generated from a spatial random field or process. The focus on the relationship between $\X$ and $\Z$ suggests that the problematic multicollinearity is between $\x$ and $\z$, as summarized in \cref{fig:organization}. The fact that some methods designed to alleviate spatial confounding focus on situations where $\x$ and $\z$ are collinear or highly correlated supports this idea. However, the theoretical results in this line of work are usually not related to characteristics of the observed realization $\x$ (or $\z$), and the assumptions made in the theoretical results do not always ensure \emph{empirical} correlation between a given set of realizations $\x$ and $\z$. Similarly, the characteristics of a particular realization $\x$ are not assessed in determining whether it is spatially dependent in a problematic sense. The underlying belief may be that if $\x$ and $\z$ are collinear or highly correlated and ``spatial,'' there will be collinearity between $\x$ and a spatial random effect in an analysis model. However, papers in this line of work spend very little time discussing the impact of spatial analysis models. For example, \citet{Thaden}  defines spatial confounding as occurring when: 1) $\X$ and $\Z$ are stochastically dependent, 2) $\textrm{E} \left( \Y | \X, \Z \right) \neq  \textrm{E} \left( \Y | \X \right)$, and 3) $\Z$ has a ``spatial'' structure. Notice that this definition emphasizes the relationship between $\X$, $\Z$, and $\Y$, and it mirrors more general definitions of confounders in causal inference research. It is not entirely clear what it means for $\X$ to have spatial structure or why it is problematic for $\X$ to have such a structure. More importantly, by this definition, spatial confounding exists regardless of the analysis model chosen. 

We note that not every paper completely ignores the analysis model. For example, \citet{dupont2022spatial+} explicitly stated they were viewing spatial confounding from the perspective of fitting spatial models via thin-plate splines. While they stated that the smoothing that comes from fitting a spatial model  contributes to the problem, the relationship between $\x$ and $\z$ remained emphasized. For example, the authors emphasized ``if the correlation between the covariate and the spatial confounder is high, the smoothing applied to the spatial term in the model can disproportionately affect the estimate of the covariate effect.'' In other words, it did not appear that the smoothing alone was problematic. For this reason, we group this work here, rather than the analysis model spatial confounding, although we note this work is one with elements of both types of spatial confounding. 

We take a moment to highlight several notable beliefs commonly found in work that focuses on the data generation spatial confounding.  These beliefs may or may not be true, but appear to be consistent across this line of work.  First, the primary source of spatial confounding comes from the (potentially unknown) process that generated the data rather than the model fitting process. Second, fitting a spatial analysis model will lead to distortion in inferences when spatial confounding is present. However, here, spatial analysis models -- whether of the form of \eqref{eq:genericSpatial}, a generalized additive model (GAM), or something else -- are often treated as interchangeable. There is often no exploration of the impact of a particular choice of a spatial model on inference, and inferior inferences for one type of spatial model are assumed to hold for other spatial models. Finally, it seems researchers assume the observed data (i.e., $\y$ and $\x$) do not give insight into whether spatial confounding is present or should be accounted for in analyses.

\subsection{Approaches to Alleviating Spatial Confounding} \label{subsec:alleviating}
There have been numerous methods designed to alleviate spatial confounding. In this sub-section, we take a moment to point out that most of them can be categorized as being motivated by either the analysis model or data generation type of spatial confounding, as outlined in \cref{fig:relieve}. 

The first methods to alleviate spatial confounding were motivated by the analysis model source of spatial confounding. For areal analyses, \citet{Reich} and \citet{Hodges_fixedeffects} first proposed a methodology sometimes known as restricted spatial regression. This method suggested to, in a sense,  replace the spatial random effect $g(\bm{s})$ in a spatial analysis model with a new spatial random effect $h(\bm{s})$ in an adjusted spatial analysis model. This new spatial random effect is projected onto the orthogonal complement of the column space of $\x$. By ``smoothing'' orthogonally to the fixed effects, this methodology aimed to alleviate collinearity between the $\x$ and the estimated variance-covariance matrix of $h(\bm{s})$. In doing so, it directly addresses the analysis model source of confounding. This approach motivated and continues to motivate many further methodologies designed to alleviate spatial confounding \citep{Hughes,Hanks,chiou2019adjusted,prates2019alleviating,adin2021alleviating,azevedo2021mspock,hui2021spatial,marques2022mitigating}. Most of these methods continue to involve changing the spatial random effect (or analogous of it for other models) in the spatial analysis model. In other words, the adjustment from a model of the form \eqref{eq:genericSpatial} to \eqref{eq:genericadjSpatial} primarily involves replacing the spatial random effect, and the data remains unaltered. As noted previously, these adjusted analysis models are typically offered with the caveat that there may be some situations when traditional analysis models would be more appropriate (although there is no consensus on when that is). 

In this paper, we do not explore methods influenced by restricted spatial regression in the rest of the paper. Recently, restricted spatial regression models have been shown to perform poorly. \citet{khan2020restricted} demonstrated that inference on $\beta_x$ is often worse with restricted spatial regression analysis models than with non-spatial analysis models. \citet{zimmerman2021deconfounding} subsequently offered a more in-depth, thorough review of restricted spatial regression analysis models. These authors showed that smoothing orthogonally to the fixed effects distorted inference  for various inferential goals and concluded that employing such analysis models was ``bad statistical practice.'' Importantly, both \citet{khan2020restricted} and \citet{zimmerman2021deconfounding} show that the point estimates obtained from restricted spatial regression models are equivalent to those obtained from a non-spatial model.

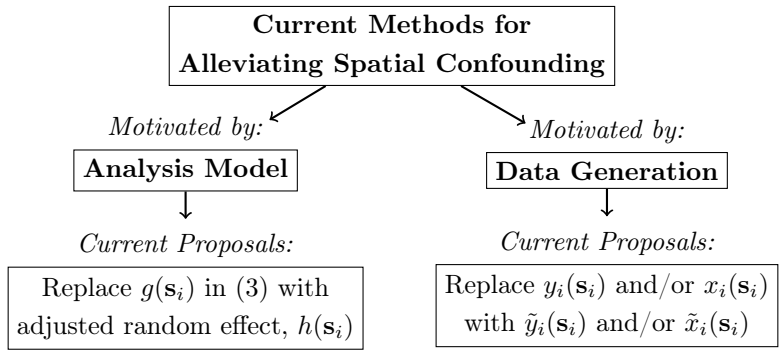
\begin{figure}
\begin{center}
\begin{tikzpicture}[place/.style={rectangle,draw=black},place2/.style={rectangle,draw=white},
   transition/.style={rectangle,draw=black!50,fill=black!20,thick}]
  \node[place, align=center] (11) {\textbf{Current Methods for}\\ \textbf{Alleviating Spatial Confounding}};
    \node[place2] (13) [below=of 11] {\hspace*{.1in}};
  \node[place, align=center] (17) [left=of 13, label=above:\emph{Motivated by:}] {\textbf{Analysis Model}};
  \node[place, align=center] (12) [right=of 13, label=above:\emph{Motivated by:}] {\textbf{Data Generation}};
 \node[place, align=center] (20) [below=of 17, label=above:\emph{Current Proposals:}]{Replace $g(\mathbf{s}_i)$ in \eqref{eq:genericSpatial} with\\ adjusted random effect, $h(\mathbf{s}_i)$};
  \draw [thick,->, shorten >= .3in] (11) to (17);
  \draw [thick,->, shorten >= .3in] (11) to (12);
  \node[place, align=center] (18) [below=of 12, label=above:\emph{Current Proposals:}] {Replace $y_i(\mathbf{s}_i)$ and/or $x_i(\mathbf{s}_i)$\\ with $\tilde{y}_i(\mathbf{s}_i)$ and/or $\tilde{x}_i(\mathbf{s}_i)$};
   \draw [thick,->, shorten >= .25in] (12) to (18);
   \draw[thick,->, shorten >= .25in] (17) to (20);
\end{tikzpicture}
\captionof{figure}{A diagram to illustrate current recommendations to relieve issues related to spatial confounding by the motivating source of spatial confounding.}\label{fig:relieve}
\end{center}
\end{figure}

Researchers motivated by data generation spatial confounding rely heavily on assumptions about how the data arose when developing methodology to alleviate spatial confounding. Thus, there can be various formulations. We focus on two methodologies proposed by \citet{Thaden} and \citet{dupont2022spatial+} as illustrative examples of such approaches (described in more detail in \cref{subsec:biasadj}). 
In both these works, the authors assume that the observed data are truly from a model with a form similar to \eqref{eq:model_0} \citep[in simulation studies ][ introduced another unobserved spatial random effect to this model]{dupont2022spatial+} and that $\x = \beta_{z:x} \z + \bm{\epsilon}_x$, where $\bm{\epsilon}_x$ is Gaussian noise. Based on these assumptions, the authors proposed methodologies to alleviate spatial confounding that replace (or are equivalent to replacing) either $\y$ or $\x$ in the analysis model.  The details of these approaches are given in \cref{subsec:biasadj}.

Recall, \citet{Thaden} offered little discussion of the impact of the spatial analysis model on inference, and \citet{dupont2022spatial+} felt that their proposed methodology would work in settings beyond the thin plate splines setting they explored. Subsequent work has claimed both approaches are useful for other types of spatial models \citep{schmidt2021discussion,dupont2022spatial+}. As discussed in \cref{subsec:datagen}, this is characteristic of work motivated by data generation spatial confounding.  The unspoken belief is that something must be known about how the data were generated to appropriately analyze it. If the data were truly generated in line with the assumptions made, the proposed methodologies should be superior to traditional spatial regression analysis models (and non-spatial analysis models). 

In the rest of this paper, we give theoretical results that show that both the analysis model and data generation types of spatial confounding can impact inference, sometimes in competing ways. Importantly, we also show that methods designed to alleviate spatial confounding that focuses on only one type of spatial confounding can, in some cases, distort inference more than a spatial regression model.

\section{Two Views of Spatial Confounding Bias} \label{sec:mainresults}

While the last section was focused on defining two sources of confounding and putting current literature into this framework, in this section we explore the bias in the estimates of $\beta_x$ for various analysis models.   We first focus on the bias from the non-spatial and spatial analysis models \eqref{eq:OLSmodel} and \eqref{eq:genericSpatial} in \cref{subsec:nssp}.  We introduce theoretical results for the bias of these two models from a data generation confounding perspective (\cref{subsubsec:datagen}) and analysis model confounding perspective (\cref{subsubsec:analysis}).  In \cref{subsec:biasadj} we look at bias for two types of adjusted spatial models of the form \eqref{eq:genericadjSpatial}.


\subsection{Bias: Non-Spatial and Spatial Analysis Models} \label{subsec:nssp}
In this sub-section, we consider how the data generation and analysis model types of spatial confounding may impact bias in the estimation of $\beta_x$. We consider this for the non-spatial analysis and spatial analysis models.  To do so, we follow the set-up explored in \citet{Pac_spatialconf}. This article has shaped much of the current work focused on the data generation issue, as well as many of the most recent methods designed to alleviate spatial confounding \citep[see, e.g.,][]{dupont2022spatial+,Thaden,page2017estimation,keller2020selecting,marques2022mitigating,narc:2024}. 

Mirroring the work in \citet{Pac_spatialconf}, we assume that our response variable was generated from a model of the form \cref{eq:model_0}. However, instead of a particular set of realizations for $\bm{x}$ and $\bm{z}$, we use the processes $\bm{X}$ and $\bm{Z}$: 
\begin{eqnarray} \label{stochastic}
   \Y (\bm{s}_i) = \beta_0 + \beta_x \X (\bm{s}_i) + \beta_z \Z (\bm{s}_i) + \epsilon_i,
\end{eqnarray}
where $\epsilon_i$ is defined as in \eqref{eq:model_0}. We assume that $\bm{X}$ and $\bm{Z}$ are each generated from Gaussian random processes with positive-definite, symmetric covariance structures. In \citet{Pac_spatialconf}, the author considered two settings: one in which the author stated there was no confounding in the data generation process and one in which the author stated there was confounding in the data generation process. We restrict our attention to the situation where there is confounding in the data generation process.

Throughout this section, we assume $\X$ and $\Z$ are generated from Gaussian processes with Mat\'ern spatial correlations: 
\begin{eqnarray} \label{maternclass}
\bm{C}\left( d ; \theta, \nu \right) = \frac{1}{\Gamma(\nu) 2^{\nu-1}} \left( \frac{ 2 \sqrt{\nu} d }{\theta} \right)^{\nu} K_{\nu} \left( \frac{ 2 \sqrt{\nu}  d }{\theta} \right),
\end{eqnarray}
where $d$ is the Euclidean distance between two locations, $K_{\nu}$ is the modified Bessel function of the second order with smoothness parameter $\nu$, and $\theta$ is the spatial range. We allow  $\X = {\X}_c + {\X}_u$, where $\textrm{Cov} \left(  \X \right) = \sigma_c^2 \Rtwo{c}  + \sigma_u^2 \Rtwo{u} $, $\textrm{Cov} \left(  \Z \right) = \sigma_z^2 \Rtwo{c}  $, and $\textrm{Cov} \left( \X,  \Z \right) = \rho \sigma_c \sigma_z \Rtwo{c}$. We assume that $\Rtwo{c}$ and $\Rtwo{u}$ are each a member of  \eqref{maternclass} with the same $\nu$ and potentially different spatial range parameters. We stress that the source of confounding here is $\rho$, and the spatial aspect of the confounding is the shared spatial correlation functions in $\Rtwo{c}$ and $\Rtwo{u}$. Although the processes are correlated and share a common spatial dependence functions, there is no guarantee that a particular set of realizations  $\x$ and $\z$ will be correlated or share specific spatial patterns.  

\subsubsection{Bias from a Data Generation Perspective} \label{subsubsec:datagen}
We first explore bias from the perspective of data generation spatial confounding. Work on data generation confounding tends to treat $\X$ and $\Z$ stochastically when deriving bias terms \citep{Pac_spatialconf,page2017estimation}. Generalized least squares estimators are used when considering bias for a spatial regression analysis model. We adopt this approach here. 

In \cref{stoch_ols_bias} we calculate the bias terms  $\bias{\betaxp{NS} | \Xone }= \beta_x - \textrm{E} \left( \hat{\beta}^{NS} | \Xone \right),$ for a non-spatial regression analysis model, and $\bias{\betaxp{S} | \Xone }= \beta_x - \textrm{E} \left( \hat{\beta}^{S} | \Xone \right),$ for a spatial regression analysis model of the form \eqref{eq:genericSpatial}. Here, $\Xone = [\boldone ~ \X ]$.

\begin{rmk} \label{stoch_ols_bias}
 Let the data generating model be of the form \eqref{stochastic} with  $\X = {\X}_c + {\X}_u$ and $\Z$ having the following characteristics: 
 \begin{enumerate}
     \item $\textrm{Cov} \left(  \X \right) = \sigma_c^2 \Rtwo{c}  + \sigma_u^2 \Rtwo{u} $ 
     \item $\textrm{Cov} \left(  \Z \right) = \sigma_z^2 \Rtwo{c}  $, and 
     \item $\textrm{Cov} \left( \X,  \Z \right) = \rho \sigma_c \sigma_z \Rtwo{c}$
 \end{enumerate}
where $\Rtwo{c}$ and $\Rtwo{u}$ are of the form \eqref{maternclass} with the same $\nu$. If a non-spatial analysis model of the form \eqref{eq:OLSmodel} is employed with variance parameters assumed known, then the $\bias{\beta_X^{NS} | \Xone }= \beta_x - \textrm{E} \left( \hat{\beta}_X^{NS} | \Xone \right)$ can be expressed as: 
 \begin{eqnarray} \label{olsstochbias}
 	\beta_z \rho  \frac{\sigma_z}{\sigma_c} \left[ \left( {\Xone}^T \Xone \right)^{-1} {\Xone}^T \bm{K} \left( \X - \mu_x \boldone \right) \right]_2
\end{eqnarray}
If instead, a spatial analysis model of the form \eqref{eq:genericSpatial} is employed with variance parameters assumed known, then $\bias{\beta_X^S | \Xone }= \beta_x - \textrm{E} \left( \hat{\beta}_X^{S} | \Xone \right)$ can be expressed as: 
\begin{eqnarray} \label{glsstochbias}
 \beta_z \rho \frac{\sigma_z}{\sigma_c} \left[ \left( {\Xone}^T \bm{\Sigma}^{-1} \Xone \right)^{-1} {\Xone}^T \bm{\Sigma}^{-1} \bm{K} \left( \X - \mu_x \boldone \right) \right]_2
\end{eqnarray}
\noindent where $\bm{K}= p_c \left(p_c \bm{I}  + (1-p_c) \Rtwo{u} \Rtwo{c}^{-1} \right)^{-1} $, $p_c=\frac{\sigma_c^2}{\sigma_c^2 + \sigma_u^2}$, $\bm{\Sigma} = \beta_z^2 \sigma_z^2 \Rtwo{c} + \sigma^2 \bm{I}$, and $[ ]_2$ indicates the second element of the vector.
 \end{rmk}
 \begin{proof}
See \cref{app:stoch_ols_bias} and \cref{app:stoch_gls_bias} for derivations. 
\end{proof}
We note that \eqref{glsstochbias} is equivalent to Equation (6) in \citet{Pac_spatialconf} when $\nu=2$. The bias terms \eqref{olsstochbias} and \eqref{glsstochbias} are very complicated. We take a moment to point out several things. First, for the spatial model the ``true'' precision of $\Y$ (conditional on $\X$), $\bm{\Sigma}^{-1}$, is used, effectively ignoring the impact of the particular analysis model chosen. As we have discussed, this is very common in explorations of bias influenced by the data generation spatial confounding. However, we note that \citet{Pac_spatialconf} did include a brief description of the impact of analysis models in Section 2.1 of that paper. Second, it is difficult to derive insights from these forms of bias. They are heavily dependent not only on the spatial range parameters and the various other variance parameters, but also on the distributional assumptions on $\X$.

In \citet{Pac_spatialconf}, the author measured the bias due to spatial confounding with the term $c_S(\X)= \left( {\Xone}^T \bm{\Sigma}^{-1} \Xone \right)^{-1} {\Xone}^T \bm{\Sigma}^{-1} \bm{K} \left( \X - \mu_x \boldone \right)$ from \cref{glsstochbias}. Here, we also introduce the non-spatial equivalent  $c_{NS}(\X)= \left( {\Xone}^T  \Xone \right)^{-1} {\Xone}^T \bm{K} \left( \X - \mu_x \boldone \right)$ from \cref{olsstochbias}. More specifically, \citeauthor{paciorek2007computational} considered $\textrm{E}_{\X} \left( c_S(\X) \right)$. To control for the influence for the marginal variance parameters and $\beta_z$, \citeauthor{paciorek2007computational} calculated (via simulations) $\textrm{E}_{\X} \left( c_S(\X) \right)$ for various values of $p_c$ (defined in \cref{stoch_ols_bias}) and the term $p_z= \frac{\beta_z^2 \sigma_z^2}{\beta_z^2 \sigma_z^2 + \sigma^2}$. He did this for the case where $\Rtwo{c} $ and $\Rtwo{u}$ are members of \eqref{maternclass} with $\nu =2$. We reconsider this same simulation setup, but provide results for both $\nu=2$ and $\nu=0.5$.  By considering this measure of bias for both parameter values, we get a better idea of how bias can change depending on the assumptions of the data generation processes.  

The results, provided in Figure \ref{fig:data_illustrationa} and produced using his code available at \url{https://www.stat.berkeley.edu/~paciorek/research/code/code.html}, suggest that a spatial regression analysis model {could} result in reduced bias relative to a non-spatial analysis when $\theta_u \ll \theta_c$. Recall that $\theta_u$ is the range parameter for the portion of $\X$ process that is independent of $\Z$ (the unconfounded component of $\X$), while $\theta_c$ is the range parameter for the process contributing to both $\X$ and $\Z$ (the confounded component of $\X$).  Thus, the results suggest the spatial regression analysis {may} reduce bias when the spatial range for the confounded component of $\X$ is larger than the range for the unconfounded component.  

To examine how the bias changes for the spatial model for different parameter values,  consider $\textrm{E}_{\X} \left( c_{S}(\X) \right)$ as represented in Figure \ref{fig:data_illustrationa}.  This figure provides images of $\textrm{E}_{\X} \left( c_{S}(\X) \right)$ for 100 locations on a grid of the unit square for different fixed values of $p_c$, $p_z$, $\theta_c$, and $\theta_u$ when $\nu = 2$ (a) and $\nu =0.5$ (b).  Consider (a), the case when $\nu=2$. The upper left subplot of the image matrices provide a colored image of $\textrm{E}_{\X} \left( c_{S}(\X) \right)$ when $p_c = p_z = 0.1$ and $\theta_c$ varies from 0 to 1 (x-axis) and $\theta_u$ varies from 0 to 1 (y-axis).  As $\theta_c$ increases, holding all else constant,  $\textrm{E}_{\X} \left( c_{S}(\X) \right)$ decreases.  In contrast, as $\theta_u$ increases, holding all else constant, $\textrm{E}_{\X} \left( c_{S}(\X) \right)$ increases.  Moving to the other subplots within the image matrix of (a) shows the same colored representation of $\textrm{E}_{\X} \left( c_{S}(\X) \right)$, but for different values of $p_c$ and $p_z$. As either $p_c$ or $p_z$ increase, $\textrm{E}_{\X} \left( c_{S}(\X) \right)$ also increases. Notice also that for any given value of $p_c$ and $p_z$, we see the same behavior for $\textrm{E}_{\X} \left( c_{S}(\X) \right)$ as $\theta_u$ and $\theta_c$ changes that we saw in the first subplot considered.  Namely, that the bias for the spatial model will decrease as either $\theta_u$ decreases or $\theta_c$ increases. 

While it is interesting to see how the bias changes for the spatial model, we are most interested in how it changes relative to the bias for the non-spatial model.  To make this comparison, note that, as \citet{Pac_spatialconf} pointed out, $\textrm{E}_{\X} \left( c_{NS}(\X) \right) \approx p_c $; thus, the bias term for the non spatial analysis model is represented by $p_c$. It can also be shown that  $\textrm{E}_{\X} \left( c_{S}(\X) \right) \approx p_c $ when $\theta_c = \theta_u$.  The diagonal of each subplot in the image matrix represents the case when $\theta_c = \theta_u$, or when both $\textrm{E}_{\X} \left( c_{NS}(\X) \right) \approx \textrm{E}_{\X} \left( c_{S}(\X) \right) \approx p_c$. Below this diagonal, $\theta_u < \theta_c$ and the color shows that $\textrm{E}_{\X} \left( c_{S}(\X) \right) < p_c$.  In other words, as represented by the blue coloring below the diagonal, the spatial model has reduced bias relative to the non-spatial model when $\theta_u \ll \theta_c$.


\begin{figure}
\centering
\begin{minipage}{.7\textwidth}
\includegraphics[width=\textwidth]{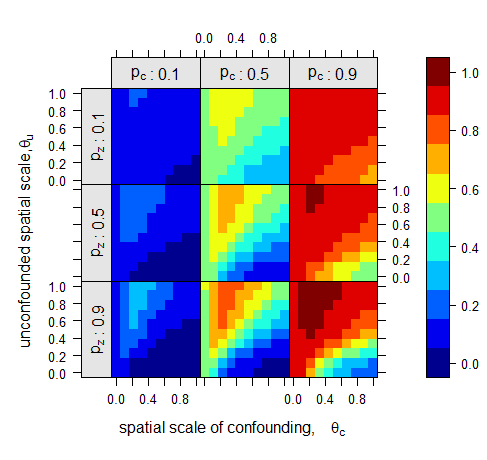}\\
\centering (a)\\
\includegraphics[width=\textwidth]{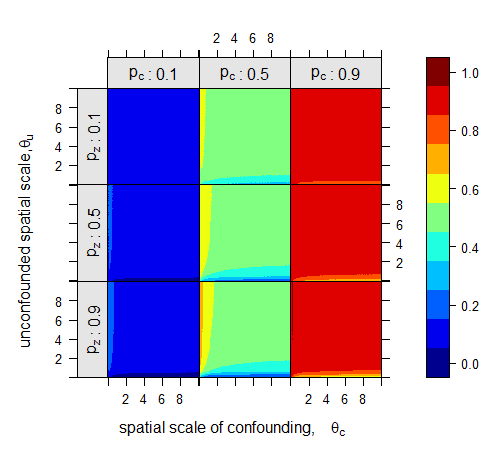}\\
(b)
\end{minipage}
\captionof{figure}{This image depicts $\textrm{E}_{\X} \left( c_{S}(\X) \right)$ for 100 locations on a grid of the unit square when $\Rtwo{u}$ and $\Rtwo{c}$ belong to the Mat\'ern class with $\nu=2$ (a) and $\nu=0.5$ (b). This image was created from Christopher Paciorek's code using the \texttt{fields} and \texttt{lattice} packages \citep{lattice,fields}.}
\label{fig:data_illustrationa}
\end{figure}


However, a spatial regression analysis could also increase bias relative to a non-spatial analysis when $\theta_u \gg \theta_c$.  This is represented in Figure \ref{fig:data_illustrationa}(a) in that, above this diagonal line $\theta_u > \theta_c$ and the color shows that $\textrm{E}_{\X} \left( c_{S}(\X) \right) > p_c$.  \citet{Pac_spatialconf} explicitly acknowledged that the case where $\theta_u \gg \theta_c$ is likely of limited interest in real applications. However, this case has increasingly influenced further research in spatial confounding. Or rather, the fact that bias for a spatial analysis model {can} be increased relative to the bias for a non-spatial model has influenced further research. Other papers often use this observation to support statements suggesting that spatial confounding occurs when the ``spatial range of the observed risk factors is larger than the unobserved counterpart'' \citep{marques2022mitigating}. However, it is rarely acknowledged that these simulations considered only a very specific case \citep[there are exceptions; see e.g.,][]{keller2020selecting}. 

Further adding to the complexity of this bias comparison, it turns out that the behavior of bias from spatial confounding is dependent on the distributional assumptions for $\X$. To illustrate this issue, we now repeat the simulation study for the case when $\Rtwo{c} $ and $\Rtwo{u}$ are members of \eqref{maternclass} with $\nu =0.5$ as shown in Figure \ref{fig:data_illustrationa}(b).  Here, the spatial process is less smooth than when $\nu=2$.  In contrast to the plots provided in (a) where we considered values of  $\theta_c$ and $\theta_u$ up to 1, when $\nu=0.5$ the images can look fairly flat for these same values. Thus, to better illustrate the behavior, we consider values of $\theta_u$ and $\theta_c$ up to 10.  Again, in \cref{fig:data_illustrationa}(b), $\textrm{E}_{\X} \left( c_{NS}(\X) \right) \approx p_c $, and  when $\theta_c = \theta_u$, $\textrm{E}_{\X} \left( c_{S}(\X) \right) \approx p_c $.  In this figure we see the bias modification term is almost always equal to the non-spatial equivalent; namely, the color shown in the images representing $\textrm{E}_{\X} \left( c_{S}(\X) \right)$ match the value of $p_c \approx \textrm{E}_{\X} \left( c_{NS}(\X) \right)$.  This figure also shows that bias reduction can occur when $\theta_u$ is less than two, regardless of the value of $\theta_c$. Similarly, bias can be increased when $\theta_c$ is less than two, across all values of $\theta_u$. In other words, there is no longer strong evidence to support statements that spatial confounding impacts bias when the ``spatial range of the observed risk factors is larger than the unobserved counterpart.''


We note that these numerical analyses do not explore the behavior of the variance of the estimates, which will also be different for each model.  Understanding the behavior of the biases themselves is quite nuanced, and the variances even more so. We are able to explore a hint of the differences in variances in \cref{sec:simstudies}, but better understanding of such uncertainty is needed future work.

Importantly, however, these examples illustrate how sensitive our conclusions about the impact of spatial confounding are to the distributional assumptions we make about $\X$ and $\Z$.  Bias from a data generation perspective depends on the assumptions about the generating spatial processes, including, but not limited to, the choice of spatial dependence structure (e.g., the selection of $\nu$).

\subsubsection{Bias from an Analysis Model Perspective}\label{subsubsec:analysis}
In this sub-section, we focus on the analysis model type of spatial confounding. To make our results comparable to the setting explored in \cref{subsubsec:datagen}, we assume that for a particular set of realizations $\x$ and $\z$, the response $\bm{y}$ is generated from a model of the form \cref{eq:model_0}.  We can assume that the processes $\X$ and $\Z$ are generated as before. However, the results in this section do not depend on any distributional assumptions about $\X$ and $\Z$.  Unlike in \cref{subsubsec:datagen}, we assume that all variance parameters are unknown. As we will see, this results in conceptualizing spatial confounding by the relationships that the $\x$, $\y$, and $\z$ have with the eigenvectors of an estimated precision matrix $\Sigmah^{-1}$.

We consider both the non-spatial analysis model and the class of spatial analysis models of the form \eqref{eq:genericSpatial}. First, in \cref{lemma:olsfixed}, we derive the bias term that results from fitting a non-spatial analysis model. 
\begin{lemma} \label{lemma:olsfixed} 
Let the data generating model be of the form \eqref{eq:model_0} with $\x$ known. If a non-spatial analysis model of the form \eqref{eq:OLSmodel} is fit, then $\bias{\hat{\beta}_{x}^{NS}} = \beta_x - E\left(\hat{\beta}_{x}^{NS} \right)$ can be expressed as:
	\begin{eqnarray*} \label{eq:obsolsfixed}
	\frac{\beta_z}{	 \euclnorm{\boldone}^2  \euclnorm{\x}^2   -  \left[ \euclmetric{\x}{\boldone}  \right]^2  } \left(  \euclnorm{\boldone}^2  \euclmetric{\x}{\z} - 	\euclmetric{\x}{\boldone} \euclmetric{\z}{\boldone}\right),
 \end{eqnarray*}
where $ \langle \cdot, \cdot \rangle $ is the standard Euclidean inner product  and $||\cdot||$ represents the norm induced by it. \end{lemma}
Because this ends up being a special case of the bias for the GLS estimators discussed next, we delay a discussion of these terms. Now, we assume that we fit a spatial analysis model of the form \eqref{eq:genericSpatial}. 
\begin{lemma} \label{lemma:glsfixed}
Let the data generating model be of the form \cref{eq:model_0} with $\x$ known. If a spatial analysis model of the form \eqref{eq:genericSpatial} is fit, and results in the positive definite estimate $\Sigmah$, then the bias $\bias{\hat{\beta}_{x}^{S} }= \beta_x - E\left(\hat{\beta}_{x}^{S}\right)$ can be expressed as:
\begin{eqnarray}  \label{eq:obsglsfixed}
\frac{\beta_z}{	 \sigmainvnorm{\boldone}^2  \sigmainvnorm{\x}^2   -  \left[ \sigmainvdot{\x}{\boldone} \right]^2  } \left(   \sigmainvnorm{\boldone}^2  \sigmainvdot{\x}{\z} - 	\sigmainvdot{\x}{\boldone} \sigmainvdot{\z}{\boldone} \right).
\end{eqnarray}
\end{lemma}
\begin{proof}
See \cref{app:obsglsfixed} for the calculations.
\end{proof}
Here, the estimate of precision matrix $\bm{\Sigma}^{-1}$ is $\Sigmah^{-1}$. We define the inner product $\langle m, n \rangle_{\bm{\Sigmah}^{-1}} = m^T \bm{\Sigmah}^{-1} n$ for $m,n \in \mathbb{R}^n$, and we let $|| \cdot ||_{\bm{\Sigmah^{-1}}}$  be the norm induced by it (see \cref{appa_diff_geometry} for more details). We do not make any assumptions of how the term $\Sigmah^{-1}$ is estimated (e.g., Bayesian vs. restricted maximum likelihood), but we acknowledge two different methods of fitting the same analysis model could result in different $\Sigmah^{-1}$. Finally, we note in \cref{ols_special_gls} that the bias term in \cref{eq:obsolsfixed} is a special case of the bias term in \cref{eq:obsglsfixed}. 

\begin{rmk} \label{ols_special_gls}
When $\Sigmahinv = \bm{I}$, $\bias{\hat{\beta}_{x}^{NS}}$ in \cref{eq:obsolsfixed} is a special case of $\bias{\hat{\beta}_{x}^{S} }$ in \cref{eq:obsglsfixed}.
\end{rmk}

The bias term in \eqref{eq:obsglsfixed} is a function of $\beta_z$, which makes intuitive sense. Although this will, of course, not be known, we note that its impact on inference is the same across all analysis models belonging to the spatial analysis models as well as the non-spatial analysis model. For the moment, we focus on the other terms. We begin with the numerator of \eqref{eq:obsglsfixed} (ignoring $\beta_z$):
\begin{eqnarray} \label{eq:numobsglsfixed}
  \sigmainvnorm{\boldone}^2  \sigmainvdot{\x}{\z} - 	\sigmainvdot{\x}{\boldone} \sigmainvdot{\z}{\boldone}. 
\end{eqnarray}
Broadly speaking, this term tends to get smaller when one of two things happen. The first situation occurs when the low frequency eigenvectors of $\Sigmahinv$ are ``flat.'' We say an eigenvector is ``flat'' if  there is a small angle (with respect to the Euclidean norm) between it and the column vector of ones, $\boldone$. In other words, saying an eigenvector is ``flat'' means it is roughly constant.  We say an eigenvector is ``low-frequency'' if its associated eigenvalue is less than 1. Recall that when we plot a low-frequency eigenvector, the values are smooth and slow changing across the space. When this occurs, all terms involving $\boldone$ (i.e., $ \sigmainvnorm{\boldone}^2$, $\sigmainvdot{\x}{\boldone}$, and $\sigmainvdot{\z}{\boldone}$) will become smaller in magnitude. To illustrate this, we randomly generate locations for 140 observations on a $[0,10] \times [0,10]$ window. Using these locations, we represent different potential $\Sigmahinv$'s by calculating the inverses of variance-covariance matrices for members of the Mat\'ern class. In \cref{fig:bias_illustration}, we use colors to denote three possible values of $\nu$: $\nu = 0.5$ (the exponential), $\nu =1$ (Whittle), and $\nu =2$. For fixed $\nu$, we then calculate various variance-covariance matrices by allowing $\theta$ to vary. For the inverse of each unique matrix, we calculate $\sigmainvnorm{\boldone}$. For fixed $\nu$, we can expect the lowest frequency eigenvectors of the eigendecomposition of the associated $\Sigmahinv$ to become flatter as $\theta$ increases. In \cref{fig:bias_illustration}(a), we can see that $\sigmainvnorm{\boldone}$ decreases in magnitude as $\theta$ increases for all values of $\nu$. This trend will also be seen in cross-products involving $\boldone$, as can be seen in \cref{fig:bias_illustration}(b). Note that in these plots, the black line denotes the Euclidean norm. Almost all of the values of $ \sigmainvnorm{\boldone}$ are less than this in magnitude.  In many practical situations where spatial covariance matrices are employed, this will tend to happen.

The second situation that will tend to decrease the magnitude of \eqref{eq:numobsglsfixed} occurs when there are small angles (again with respect to the Euclidean norm) between either $\x$ or $\z$  and low-frequency eigenvectors of $\Sigmahinv$. When this occurs, we say that $\x$ (or $\z$) is spatially smooth with respect to $\Sigmahinv$. Recall, for us, low frequency eigenvectors are those with associated eigenvalues less than 1.  We note that \eqref{eq:numobsglsfixed} is symmetric in $\x$ and $\z$. As just one of these variables becomes more correlated with a low frequency eigenvector, all terms involving it  will tend to decrease in magnitude. Both variables being correlated with low frequency eigenvectors will tend to be associated with a further reduction in the magnitude of the bias. As an illustration, we again use the 140 locations just discussed. We generate a realization $\x$ from an exponential process (\eqref{maternclass} with $\nu=0.5$) at these locations with $\theta=10$. In \cref{fig:bias_illustration}(c), we illustrate how this realization appears spatially smooth. Because $\x$ is spatially smooth, it will often be correlated with low frequency eigenvectors. Unsurprisingly, as shown in \cref{fig:bias_illustration}(d), $\sigmainvnorm{\x}$ is always smaller than the corresponding Euclidean norm. We note that if either $\x$ or $\z$ is linearly dependent with $\boldone$, i.e., constant, then the bias term in \eqref{eq:numobsglsfixed} will be 0. Thus, the flatter $\x$ and $\z$ become, the smaller the bias. 

\begin{center}
\begin{minipage}{0.45\textwidth}
\includegraphics[width=\linewidth]{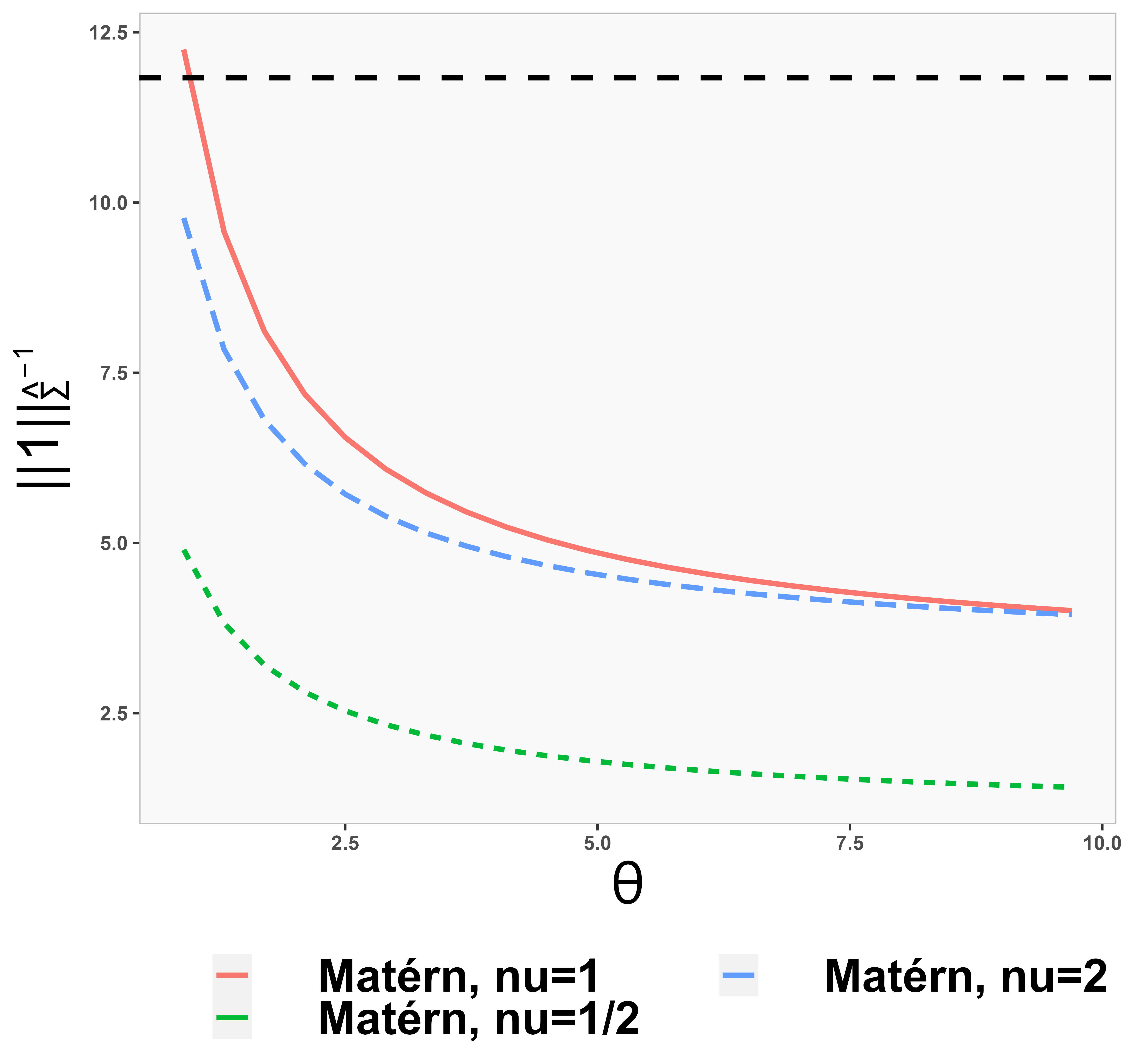}
\captionof*{figure}{a) $\sigmainvnorm{\boldone}$}
\end{minipage}\hfill
\begin{minipage}{0.45\textwidth}
\includegraphics[width=\linewidth]{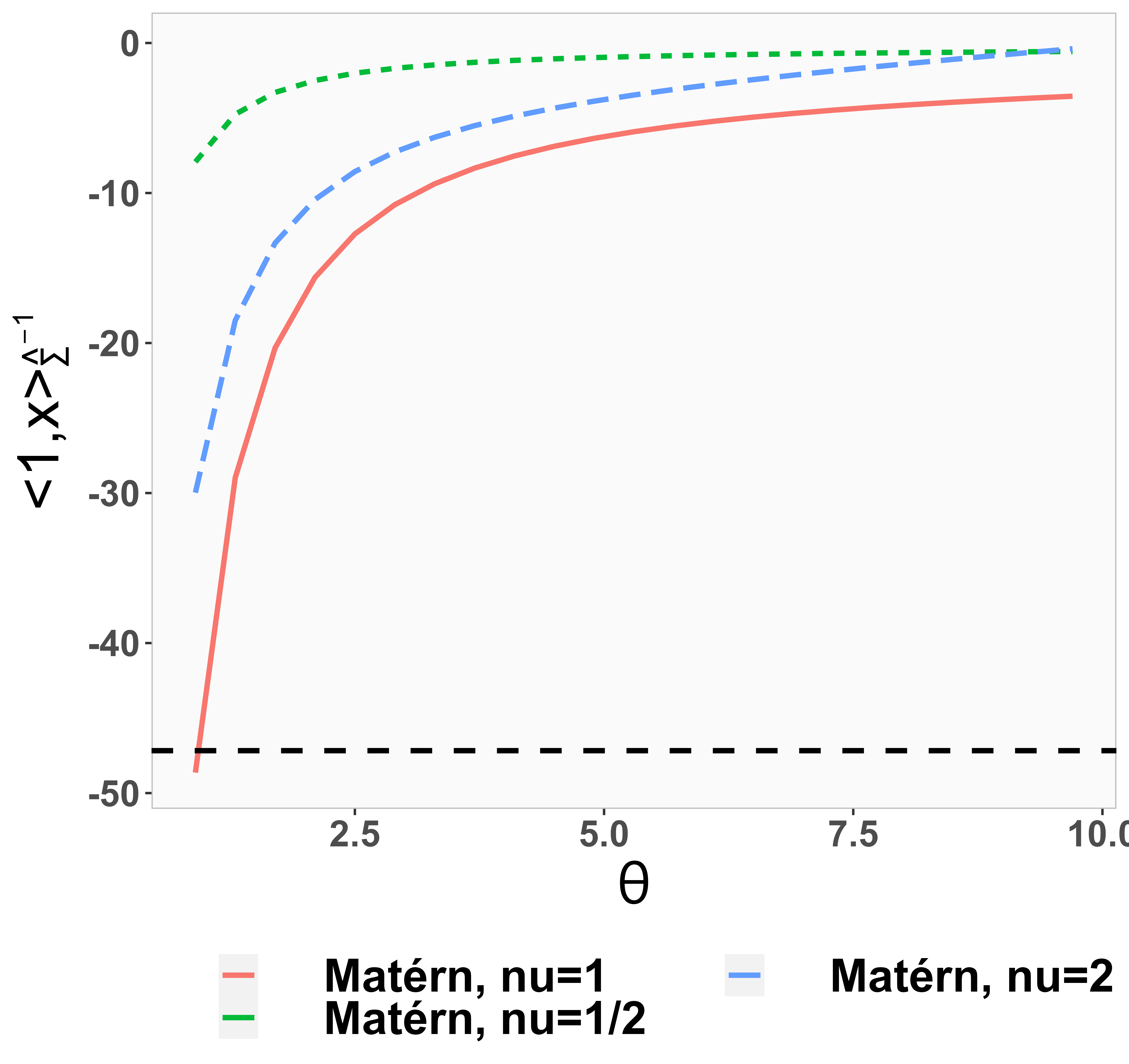} 
\captionof*{figure}{b) $\sigmainvdot{\x}{\boldone}$  } 
\end{minipage}
\begin{minipage}{0.45\textwidth}
\includegraphics[width=\linewidth]{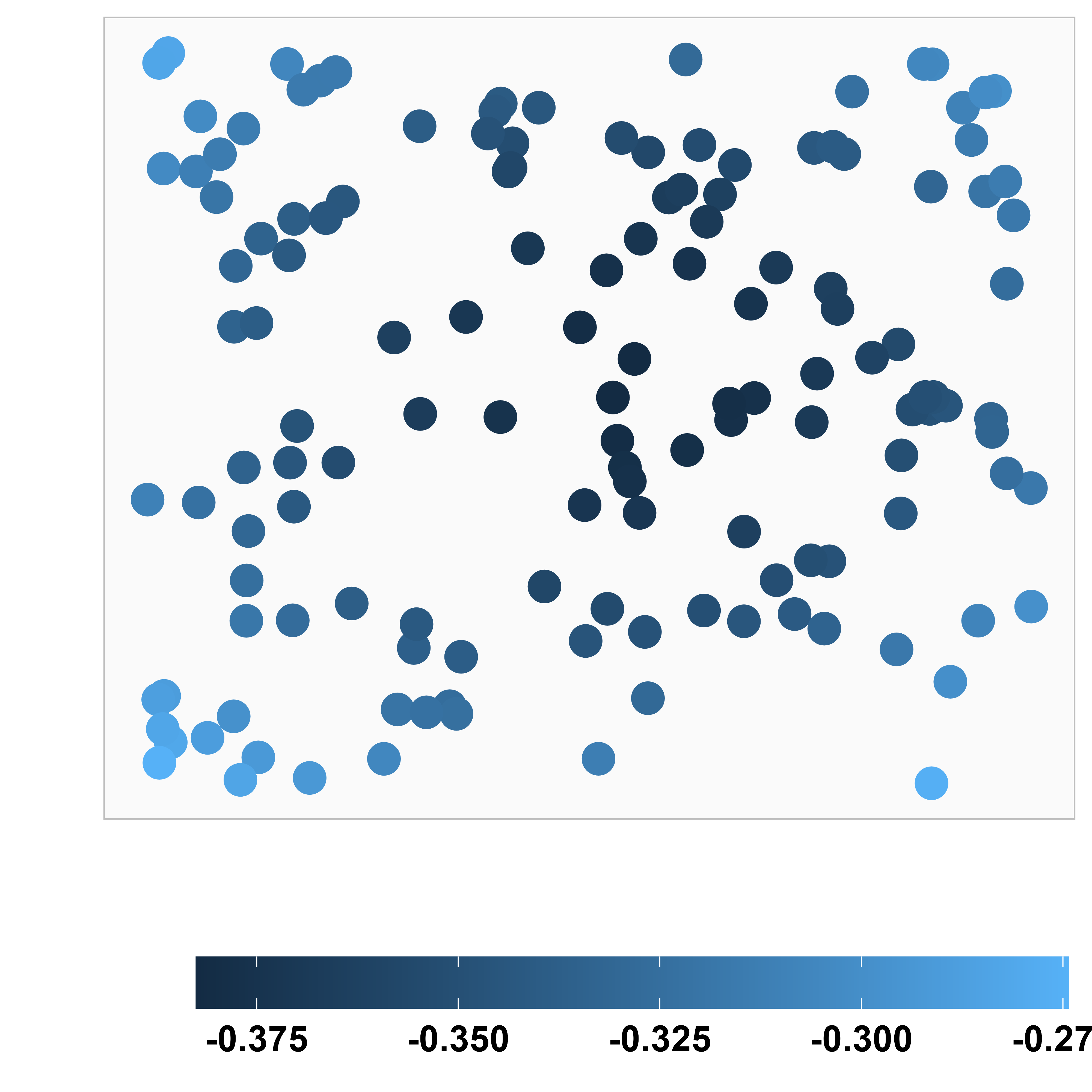}
\captionof*{figure}{c) Illustration of $\x$}
\end{minipage} \hfill
\begin{minipage}{0.45\textwidth}
\includegraphics[width=\linewidth]{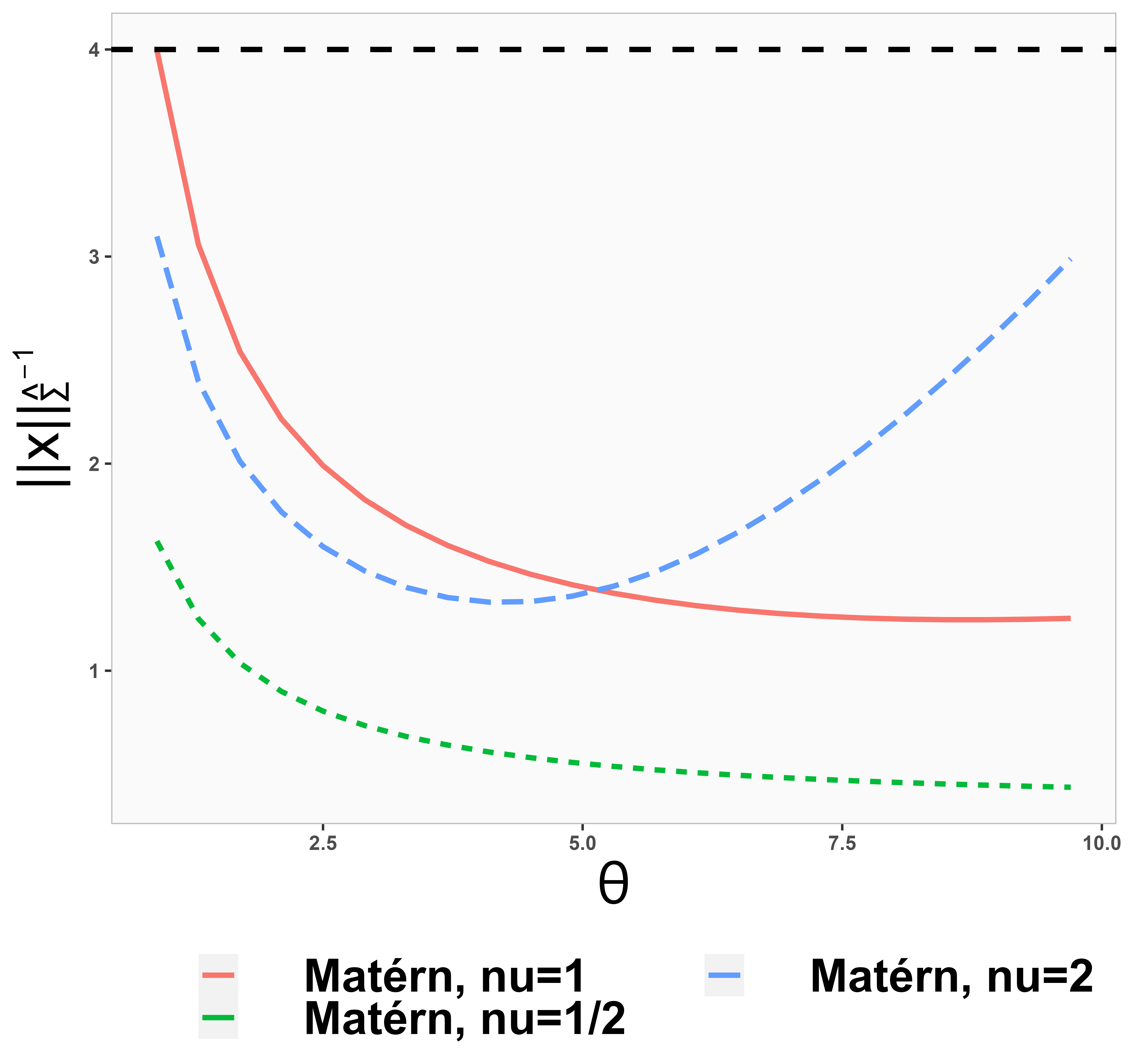}
\captionof*{figure}{ d)$\sigmainvnorm{\x}$}
\end{minipage}
\captionof{figure}{Illustrations of components of \eqref{eq:numobsglsfixed}. All plots were made with the \texttt{ggplot2} package \citet{ggplot2}. }
\label{fig:bias_illustration}
\end{center}

The behavior of \eqref{eq:numobsglsfixed} supports the traditional view that fitting a spatial analysis model helps improve inference on $\beta_x$. When the low-frequency eigenvectors of $\Sigmahinv$ mirror the patterns of either $\x$ or $\z$, fitting a spatial analysis model will tend to result in better estimates of $\beta_x$ than a non-spatial model. It also highlights that what it means to be ``spatially smooth'' for the purposes of bias reduction depends on the analysis model chosen. To see this, note that in \cref{fig:bias_illustration} (a), (b), and (d), the magnitudes can be quite different for different choices of $\Sigmah^{-1}$, particularly when $\theta$ is small.

Recall from our discussion in \cref{subsec:datagen}, researchers are sometimes concerned with collinearity between $\x$ and $\z$ as a possible source of confounding bias. We note that when $\z = \alpha \x$, for $\alpha \neq 0$, then \eqref{eq:numobsglsfixed} is always less than or equal to $\alpha \sigmainvnorm{\boldone}^2 \sigmainvnorm{\x}^2$. If $\x$ is correlated with low-frequency eigenvectors or the low-frequency eigenvectors are flat, this term will typically be smaller for a spatial analysis model than for a non-spatial analysis model. In other words, this suggests spatial analysis models can still reduce bias relative to a non-spatial analysis model when $\x$ or $\z$ are collinear so long as at least one of them is spatially smooth.

We now turn our attention to the denominator of \eqref{glsstochbias}:
\begin{eqnarray*}  \label{eq:demobsglsfixed}
\sigmainvnorm{\boldone}^2  \sigmainvnorm{\x}^2   -  \left[ \sigmainvdot{\x}{\boldone} \right]^2 = \\
 \sigmainvnorm{\boldone}^2  \sigmainvnorm{\x}^2 \sin{\sigmainvangle{\x}{\boldone}}^2,
\end{eqnarray*}
where $\sigmainvangle{\x}{\boldone}$ is the angle between $\x$ and $\boldone$ with respect to the Riemannian metric induced by $\Sigmahinv$ (see \cref{appa_diff_geometry} for more details). The term $\sin{\sigmainvangle{\x}{\boldone}}$ will be minimized when $\x$ is linearly dependent with $\boldone$ (i.e., constant), and it will be maximized when $\x$ is perpendicular (with respect to the Riemannian metric induced by $\Sigmahinv$) to $\boldone$. In other words, because we are considering the denominator of the bias, the flatter $\x$ becomes the larger the bias. This behavior supports the insights from research into the analysis model source of spatial confounding: $\x$ which are ``too'' spatially smooth \textit{can} distort inference on $\beta_x$. 

Pulling these insights together, we see that bias \emph{can} decrease with a spatial analysis model in settings where $\x$, $\z$ are spatially smooth or cases in which low-frequency eigenvectors of $\Sigmahinv$ are flat. We emphasize again that what it means to be spatially smooth depends on the correlation of $\x$ and $\z$ with the low-frequency eigenvectors of $\Sigmahinv$.  Notably, for cases when $\x$ is not only spatially smooth, but flat, the numerator and denominator of \eqref{glsstochbias} can work in opposite directions. At the extreme, when $\x$ is collinear with $\boldone$, the bias will be 0 (where we use the mathematical convention that $\frac{0}{0} =0$). However, as $\x$ becomes flatter, it's possible that the denominator will shrink faster than the numerator in some settings. In this case, the flatness of $\x$ can effectively serve to increase the bias. This reinforces the observations made by researchers influenced by analysis model spatial confounding. Finally, we note that for the case of collinearity between $\x$ and $\z$ (i.e., returning to $\z = \alpha \x$, $\alpha \neq 0$), the overall bias term is $\alpha$ for both spatial analysis models and non-spatial analysis models. This suggests, contrary to some research in data generation spatial confounding, that bias induced by collinearity between $\x$ and $\z$ is not exacerbated by fitting a spatial analysis model.

Again, the selection of the dependence structure will play a key role in the amount of bias.  When fitting a spatial model with a Mat\'{e}rn covariance structure, it is common to choose a fixed value of the smoothing parameter, $\nu$, and estimate other correlation function parameters.  The same principle applies to a GMRF, where a neighborhood structure is chosen and fixed. These choices can have large impacts on the amount of bias for the same data set.

\subsection{Bias: Adjusted Spatial Analysis Models} \label{subsec:biasadj}
In this section, we consider the impact of the analysis model on inference of $\beta_x$ for the gSEM and Spatial+ approaches referenced in \cref{subsec:alleviating}. Recall, these models were developed to improve inference on $\beta_x$ when certain assumptions about the data generation process are assumed to be true. For the gSEM and Spatial+ methods, these assumptions include that $\x = \beta_{z:x} \z + \epsilon_x$.   When this is the case, the data generation source of spatial confounding suggests that fitting gSEM or Spatial+ will reduce bias relative to a spatial analysis model or non-spatial analysis model.  (Specifically, from \cref{stoch_ols_bias}, $\rho = 1$, $\theta_u$ is essentially 0, and $\sigma_z^2$ and $\sigma^2_c$ are functions of $\beta_{z:x}$.)  Thus, we focus here on exploring the bias from an analysis model perspective.

We take a moment to give details on both approaches. The gSEM approach, summarized in \cref{gsem}, is equivalent to replacing $\y$ and $\x$ with $\ry$ and $\rb$ \citep{Thaden}. This latter set of variables are defined to be the residuals, respectively, from  spatial analysis models using $\y$ and $\x$ as the response variable \citep[the examples in ][ use no covariates]{Thaden}. These residuals are then used to fit a non-spatial analysis model of the form \eqref{eq:OLSmodel}, and inference for $\beta_x$ is based on that estimate.  \citet{dupont2022spatial+} used the approach for gSEM described in \cref{gsem} and found that the gSEM approach improved inference only when smoothing was used in Steps 1 and 2, and we adopt this convention from hereon out.  The Spatial+ approach, summarized in \cref{spatialplus}, involves replacing $\x$ with $\rb$. The analysis model used for inference on $\beta_x$ is then a spatial analysis model with response $\y$ and covariate $\rb$.

\begin{analysis}[gSEM] \label{gsem}
The gSEM approach can be summarized as follows:
\begin{enumerate}
    \item Define $\rb$ to be the residuals from the fitted values of a spatial analysis model with $\x$ as the response variable.
    \item Define $\ry$ to be the residuals from the fitted values of a spatial analysis model with $\y$ as the response variable.
    \item Fit an analysis model of the form \eqref{eq:OLSmodel} with response $\ry$ and covariate $\rb$.
\end{enumerate}
\end{analysis}

\begin{analysis}[Spatial+] \label{spatialplus}
The Spatial+ approach can be summarized as follows:
\begin{enumerate}
    \item Define $\rb$ to be the residuals from the fitted values of a spatial analysis model with $\x$ as the response.
    \item Fit an analysis model of the form \eqref{eq:genericadjSpatial} with response $\y$ and covariate $\rb$.
\end{enumerate}
\end{analysis}

The spatial analysis model used in \cite{Thaden} is a repeated measures model with a random effect for each areal region, while the spatial analysis model used in \cite{dupont2022spatial+} is a generalized additive model with penalized thin plate splines across the space.  There are many ways to fit a spatial analysis model to obtain residuals and it makes sense that different choices will result in different estimates of $\beta_x$.  Each unique data set may call for a different approach to spatial modeling.  For example, \cite{Cressie93} provides discussions on how different data sets and analysis goals would result in different modeling choices for the spatial dependence: through the mean or ``large-scale variation'' versus the stocahstic-dependence structure or ``small-scale variation.''    Thus, there are several ways we could adapt the gSEM and Spatial+ approaches to the spatial linear mixed model setting considered here and the biases of the corresponding estimates of $\beta_x$ will naturally be different for different choices.  Here we provide a big-picture exploration of the bias for a more general form of residuals relying on insights from \cref{subsubsec:analysis} to understand bias when the final model is of the form \eqref{eq:genericadjSpatial}.  In this section, we focus on the bias of the Spatial+ approach, which was developed for data scenarios with a single observation per location as considered throughout this work.  However, we do provide a similar equation for the bias for a general gSEM approach in \cref{app:obsgsemfixed}.

\begin{thm} \label{thm:spatialplus_bias}
Let the data generating model be of the form \cref{eq:model_0} with $\x$ known. Let $\rb$ be the residuals from a generic spatial analysis model with response variable $\x$ and fitted values $\hat{\x} = \bm{S}_x\x$.

If the final analysis model is an adjusted spatial analysis model of the form \eqref{eq:genericadjSpatial} with $\tilde{\x} = \rb$, $\tilde{\y} = \y$, and results in the estimate $\Sigmah$, then the bias $\bias{\hat{\beta}_{x}^{S+} }= E\left(\hat{\beta}_{x}^{S+}\right) - \beta_x$ can be expressed as:
\begin{eqnarray}  \label{eq:obsadjglsfixed}
A_2^*(\rb,\x) + B_2(\rb,\x),
\end{eqnarray}
where
\begin{align*}
 A_2^*(\rb,\x) &=\beta_x \left(\frac{ \left( \sigmainvnorm{\boldone}^2 \sigmainvdot{\rb}{\x} - 	\sigmainvdot{\rb}{\boldone}	\sigmainvdot{\x}{\boldone} \right)}{ \sigmainvnorm{\boldone}^2 \sigmainvnorm{\rb}^2 -  \left[ \sigmainvdot{\rb}{\boldone} \right]^2} -1 \right), \\
 B_2(\rb,\x) &= \beta_z \left(\frac{ \left( \sigmainvnorm{\boldone}^2 \sigmainvdot{\rb}{\z} - 	\sigmainvdot{\rb}{\boldone}	\sigmainvdot{\z}{\boldone} \right) } { \sigmainvnorm{\boldone}^2 \sigmainvnorm{\rb}^2 -  \left[ \sigmainvdot{\rb}{\boldone} \right]^2}\right).
\end{align*}
\end{thm}
\begin{proof}
See \cref{app:obsadjglsfixed} for the proof.
\end{proof}

Recall from \cref{subsubsec:analysis} that the norms and inner product terms (e.g., $||\boldone||^2_{\bm \Sigma^{-1}}$ and $\sigmainvdot{\x}{\boldone}$) will tend to 0 as the corresponding vectors are similar to the low frequency eigenvectors of $\hat{\bm{\Sigma}}^{-1}$.  Consider this behavior with respect to $\sigmainvnorm{\rb}$.  The intent of obtaining $\rb$ is to remove the smooth spatial behavior from $\x$.  Thus, it is reasonable to assume that $\rb$ will \emph{not} be spatially smooth with respect to $\bm{\Sigma}^{-1}$ and thus, $\sigmainvnorm{\rb}$ will not tend to 0.  We explore this in Figure \ref{fig:rx_bias_illustration}.  Using the same values of $\x$ from Figure \ref{fig:bias_illustration}, we use penalized thin plate splines to predict $\hat{\x}$ and compute $\rb = \x - \hat{\x}$; these residuals are illustrated in \cref{fig:rx_bias_illustration}(c).  While some clustering is still apparent, the range is visibly smaller than the range of the original $\x$.  It is easy to see in \cref{fig:rx_bias_illustration}(a) that the smoother (i.e., $\nu=2$) and longer the range (i.e., $\theta$ large) of the spatial dependence, the larger the norm induced by $\Sigmahinv$ will be for the residuals.  This holds true for the other values of $\nu$, although it is harder to see in this plot.  This behavior is similar for any inner product that includes $\rb$ (see \cref{fig:rx_bias_illustration}(b) and (d)).  Because all of the norms and inner products in the denominator of both $A_2^*(\rb, \x)$ and $B_2(\rb,\x)$ include $\rb$, we expect that in contrast to the bias of the spatial model which could increase due to a denominator that decreases faster than the numerator, the bias from the Spatial+ approach will generally not increase as a result of a faster decreasing denominator.  In other words, the denominator for the Spatial+ model bias should be more ``well behaved.''

Now consider the numerator of $B_2(\rb, \x)$.  Just as in \cref{subsubsec:analysis}, when the low-frequency eigenvectors of $\bm{\Sigma}^{-1}$ are ``flat,'' any inner product with $\boldone$ will tend to 0.  Additionally, if $\z$ is spatially smooth with respect to $\bm{\Sigma}^{-1}$, any of the inner products that include $\z$ will tend to 0.  In these cases, the numerator will tend to 0 while the denominator will be well behaved so that $B_2(\rb, \x)$ tends to 0.

Moving to $A^*_2(\rb, \x)$ we find a bit more complicated of a relationship.  For $A_2(\rb, \x)$ to tend to 0, as is desired, we want the numerator of the fraction within this term to be as similar as possible to the denominator so that the fraction will tend to 1.  This will happen when $\rb$ and $\x$ are similar with respect to $\bm \Sigma^{-1}$.  To explore this, consider the fitted values of $\x$, $\hat{\x}$, used to obtain $\rb = \x - \hat{x}$.  Generally, $\hat{\x}$ will be a spatially-smoothed version of $\x$.  In one extreme, they are perfectly smooth so that $\hat{\x} = \alpha\boldone$.  In this case, the numerator will be equal to the denominator so that $A_2^*(\rb, \x)$ will be 0 (and in fact, the total bias will be the same as when the spatial model is fitted; see \cref{app:obsadjglsfixed} for proof details).  In the other extreme, the fitted values will perfectly match $\x$ and all terms involving $\rb$ will be 0, resulting in a bias due to $A_2^*(\rb, \x)$ of $\beta_x$.  In the thin plate spline scenario considered by \cite{dupont2022spatial+}, they found that the Spatial+ approach outperformed the spatial model in terms of bias when they used a penalization term.  By optimizing the penalization, the predicted values will be somewhere between the over-fitted terms and the overly-smoothed constant.  Thus, in practice, we expect $A^*_2(\rb, \x)$ to be somewhere between these two extremes of 0 and $\beta_x$.   

How $A^*_2(\rb, \x)$ and $B_2(\rb, \x)$ behave together will require more exploration and will vary depending on the chosen analysis models for every step.  Thus, we more fully examine the estimation error in the simulation studies in \cref{sec:simstudies}.  However, notice that the range of the norms and inner products in \cref{fig:rx_bias_illustration} are all smaller than the norms and products in \cref{fig:bias_illustration}.  Thus, while the bias from the adjusted models include bias from both $\beta_x$ and $\beta_z$, the total bias \emph{may} still be reduced relative to the spatial and non-spatial models. 

\begin{center}
\begin{minipage}{0.45\textwidth}
\includegraphics[width=\linewidth]{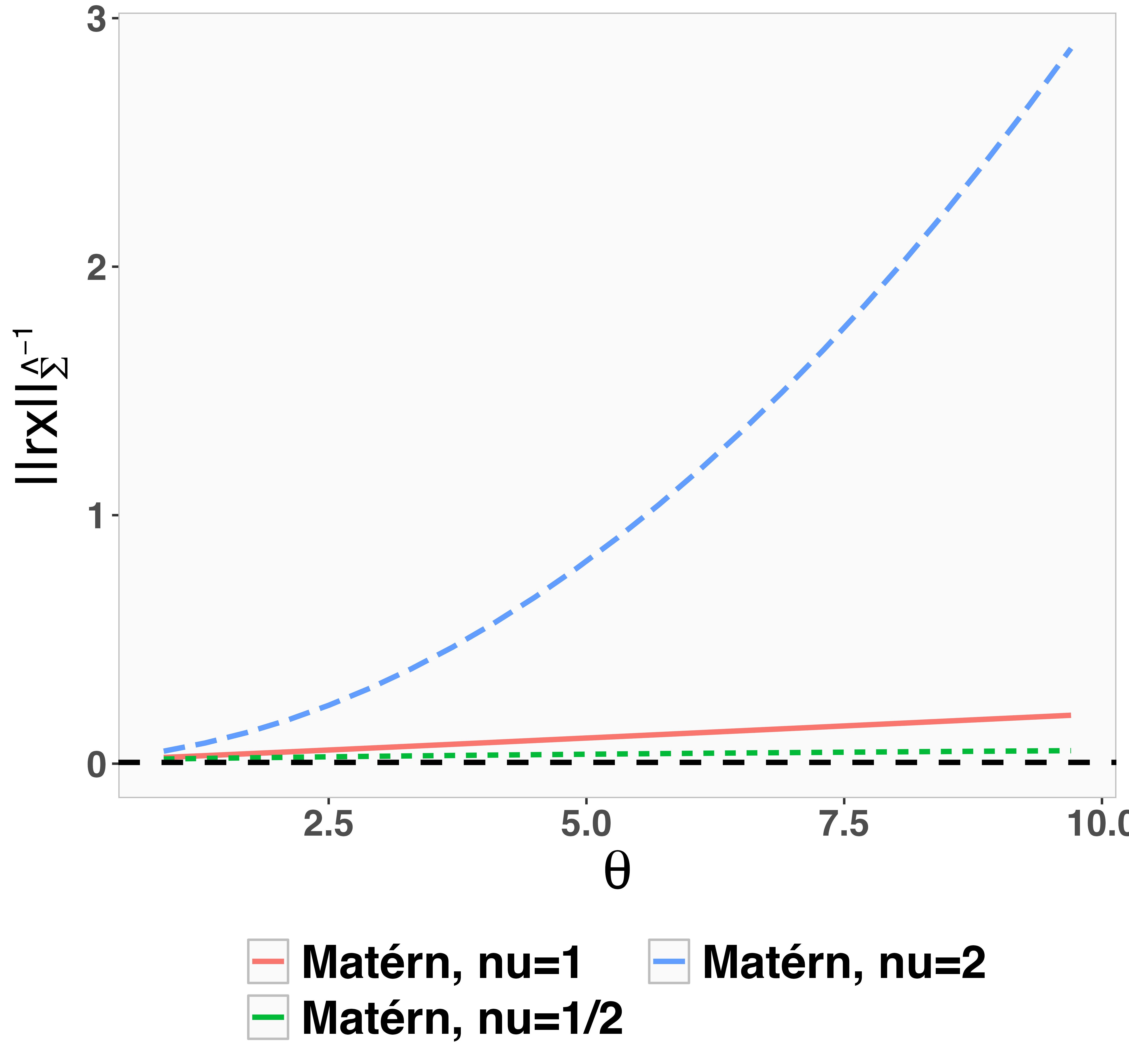}
\captionof*{figure}{a) $\sigmainvnorm{\rb}$}
\end{minipage}\hfill
\begin{minipage}{0.45\textwidth}
\includegraphics[width=\linewidth]{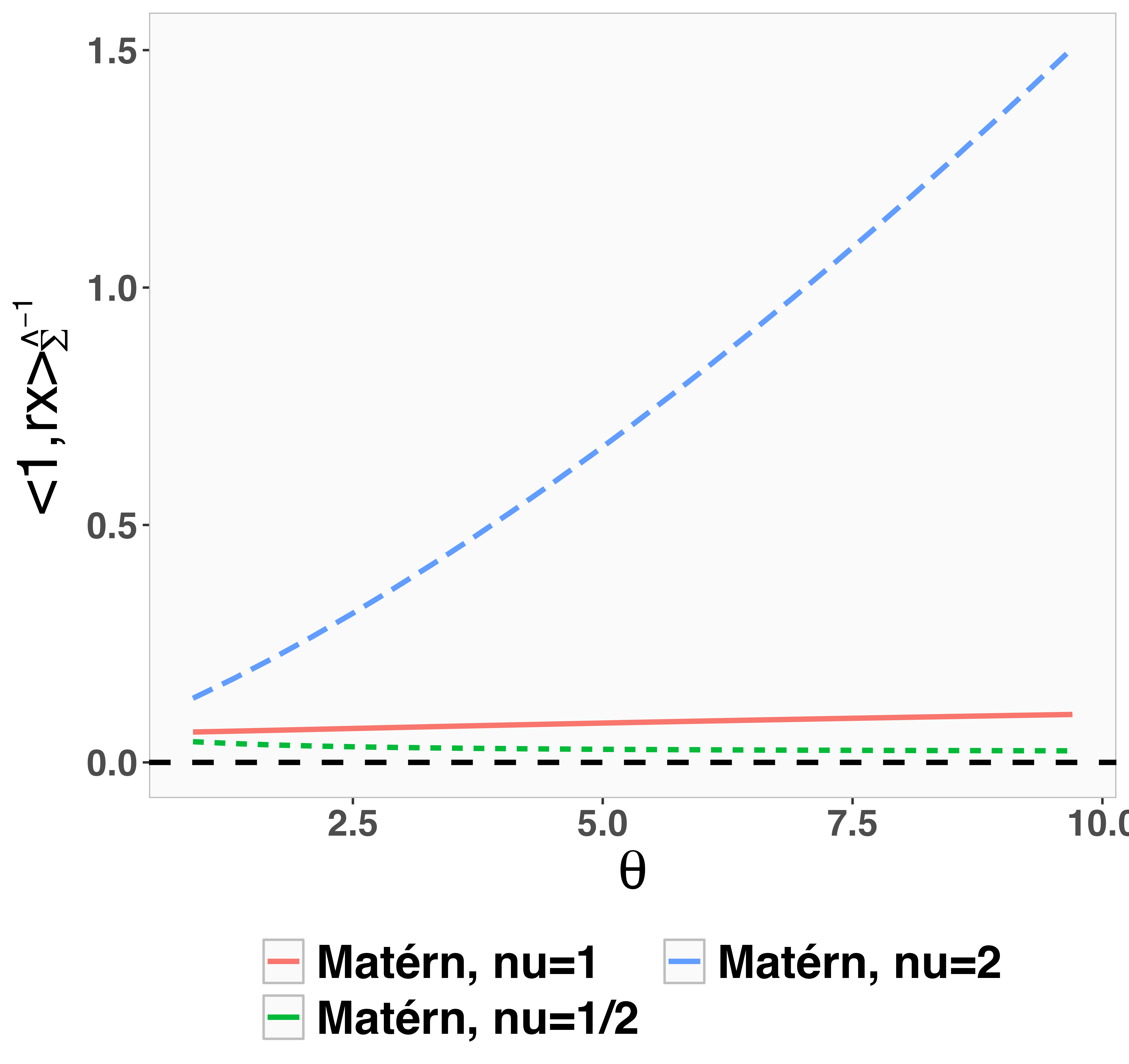} 
\captionof*{figure}{b) $\sigmainvdot{\rb}{\boldone}$  } 
\end{minipage}
\begin{minipage}{0.45\textwidth}
\includegraphics[width=\linewidth]{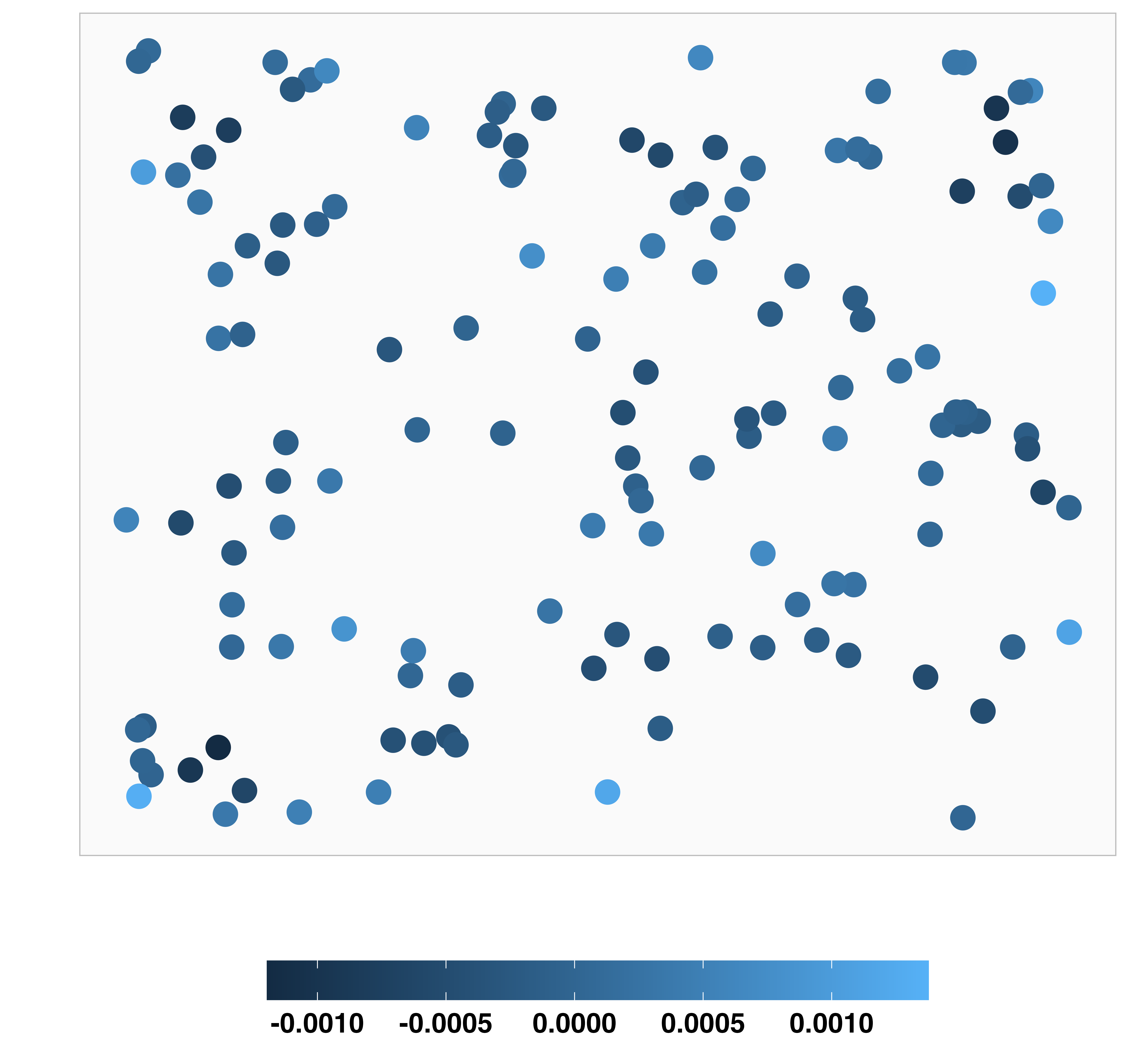}
\captionof*{figure}{c) Illustration of $\rb$}
\end{minipage} \hfill
\begin{minipage}{0.45\textwidth}
\includegraphics[width=\linewidth]{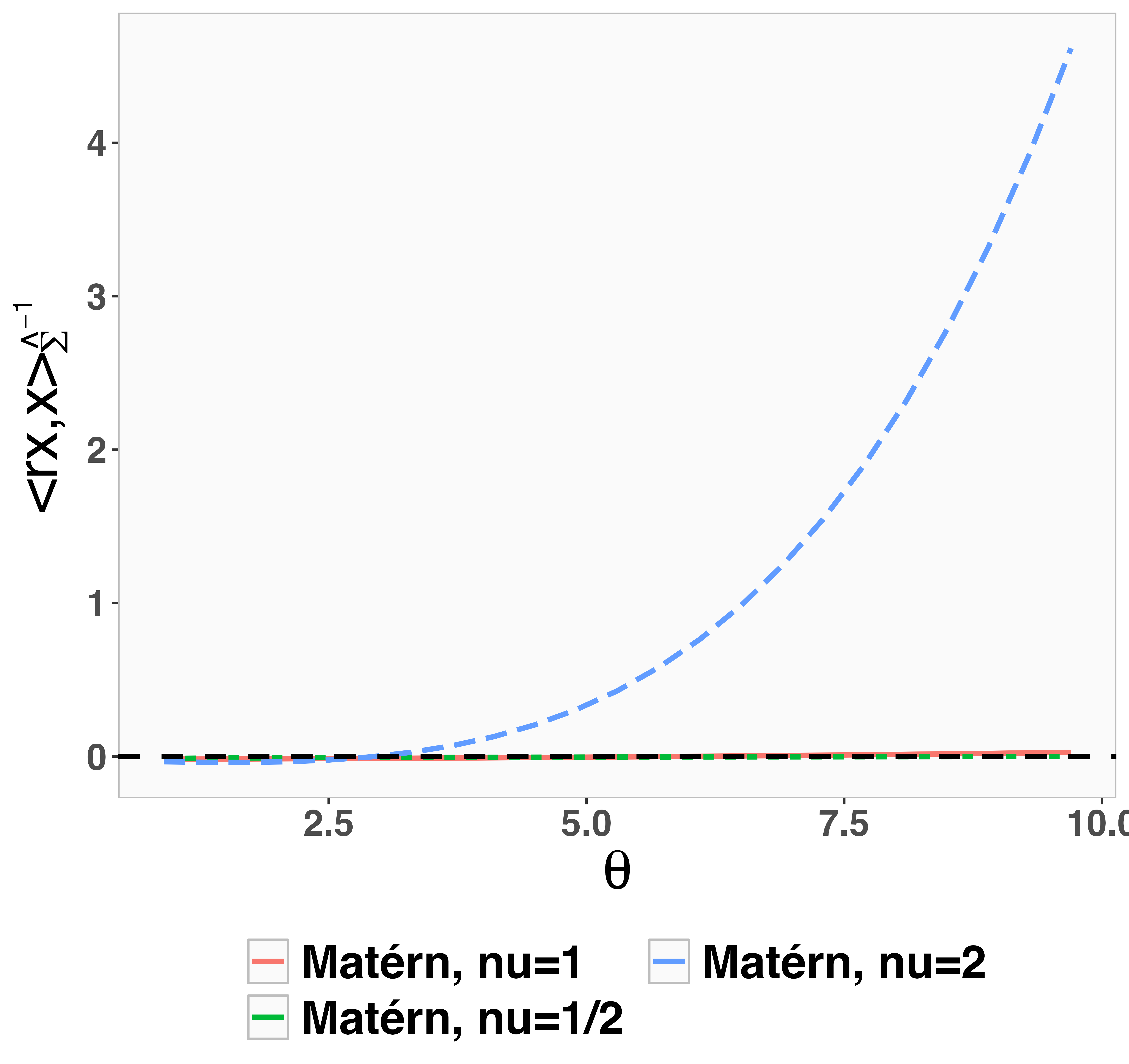}
\captionof*{figure}{ d)$\sigmainvdot{\rb}{\x}$}
\end{minipage}
\captionof{figure}{Illustrations of components of \cref{thm:spatialplus_bias}. All plots were made with the \texttt{ggplot2} package \citet{ggplot2}. }
\label{fig:rx_bias_illustration}
\end{center}

\section{Simulation Studies} \label{sec:simstudies}
In the following simulation studies, we consider settings that have been identified in the literature as times when spatial confounding can distort inference for a regression coefficient. Each simulation study is designed to explore whether insights from the data generation spatial confounding explored in \cref{subsubsec:datagen} or the analysis model spatial confounding explored in \cref{subsubsec:analysis} have any relevance to the patterns of the observed errors for estimates of regression coefficients.  We also get an idea of the variability of the bias rather than just the expected bias.

The results of this paper have primarily focused on spatial models for geostatistical data that involve positive-definite covariance structures. However, the intrinsic conditional autoregressive (ICAR) model \citep{besag1991bayesian} plays an important role in the spatial confounding literature, particularly from a Bayesian paradigm. It was the model first considered in \citet{Hodges_fixedeffects} and \citet{Reich} in the modern introduction to the phenomena of spatial confounding. As referenced previously, the Spatial+ methodology was originally developed for the thin spline plate setting, but the authors stated that the methodology should extend to the ICAR model.  Thus, in these simulation studies, we consider both geostatistical data fit to the class of models considered in our results (referred to as the ``Geostatistical data setting'') as well as areal data fit to an ICAR model (referred to as the ``Areal data setting'' because the ICAR model employs a GMRF).

Because we have not previously defined the ICAR model, we take a moment to do so here. The ICAR model incorporates spatial dependence for areal data with the introduction of an underlying, undirected graph $G = (V, E)$. Non-overlapping spatial regions that partition the study area are represented by vertices, $V = \{ 1, \ldots, n\}$, and edges $E$ defined so that each pair $(i,j)$ represents the proximity between region $i$ and region $j$. We represent $G$ by its $n \times n$ binary adjacency matrix $\bm{A}$ with entries defined such that $\textrm{diag}(\bm{A}) = 0$ and $\bm{A}_{i,j} = \mathbbm{1}_{(i,j) \in E, i \neq j}$. The ICAR model could be considered a generalization of the spatial analysis model of the form \eqref{eq:genericSpatial}, by stating that the spatial random effect has a distribution proportional to a multivariate normal distribution with mean $\bm{0}$ and precision matrix $\tau^2 \left( \bm{I} \bm{A} \bm{1} - \bm{A} \right) = \tau^2 \bm{Q}$, where $\tau^2$ controls the rate decay for the spatial dependence and $\bm{Q}$ is the graph Laplacian.

This precision matrix is not of full rank, so we use a Bayesian analysis to fit the spatial model.  We note there is a close connection between certain types of Bayesian analyses in this setting and modeling spatial random effects through the use of a smoothing penalty, as is done in the thin plate spline setting \citep{kimeldorf1970spline,Rue,dupont2022spatial+}.  However, the two approaches are not the same; for example, a Bayesian model does not provide parameter ``estimates,'' but rather provides posterior distributions.  In this case, to make comparison simple, we use the posterior means for our parameter inferences.  Because the graphs we consider are connected, there is an implicit intercept present in the ICAR model \citep{Pac_ICAR}. Therefore, we omit an intercept from our spatial analysis models. For a Bayesian analysis, $\sigma^2$ and $\tau^2$ require priors. Here, we give them Inverse-Gamma distributions with scale and rate 0.01 each. Finally, to make the non-spatial model comparable, we also use a Bayesian analysis, giving the $\sigma^2$ parameter an Inverse-Gamma prior with the same hyperparameters as the spatial model. All models are fit using Markov chain Monte Carlo (MCMC) algorithms with Gibbs updates. All MCMC's are run for 80,000 iterations with a 20,000 burn-in.

\subsection{Geostatistical Data Setting} \label{sim1:slmm}
In this subsection, we simulate data to replicate the setting explored in \cref{subsubsec:datagen}. The data are all generated from a model of the form \eqref{eq:model_0} as follows:
\[ \y = 0.3 +   3\x + \z + \bm{\epsilon}, \]
where $\bm{\epsilon}$ are independently simulated from a normal  distribution with mean 0 and variance 0.1. 

The 400 locations of the data are randomly generated on $[0,10] \times [0,10]$ window one time, and these locations are then held fixed. The realizations $\x$ and $\z$ are simulated from mean zero spatial processes, denoted respectively $\X$ and $\Z$, with spatial covariance structures defined by $\Rcovst{} = \exp \{ \frac{-d}{\theta} \}$ for Euclidean distance $d$ (i.e., the exponential field).

We define  $\X = {\X}_c + {\X}_u$ and $\Z$ as follows:
$\textrm{Cov} \left(  \X \right) =  \sigma_c^2\Rtwo{c}  + \sigma^2_u\Rtwo{u} $, $\textrm{Cov} \left(  \Z \right) = \sigma_z^2\Rtwo{c}  $, and 
 $\textrm{Cov} \left( \X,  \Z \right) = \rho\sigma_z\sigma_c  \Rtwo{c}$, and set $\sigma^2_c = \sigma^2_z = \sigma^2_u = 0.1$.
 We generate 100 datasets for each $\left( \theta_c, \theta_u, \rho \right) \in \{1,5,10\} \times \{1,5,10\} \times \{-0.9,-0.6,-0.3,0,0.3, 0.6, 0.9\}$. For each dataset, we fit five models: a non-spatial analysis model of the form \eqref{eq:OLSmodel} (``OLS''),  a spatial analysis model of the form \eqref{eq:genericSpatial} estimated using restricted maximum likelihood (``S''),  a spatial analysis model of the form \eqref{eq:genericSpatial} estimated using penalized thin plate splines (``PS''), and two adjusted models of the form \eqref{eq:genericadjSpatial} using the Spatial+ (``S+'') and gSEM (``gSEM'') approaches both making use of penalized thin plate splines using the \texttt{mgcv} R package \citep{woodtps,wood2015package,woodbook,R2024} to estimate the relevant residuals and $\beta_x$.   For the first spatial model, ``S,''  $g()$ is assumed to have spatial structure defined by $\Rcovst{} = \sigma_{s}^2 \exp \{ \frac{-d}{\theta} \}$, with unknown $\theta$ and $\sigma_{s}^2$, and $\epsilon_i$ has a mean 0 and variance $\sigma_{\epsilon}^2$. 

 Under the bias explored from a data generation perspective in \cref{subsubsec:datagen}, because we are using the exponential covariance structure ($\nu=0.5$), the average observed error, $\hat{\beta}_x - \beta_x$, should reflect the results found in \cref{fig:data_illustrationa}(b).  Specifically, that bias will be reduced for the spatial models relative to the non-spatial model when $\theta_u < 2$ and bias may be increased for the spatial model relative to the non-spatial model when $\theta_c < 2$.  Indeed, our simulation results for observed error also show this pattern as discussed in the following paragraphs. The results also provide insights about the range of the observed errors and the interaction between errors expected from the data generating process and errors introduced from the analysis model.  
 
\cref{fig:bias_analysismodel} provides boxplots of the absolute value of the error, $|\hat{\beta}_x - \beta_x|$, for the 100 generated data sets for each of the five analysis models for different values of $(\theta_c, \theta_u, \rho)$.  A figure showing boxplots for the errors for all combinations of $(\theta_c, \theta_u, \rho)$ considered is provided in \cref{app:sim}. 
 
The absolute value of the error for all models increases as $ \vert \rho \vert \rightarrow 1$.  This matches \eqref{olsstochbias} and \eqref{glsstochbias} where the expected bias from a data generating perspective is a function of $\rho$ and will be equal to $0$ when $\rho=0$.  What these equations did not show was that even when the expected bias for both the OLS and spatial models are 0, the variance of the observed errors will be different. Consider \cref{fig:bias_analysismodel}(a), where $\theta_c=\theta_u=10$ and $\rho = 0$.  In this scenario, we see that the OLS model has a much larger range of observed error values than either of the spatial models (S and PS) and the adjusted spatial models (S+ and gSEM).  We also see that both versions of the spatial models have smaller ranges of errors than the adjusted spatial models.  This pattern held for all range parameters considered when $\rho=0$.

\cref{fig:bias_analysismodel}(b) shows the absolute errors for the case when $\theta_c=1$, $\theta_u=10$ and $\rho=-0.9$.  Noticeably, the non-spatial model outperforms both the spatial and adjusted spatial models in terms of lower absolute error values.  This is exactly what we expected given the results in \cref{subsubsec:datagen}.  What is also noteworthy, albeit  expected as discussed in \cref{subsubsec:analysis}, is that the different spatial models (S and PS) have different amounts of error.  In this case, the spatial model estimated via REML slightly outperforms the spatial model fit with penalized thin plate splines.  When we consider the adjusted spatial models, these two models perform comparably to each other, but slightly worse than the spatial models.

\begin{center}
\begin{minipage}{0.45\textwidth}
\includegraphics[width=\linewidth]{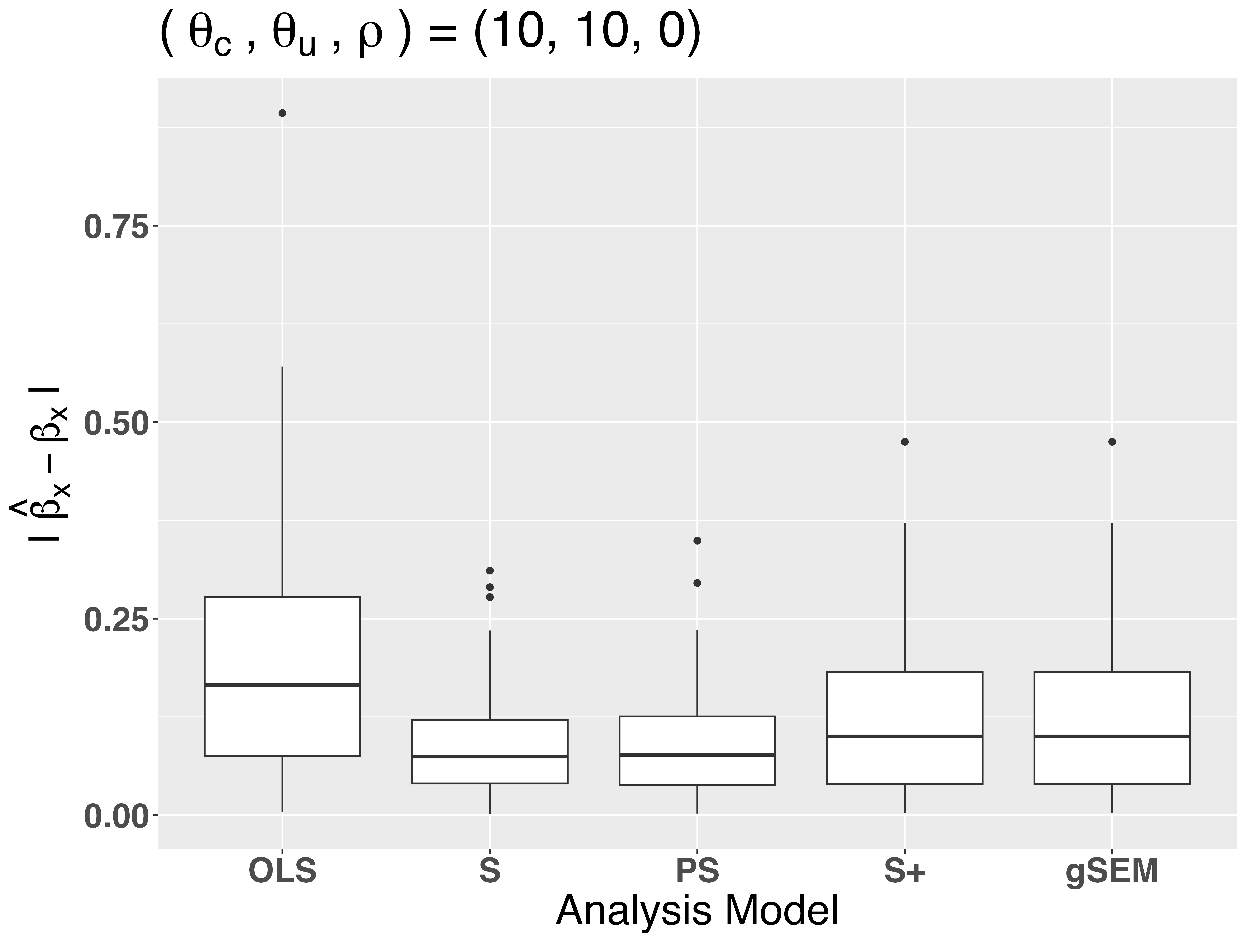}
\centering (a)
\end{minipage}\hfill
\begin{minipage}{0.45\textwidth}
\includegraphics[width=\linewidth]{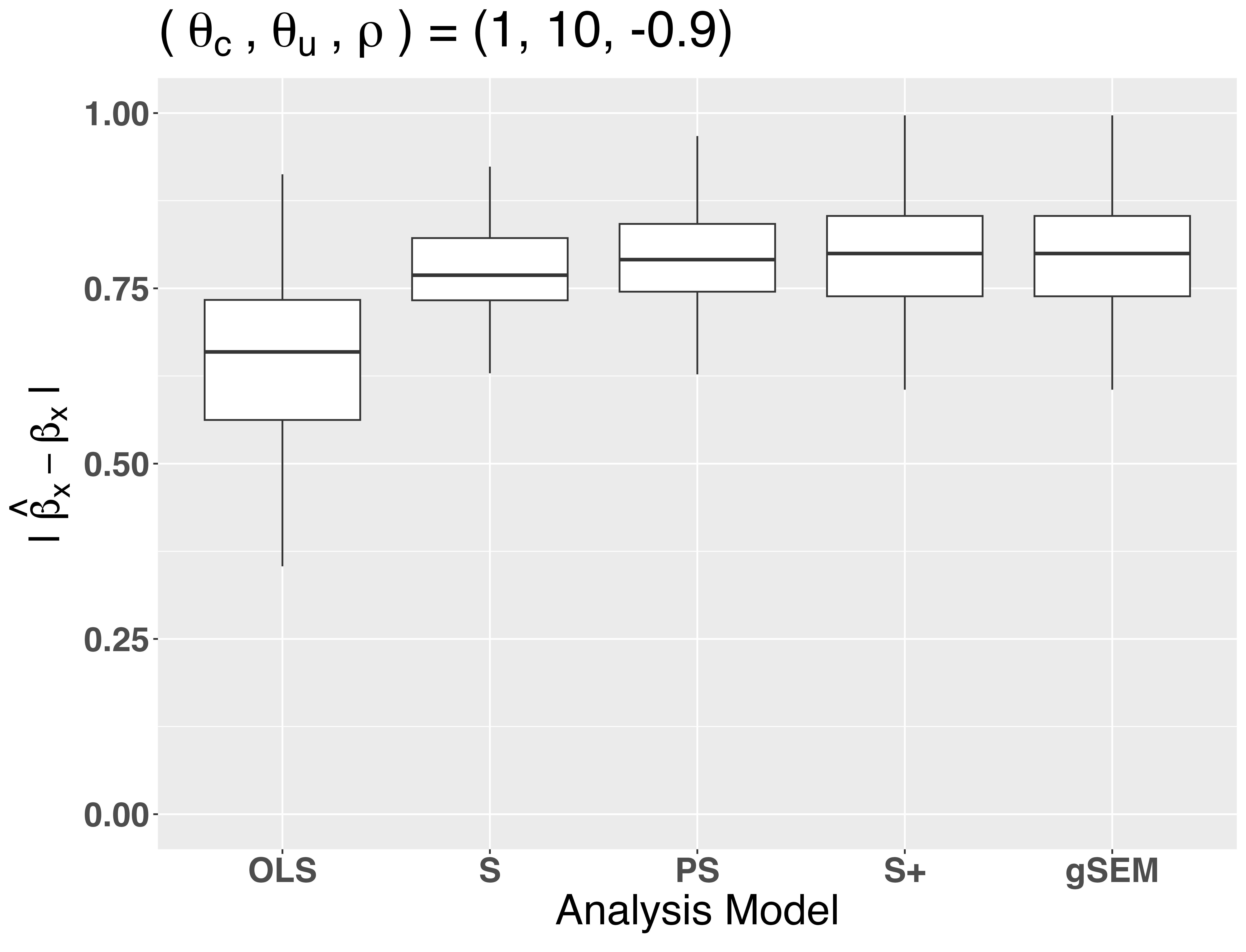}
\centering (b)
\end{minipage}\hfill
\begin{minipage}{0.45\textwidth}
\includegraphics[width=\linewidth]{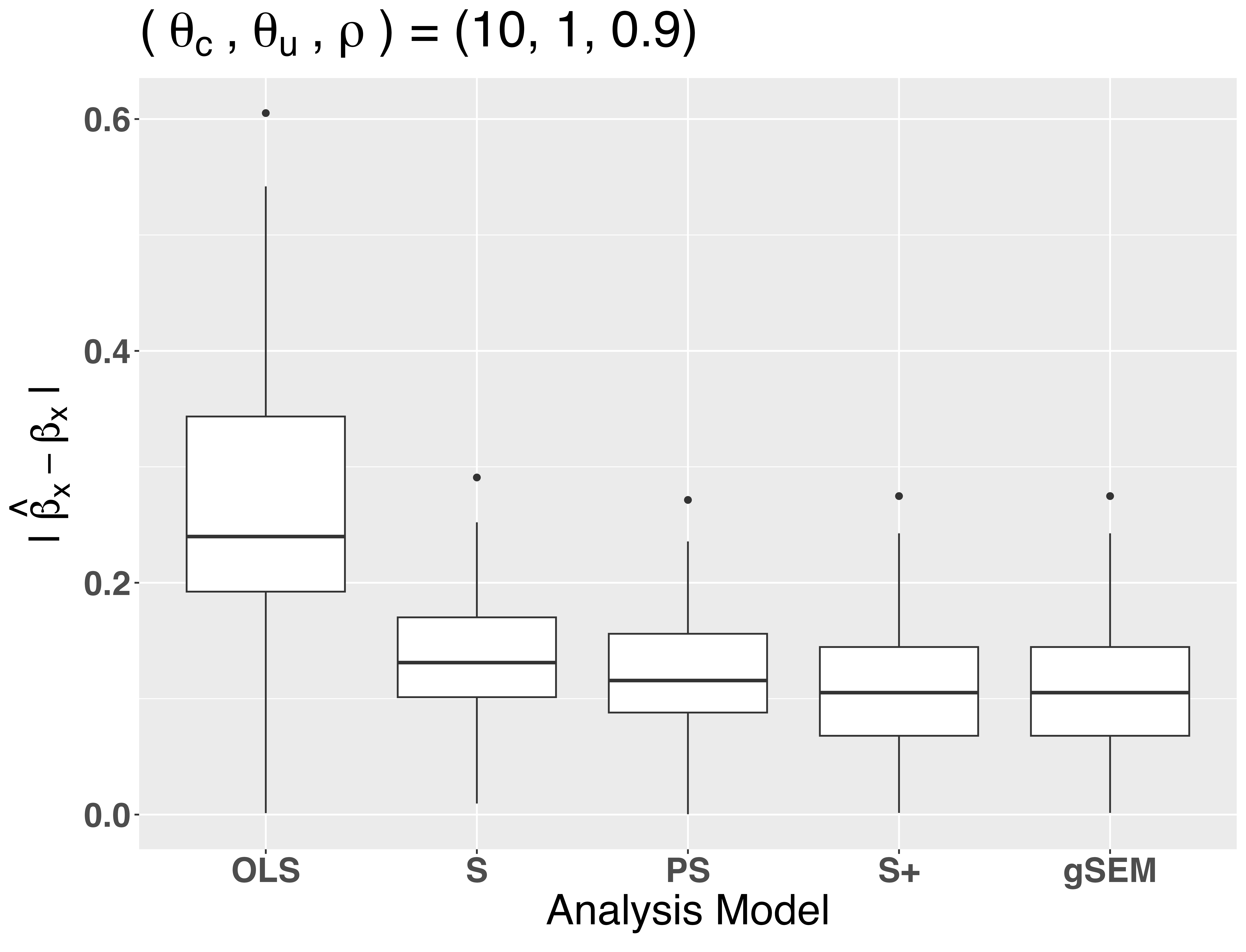}
\centering (c)
\end{minipage}\hfill
\begin{minipage}{0.45\textwidth}
\includegraphics[width=\linewidth]{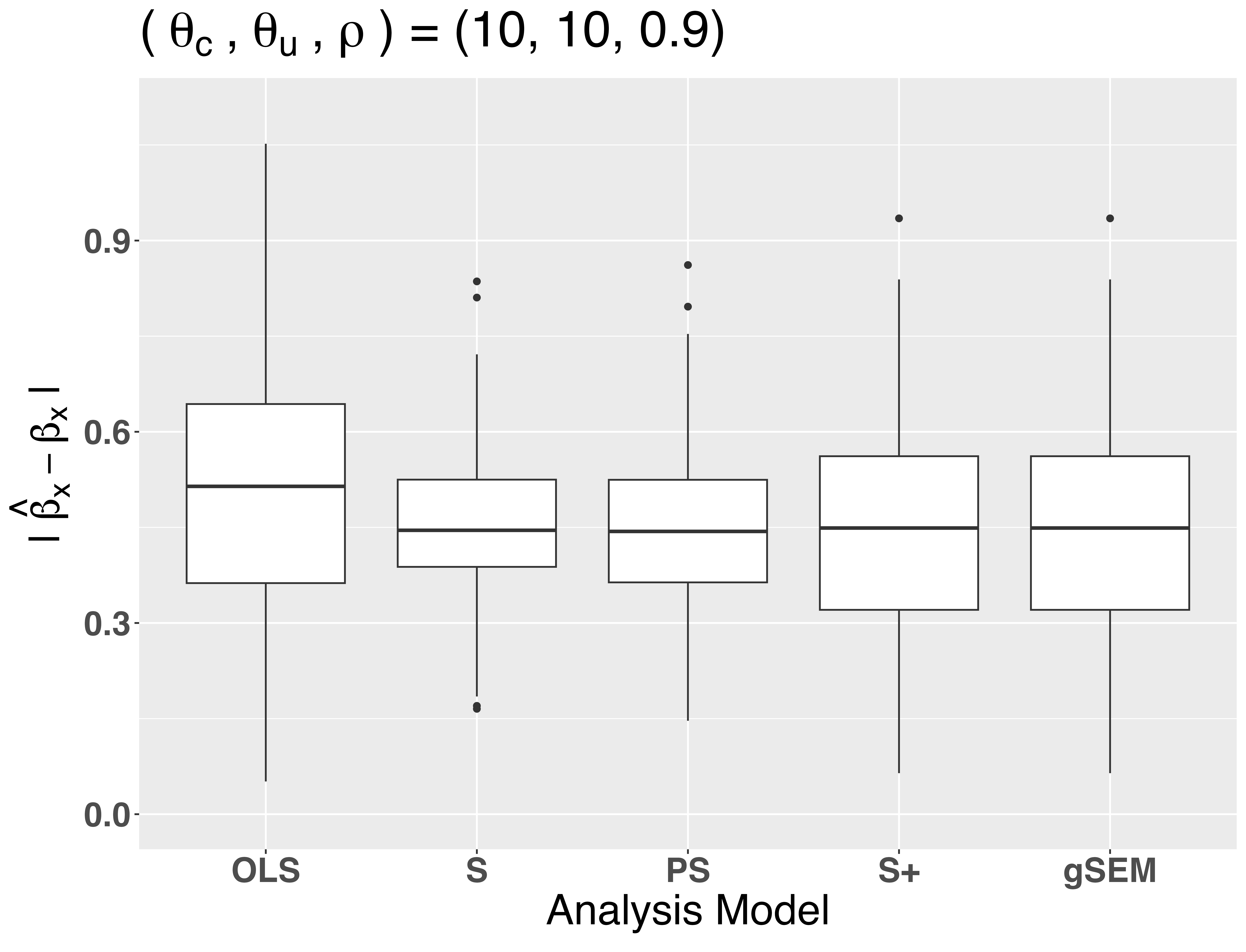}
\centering (d)
\end{minipage}\hfill
\captionof{figure}{Boxplots of the absolute errors, $|\hat{\beta_x} -\beta_x|$, for different data generating covariance parameter values. Note that the y-axes are different for each plot so that the errors for each model are more easily distinguishable within the individual plots.  Plots made with \texttt{ggplot2} \citet{ggplot2}.}
\label{fig:bias_analysismodel}
\end{center}

\cref{fig:bias_analysismodel}(c) shows the ``ideal'' result, in that the models meant to account for data-generating source of spatial dependence and improve inference do improve inference on $\beta_x$.  In fact, this data generating scenario is very similar to that used in \cite{dupont2022spatial+}: high correlation between $\x$ and $\z$ with $\theta_c$ large and $\theta_u$ small.  And in this case, the absolute error for both adjusted models is better than both spatial models and the non-spatial model. This figure emphasizes how models meant to accommodate data generating source of spatial confounding can do so when the data generating assumptions are met. However, when the data are generated even under a single different scenario (e.g., a different value of $\theta_u$ as in \cref{fig:bias_analysismodel}(d) ), these models can increase bias. Also notable is the range of the errors.  In \cref{fig:bias_analysismodel}(c), the ranges are similar for the spatial and adjusted spatial models, while the range for the non-spatial model errors is almost double that of the others.  Note that from a data generating source of confounding perspective, we expect the spatial model bias to be lower than the non-spatial model setting as shown in \cref{fig:data_illustrationa}(b) (i.e., $\theta_u < 2$ when $\nu=0.5$).

Finally, \cref{fig:bias_analysismodel}(d) illustrates how the errors compared across models when $\theta_c=\theta_u$ (here $=10$) and $\rho =0.9$.  While the average absolute error is approximately the same (slightly larger for the non-spatial model than the others), the range of the errors for the spatial models is smaller than for the non-spatial model and both of the adjusted spatial models.  

This simulation drives home the finding in \cref{subsubsec:datagen}: the amount of bias from estimating $\beta_x$ will depend on the data generating process.  However, it also highlights that the bias from the analysis model sources of confounding depends on the underlying assumptions about the data. There is not a ``one size fits all'' modeling approach.

\subsection{Areal Data Setting} \label{sim1:icar}
In the second setting, we work with areal data on an $11 \times 11$ grid on the unit square. Recall that work in  analysis model spatial confounding suggest that a covariate which is collinear or correlated with low-frequency eigenvectors of the precision matrix of the spatial random effect could induce bias in the estimation of $\beta_x$. This is thought to be true regardless of whether there is a ``missing'' spatially dependent covariate. Therefore, in this simulation, we attempt to explore whether that is the case by simulating datasets with both a spatially-smooth covariate and a covariate without spatial structure.  We also compare these results to the areal data setting when the spatially dependent covariate is and is not present in the data generating model.  

For all simulated datasets, 
\begin{align*}
\x &= 0.5 \z + \bm{\epsilon}_{x},\\
\y & = 3\x + \beta_z \z + \bm{\epsilon},
\end{align*}
where $\bm{\epsilon}_{x}$ and $\bm{\epsilon}$ are independently generated from a normal distribution with mean 0 and standard deviations 0.01 and 0.15, respectively.  

We consider two possible choices of $\z$: one in which $\z$ is spatially smooth from an analysis model spatial confounding perspective and one in which it is not. For the latter category, we simply generate $\z$ from a normal distribution with mean 0 and standard deviation 0.09, as depicted in the left plot of \cref{fig:illustrations_covariate}.  We hold this vector fixed and simulate values of $\x$ and $\y$ 100 times.  To generate the spatially-smooth covariate, we use the eigenvectors of the graph Laplacian $\bm{Q}$.  For the ICAR model, there is not a variance-covariance matrix, but rather the singular precision matrix.  However, we can treat this as the pseudo-inverse of a variance-covariance matrix \citep{Pac_ICAR}.  In this case, then, if $\z$ is strongly correlated with a low-frequency eigenvector of $\bm{Q}$, the spatial analysis model may sometimes perform more poorly than the non-spatial model as discussed in \cref{subsubsec:analysis}.   Thus, we let $\z$  be the eigenvector of $\bm{Q}$ associated with smallest, non-zero eigenvalue, depicted in the right plot of \cref{fig:illustrations_covariate}. We also hold this vector fixed for 100 simulated datasets.

We also consider two choices for $\beta_z$: $\beta_z = 0$ and $\beta_z = -1$.  When $\beta_z = 0$, we are explicitly leaving out any residual spatial dependence from the data generation model to better explore the impact of a covariate alone.  This will better enable us to examine bias from the perspective of the analysis model source of confounding.  When $\beta_z = -1$, we are including confounding from a data generating perspective.  This is the value of $\beta_z$ used in \cite{Thaden} and \cite{dupont2022spatial+} where they showed that the adjusted spatial models outperform the spatial and non-spatial models.

For each of the datasets, we consider the same five analysis models considered in \cref{sim1:slmm}, but with two adjustments. Here, the spatial analysis model is the ICAR model and we use a Bayesian approach for both the spatial and the non-spatial analysis models, as described in the introduction of this section.  For the locations needed in the penalized spline and adjusted spatial models, we use the centroids of the $11\times 11$ grid.  We note here that thin plate splines are meant for geostatistical data, where data are observed over a continuous random field.  As noted earlier, areal data and geostatistical data are quite different and their dependence structures are modeled in different manners.  Areal data is observed over a discrete space and relies on the specification of a neighborhood matrix to define an autoregressive dependence structure (e.g., GMRF).  Geostatistical data is observed over a continuous space and relies on the specification of a distance matrix to define a continuous dependence structure.  That said, it has been shown that the dependence structure of one form can approximate the dependence structure of another quite well for a fixed set of locations \citep{rue:2002}.  Because of this, and to stay true to the Spatial+ model as proposed by \cite{dupont2022spatial+}, we treat the data as a continuous random field using fixed centroids of the grid as the locations for the three models fit using splines (PS, S+, and gSEM).

\begin{center}
\begin{minipage}{0.45\textwidth}
\includegraphics[width=\linewidth]{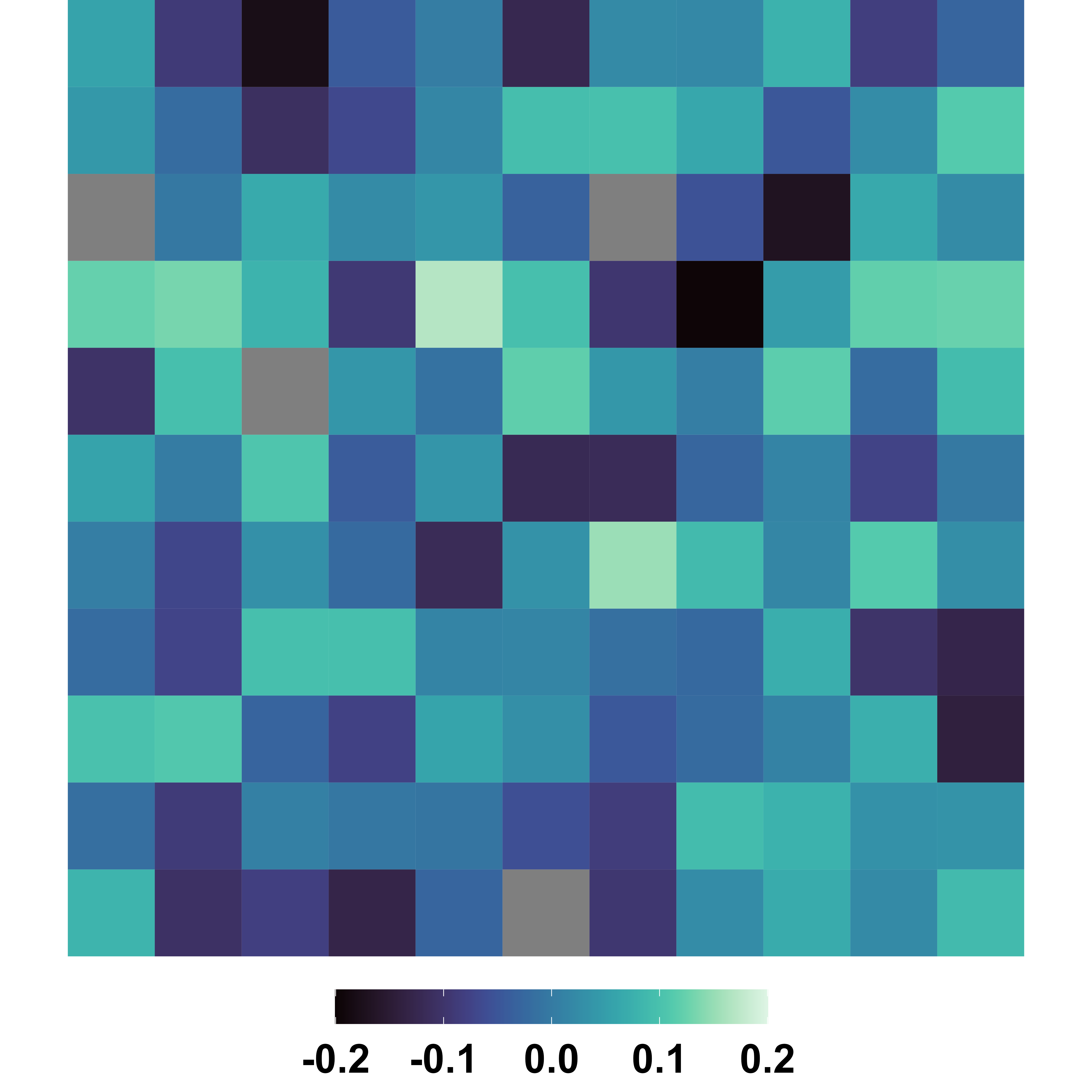}
\centering (a) Random
\end{minipage}\hfill
\begin{minipage}{0.45\textwidth}
\includegraphics[width=\linewidth]{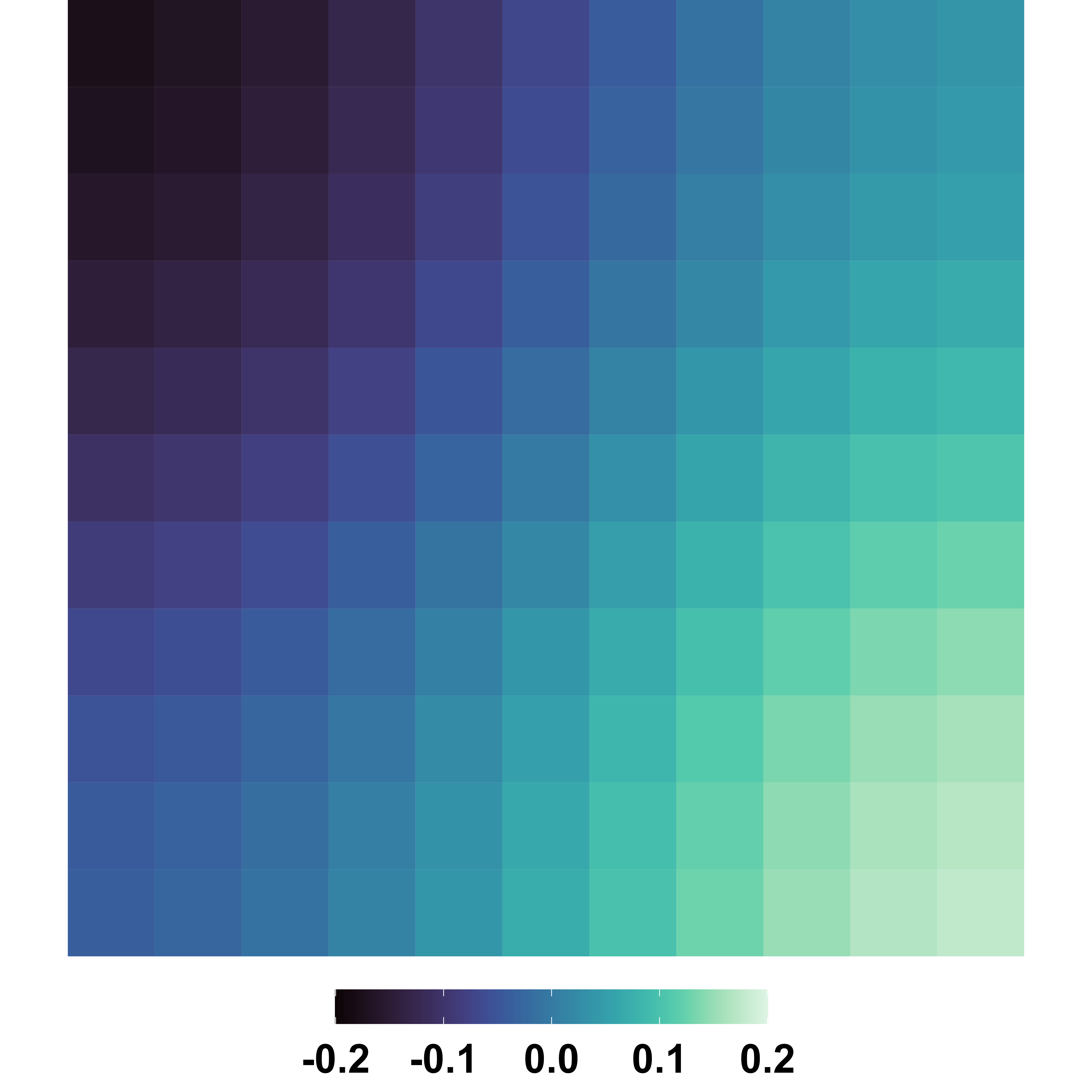}
\centering (b) Spatially-Smooth
\end{minipage}
\captionof{figure}{Visualization of (a) $\z$ generated from a normal distribution (``random''), (b) $\z$ defined as the low-frequency eigenvector of $\bm{Q}$ (``spatially smooth''). Plots made with the \texttt{ggplot2} package \citep{ggplot2}. }
\label{fig:illustrations_covariate}
\end{center}

 \cref{fig:am_bias} shows the observed absolute error for the five models and four simulation scenarios.  Plots (a) and (c) show the bias when $\z$ is random and is either included (a) or excluded (c) from the model generating $\y$.  First consider (a).  This scenario is akin to an ``ommitted variable'' where there is no spatial dependence.  In this case, the non-spatial model and two spatial models perform similarly, while the adjusted spatial models, although have larger ranges in error, can actually have less error.  Spatial+ actually has the smallest error for 58\% of the data sets.  The opposite story is true in (c), where the adjusted models tend to perform worse.  In this case, the non-spatial model is the most appropriate model given the true data generating scenario and its error tends to be the smallest among the models.  

 Moving to cases when $\z$ is generated to be spatially smooth in (b) and (d).  There are three interesting findings here.  First, when all of the spatial dependence is in the covariate (i.e., $\beta_z=0$ so that there is no additional spatial dependnece added to $\y$), the spatial model, ``S,'' actually sometimes performs \emph{worse} than the non-spatial model.  This provides evidence for the methods motivated by alleviating the analysis model sources of confounding.  However, we note that this is an extreme case, where $\x$ is very highly correlated (e.g., $\approx 0.98$) with the smallest non-zero eigenvector of the precision matrix. This is essentially a worst-case scenario, and the spatial analysis model and non-spatial analysis models still yielded similar inferences.  

 Second, it's interesting to note the differences in the two versions of the spatial models: ``S'' and ``PS.''  In the case just discussed, the penalized spline spatial model performs quite a bit worse than the Bayesian ICAR spatial model.  In plot (b), however, the penalized spline model appears to be the best fitting model and fits much better than the Bayesian ICAR spatial model.  This drives home the point that how the models are fit is as important as the model itself.  Two different spatial models can lead to different estimates and errors.  

 Third, the adjusted models perform very well in case (b).  From a data generating perspective, this is expected since this is the data generating scenario for which the methods were derived: smooth $\z$ and strong correlation between $\x$ and $\z$.  However, the adjusted models in case (d) perform quite poorly.  Recall that for the adjusted model bias, both $\beta_x$ and $\beta_z$ contribute to the total bias as shown in \eqref{eq:adj_bias}, in contrast to the spatial and non-spatial model biases, where there is no bias from $\beta_x$.  Thus, in this case, while the spatial and non-spatial models will have zero expected bias, the adjusted models will have some bias due to $\beta_x$ ($A_2^*(\rb, \x)$) and it will not be 0 since $\rb$ and $\x$ will not be similar with respect to $\Sigmahinv$.

 These simulations drive home the point that \emph{both} the data generating and analysis model sources of confounding can contribute to the bias of $\hat{\beta}_x$.  Importantly, if the data generating assumptions for which a model was developed are not met, the analysis model can have drastically worse fit than anticipated.

\begin{center}
\begin{minipage}{0.45\textwidth}
\includegraphics[width=\linewidth]{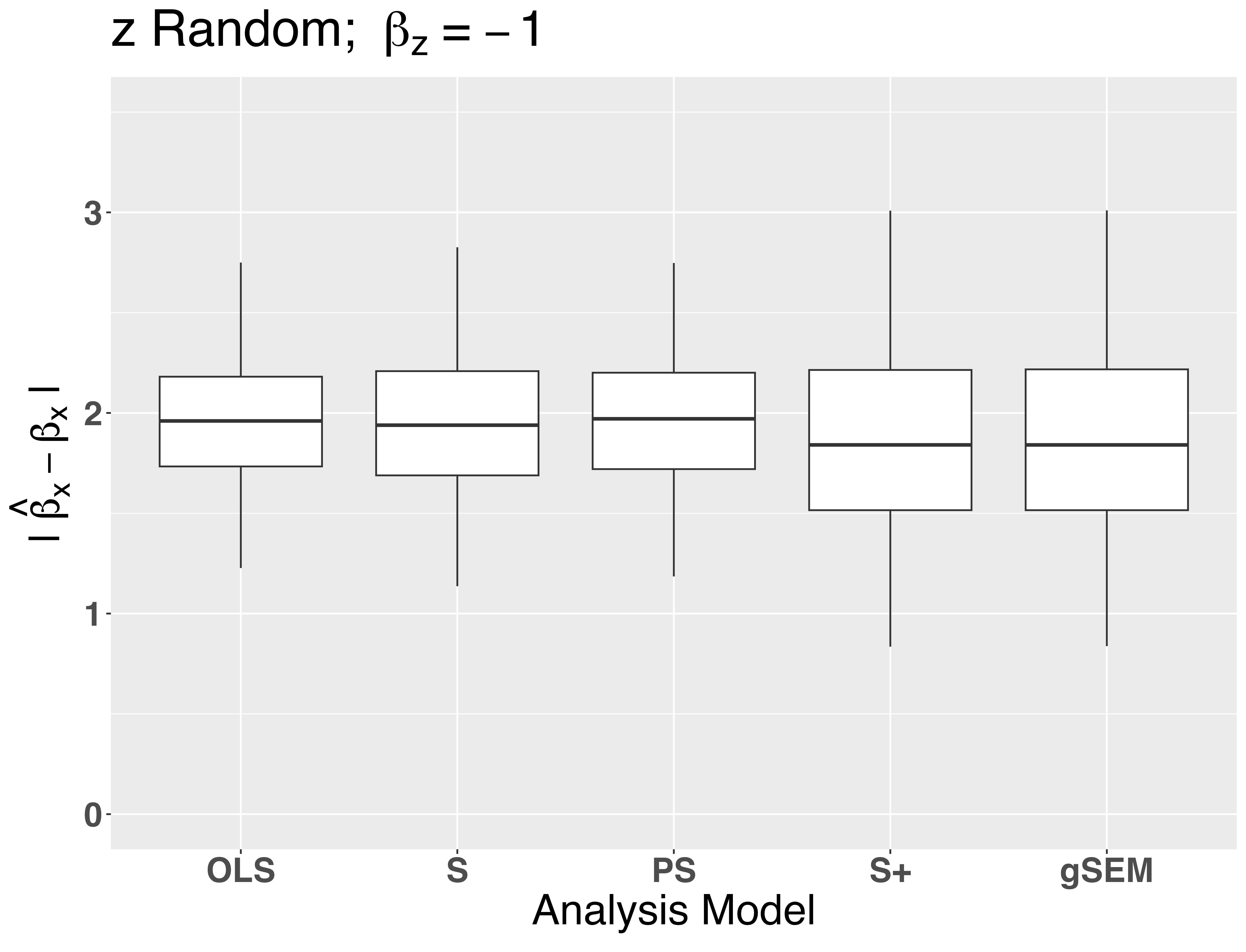}
\centering (a) 
\end{minipage}\hfill
\begin{minipage}{0.45\textwidth}
\includegraphics[width=\linewidth]{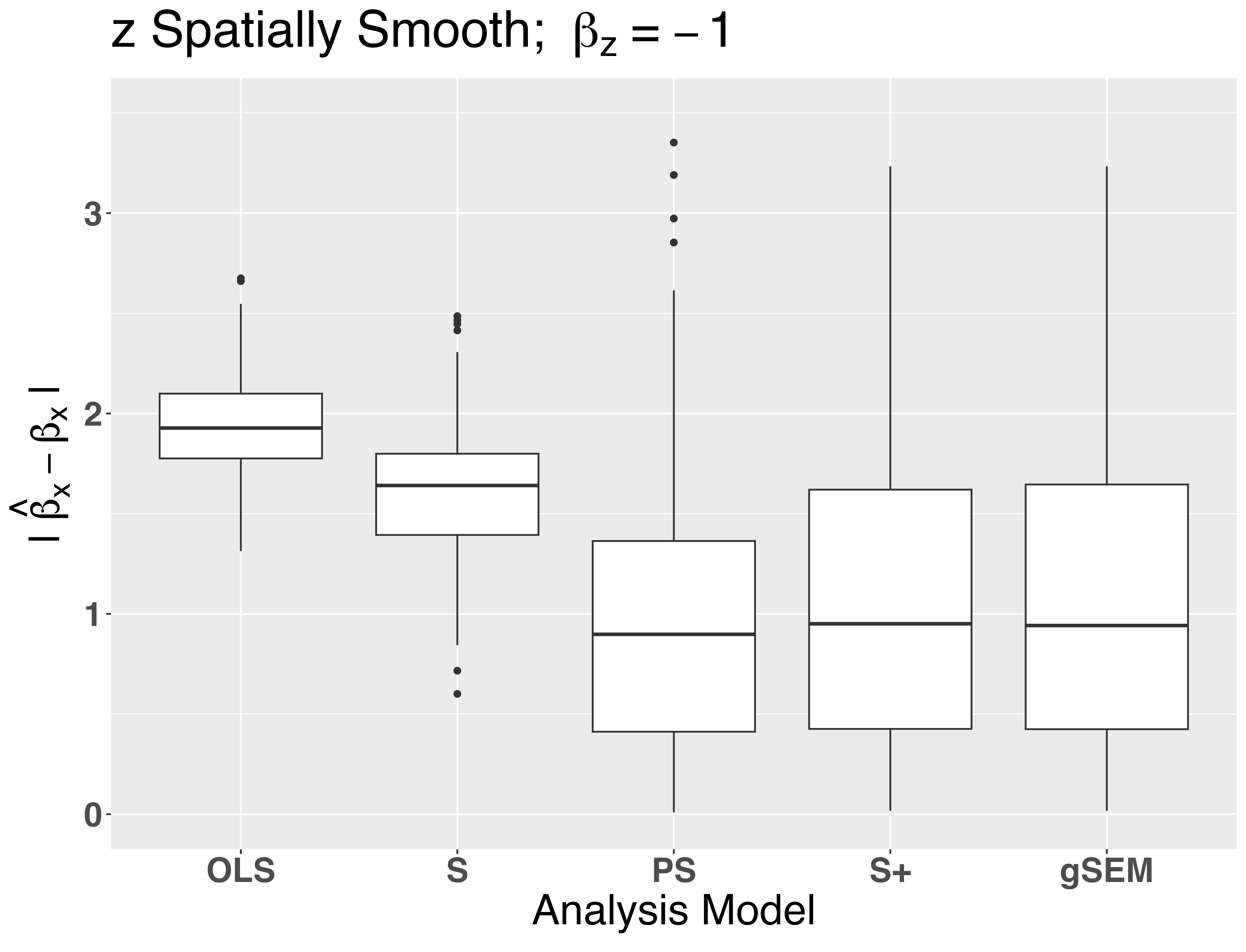}
\centering (b) 
\end{minipage}
\begin{minipage}{0.45\textwidth}
\includegraphics[width=\linewidth]{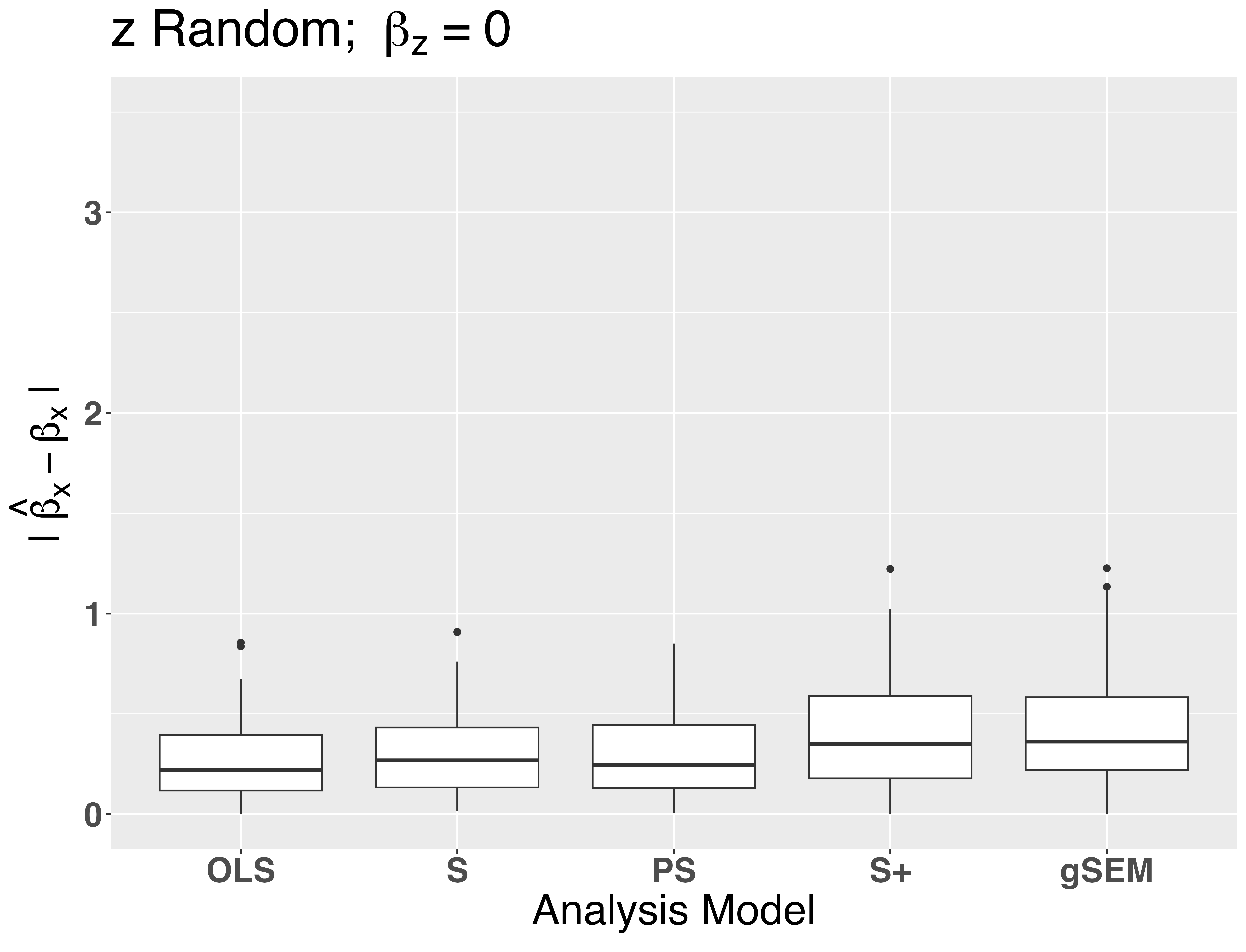}
\centering (c) 
\end{minipage}\hfill
\begin{minipage}{0.45\textwidth}
\includegraphics[width=\linewidth]{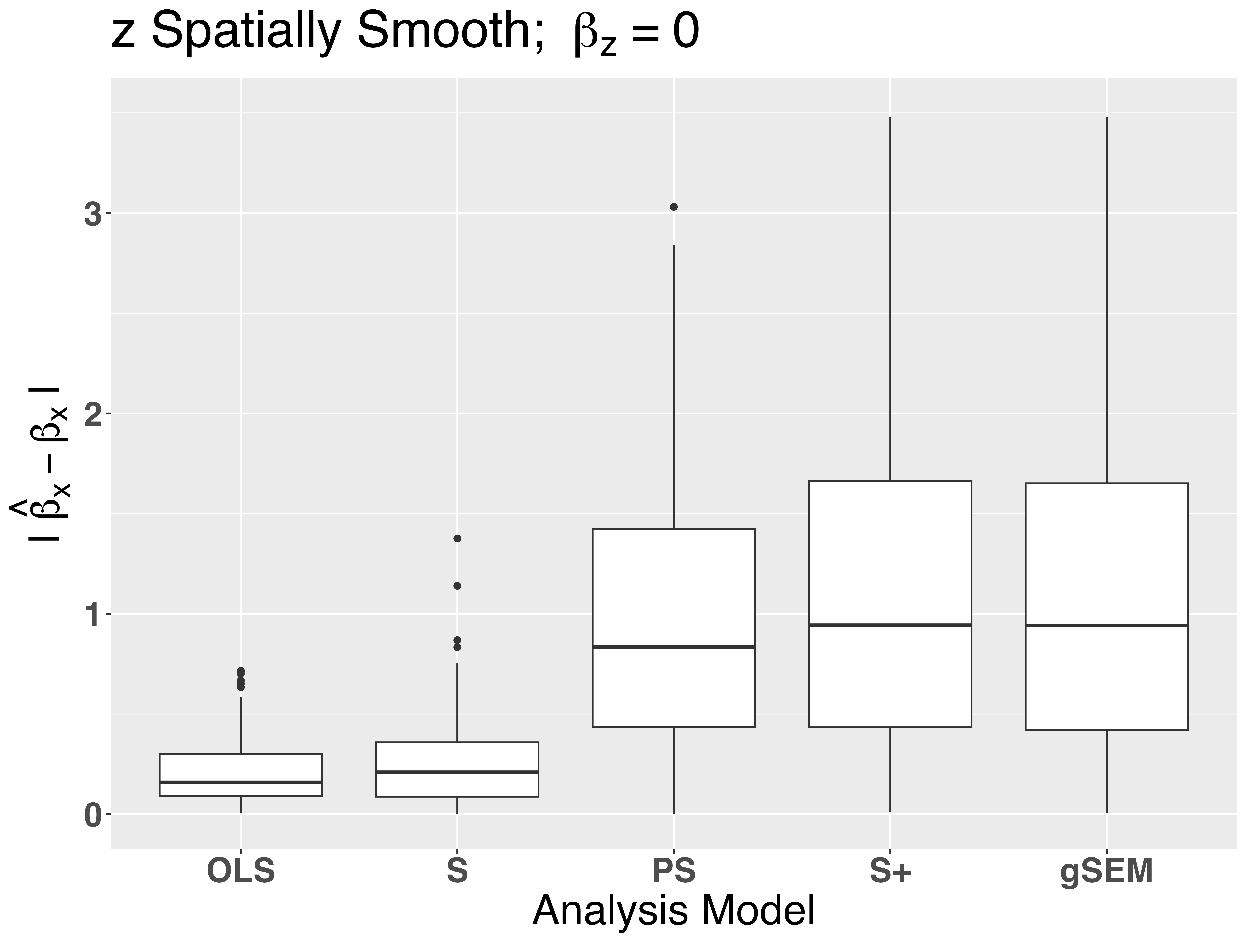}
\centering (d) 
\end{minipage}
\captionof{figure}{Absolute value of the observed bias for areal data simulations. (a) $\z$ is random and $\beta_z=-1$; (b) $\z$ is spatially smooth and $\beta_z=-1$; (c) $\z$ is random and $\beta_z=0$; (d) $\z$ is spatially smooth and $\beta_z=0$. Plots made with the \texttt{ggplot2} package \citep{ggplot2}. }
\label{fig:am_bias}
\end{center}

\section{Discussion} \label{sec:conclusion}

In this paper, we have synthesized the broad, and often muddled, literature on spatial confounding. We have introduced two broad focuses in the spatial confounding literature: the analysis model focus and the data generation focus. Using the spatial linear mixed model, we have shown how papers focused on the former category often conceptualize the problem of spatial confounding as originating from the relationship between an observed covariate $\x$ and the estimated precision matrix $\Sigmahinv$ of a spatial random effect. We then showed how papers focused on the latter category typically identify the problem of spatial confounding as originating from the relationship between an observed covariate ($\x$) and a correlated, unobserved covariate ($\z$). 

Our results highlight two important conclusions: 1) the original conceptualization of spatial confounding as problematic may not have been entirely correct, and 2) the analysis model and data generation perspectives of spatial confounding can lead to directly contradictory conclusions about whether spatial confounding exists and whether it adversely impacts inference on regression coefficients. With respect to the first point, the modern conceptualization of spatial confounding arose in work by \citet{Reich} and \citet{Hodges_fixedeffects}. In our proposed framework, these papers focused on an analysis model type of spatial confounding. In the context of an ICAR model, they argued that whenever $\x$ was collinear or highly correlated with a low-frequency eigenvector of the graph Laplacian $\bm{Q}$, the regression coefficients would be biased (relative to the regression coefficients obtained from a non-spatial model). 
Our results suggest that it is only in relatively extreme cases, where $\x$ is ``flat'' and there is no spatially smooth residual dependence, that bias for regression coefficients can increase. In our simulation study, we produced such a setting by generating $\x$ to be almost perfectly correlated with a low frequency eigenvector of the graph Laplacian. Even in this extreme scenario, however, the bias seen in a spatial analysis model was not that much different from the bias seen in a non-spatial analysis model. 

Turning our attention to the second point, the data generation perspective of spatial confounding often relies on very specific assumptions about the processes that generated $\x$ and $\z$ or on very specific assumptions about the relationship between these variables (i.e., $\x$ is a combination of $\z$ and some Gaussian noise). Our results suggested that many of the scenarios identified as problematic from a data generation perspective are not problematic from an analysis model perspective. 
This is potentially troublesome because many of these papers propose methods to ``alleviate'' spatial confounding based on the perceived problem (the relationship between $\x$ and $\z$).  

Taken together, the results and simulation studies in this paper offer support for the conventional wisdom of spatial statistics: accounting for residual spatial dependence tends to improve inference on regression coefficients. However, spatial analysis models are not interchangeable: the analysis model \textit{and the method used to fit it} matter.  As research continues into the field of spatial confounding, we offer this advice: 
\begin{enumerate}
    \item Considering spatial confounding from both an analysis model and data generation perspective may elucidate future paths for identifying when specific data or analysis models  will create or exacerbate bias.
    \item Researchers should carefully consider whether both the data generating and analysis model assumptions are met for their specific data and analysis goals.  
\end{enumerate}

This work did not focus on the variability or uncertainty of the parameter estimates.  The simulation studies hinted that this is extremely relevant to differences in model biases and work to better quantify and understand these differences is needed.

Finally, we caution against using these findings too broadly.  We considered specific models, specific data settings, and true and known smoothing in the dependence structures. As demonstrated in this study, applying research findings beyond their intended context can result in unreliable and questionable outcomes.

\pagebreak
\setcitestyle{numbers} 
\bibliographystyle{apalike} 
\bibliography{Bib} %

\newpage

\appendix

\section{Notation for metrics and norms} \label{appa_diff_geometry}


\begin{definition}[Standard Euclidean Inner Product and Norm] \label{euclidean_metric}
We use the notation $\euclmetric{\cdot}{\cdot}$ to denote the standard Euclidean inner product on the vector space of $\mathbb{R}^n$. The notation $\euclnorm{\cdot}$ is then used to refer to the norm inducted by this metric. Specifically, for $\a,\b \in \mathbb{R}^n$:
\begin{eqnarray*}
\euclmetric{\a}{\b} &=& {\a}^T \b \\
\euclnorm{\a}  &=& \sqrt{{\a}^T  \a } 
\end{eqnarray*}
\end{definition}

\begin{notation}[Angles with respect to the Standard Euclidean Inner Product] \label{euclidean_angle}
Given $a,b \in \mathbb{R}^n$, we use $\euclangle{a}{b}$ to refer to the angle between these two vectors with respect to the standard Euclidean norm. Specifically: 
\begin{eqnarray*}
\euclangle{\a}{\b} = \arccos \left( \frac{\euclmetric{\a}{\b}}{ \euclnorm{\a}\euclnorm{\b}} \right) 
\end{eqnarray*}
\end{notation}

\begin{notation}[Spectral Decomposition of $\Sigmahinv$] \label{spectral_decomp}
Let $\Sigmahinv$ be a $n\times n$ real, symmetric, positive definite matrix.

We define $\U \D \U = \Sigmahinv$ to be the spectral decomposition of $\Sigmahinv$ with $\D$ a diagonal matrix with diagonal $d_1 \geq \ldots \geq d_n > 0$.
\end{notation}

\begin{notation}[Angles between Vector and Eigenvectors of $\Sigmahinv$] \label{coef_notation}
Let $\Sigmahinv$ be a $n\times n$ real, symmetric, positive definite matrix, and $\v$ be an arbitrary vector in $\mathbb{R}^n$.

Let $\U \D \U = \Sigmahinv$ be the spectral decomposition of $\Sigmahinv$ as defined in \cref{spectral_decomp}. In this paper, we use the notation  $\bm{\theta}_{\v,\U}$ to define a $n\times 1$ vector whose $i$th element is the angle $\euclangle{\v}{\u_i}$ (with respect to the Euclidean norm as in \cref{euclidean_angle}) between $\v$ and the $i$th column $\u_i$ of $\U$. 
\end{notation}

\begin{definition}[Precision Matrix Induced Inner Product and Norm] \label{precmatrix_metric}
Given a $n \times n$ real, symmetric, positive definite matrix, $\Sigmahinv$ we use the notation $\sigmainvdot{\cdot}{\cdot}$ to denote the inner product defined by the matrix on the vector space of $\mathbb{R}^n$. The notation $\sigmainvnorm{\cdot}$ is then used to refer to the norm induced by this inner product. More precisely, for $\a,\b \in \mathbb{R}^n$:
\begin{eqnarray*}
\sigmainvdot{\a}{\b} &=& \a^T \Sigmahinv \b \\
\sigmainvnorm{\a}  &=& \sqrt{\a^T \Sigmahinv \a }
\end{eqnarray*}
\end{definition}

\begin{notation}[Angles with respect to the Precision Matrix Induced Inner Product] \label{precmatrix_angle}
Given $\a,\b \in \mathbb{R}^n$, we use $\sigmainvangle{\a}{\b}$ to refer to the angle between them with respect to the standard Euclidean norm. Technically, it would be more appropriate to use $\sigmainvangle{\a}{\b}^{\Sigmahinv}$. However, we drop the dependency on $\Sigmahinv$ unless it is required for ease of reading. Specifically: 
\begin{eqnarray*}
\sigmainvangle{\a}{\b} = \arccos \left( \frac{\sigmainvdot{\a}{\b}}{ \sigmainvnorm{\a}\sigmainvnorm{\b}} \right) 
\end{eqnarray*}
\end{notation}

\section{Proofs and Derivations}

\subsection{Derivation of \cref{stoch_ols_bias} Non-Spatial Bias}  \label{app:stoch_ols_bias}
\begin{eqnarray*}
	\textrm{E} \left(  \hat{\bm{\beta}}_{NS}  | \Xone \right)&=&  \left( {\Xone}^T \Xone \right)^{-1} {\Xone}^T \left( \beta_0 \boldone + \beta_x \X + \beta_z \textrm{E} \left( \Z | \X \right) \right) \\ 
	&=&  \bm{\beta} +  \beta_z \left( {\Xone}^T \Xone \right)^{-1} {\Xone}^T \left( \mu_z \boldone + \rho \sigma_c \sigma_z \Rtwo{c}^T \left(\sigma_c^2 \Rtwo{c}  + \sigma_u^2 \Rtwo{u} \right)^{-1} \right. \\ & &  \left. \left( \X - \mu_x \boldone \right) \right)\\
	&=&  \bm{\beta} +  \beta_z \left( {\Xone}^T \Xone \right)^{-1} {\Xone}^T \left[ \mu_z \boldone + \rho \sigma_c \sigma_z  \left(\sigma_c^2 \bm{I}  + \sigma_u^2 \Rtwo{u} \Rtwo{c}^{-1} \right)^{-1} \right. \\ & & \left. \left( \X - \mu_x \boldone \right) \right]\\ 
	&=&  \bm{\beta} +  \beta_z  \left[ \mu_z \left( {\Xone}^T \Xone \right)^{-1} {\Xone}^T \boldone + \rho \sigma_c \sigma_z \left( {\Xone}^T \Xone \right)^{-1} {\Xone}^T \left(\sigma_c^2 \bm{I}  + \sigma_u^2 \Rtwo{u} \Rtwo{c}^{-1} \right)^{-1} \right. \\ & & \left. \left( \X - \mu_x \boldone \right) \right]\\ 
	&=&  \bm{\beta} +  \beta_z  \left[ \mu_z \begin{bmatrix}
		1 \\
		0 \\
	\end{bmatrix} + \rho \sigma_c \sigma_z \left( {\Xone}^T \Xone \right)^{-1} {\Xone}^T \left(\sigma_c^2 \bm{I}  + \sigma_u^2 \Rtwo{u} \Rtwo{c}^{-1} \right)^{-1} \right. \\ & & \left. \left( \X - \mu_x \boldone \right) \right]\\ 
	&=&  \bm{\beta} +  \beta_z  \left[ \mu_z \begin{bmatrix}
		1 \\
		0 \\
	\end{bmatrix} + \rho \frac{\sigma_z}{\sigma_c} \left( {\Xone}^T \Xone \right)^{-1} {\Xone}^T \bm{K} \left( \X - \mu_x \boldone \right) \right]\\ 
\end{eqnarray*}
where $\bm{K}= p_c \left(p_c \bm{I}  + (1-p_c) \Rtwo{u} \Rtwo{c}^{-1} \right)^{-1} $ and $p_c=\frac{\sigma_c^2}{\sigma_c^2 + \sigma_u^2}$. We now restrict our attention to the second element:
\begin{eqnarray*}
\textrm{E} \left(  \hat{\beta}_x^{NS}  | \Xone \right)&=& \beta_x + \beta_z \rho  \frac{\sigma_z}{\sigma_c} \left[ \left( {\Xone}^T \Xone \right)^{-1} {\Xone}^T \bm{K} \left( \X - \mu_x \boldone \right) \right]_2
\end{eqnarray*}

\subsection{Derivation of \cref{stoch_ols_bias} Spatial Bias} \label{app:stoch_gls_bias}
\begin{eqnarray*}
	\textrm{E} \left(  \hat{\bm{\beta}}^{S}  | \Xone \right)&=&  \left( {\Xone}^T \bm{\Sigma}^{-1} \Xone \right)^{-1} {\Xone}^T \bm{\Sigma}^{-1} \left( \beta_0 \boldone + \beta_x \X + \beta_z \textrm{E} \left( \Z | \X \right) \right) \\ 
	&=&  \bm{\beta} +  \left( {\Xone}^T \bm{\Sigma}^{-1} \Xone \right)^{-1} {\Xone}^T \bm{\Sigma}^{-1}  \beta_z \left( \mu_z \boldone + \rho \sigma_c \sigma_z \Rtwo{c}^T \left(\sigma_c^2 \Rtwo{c}  + \sigma_u^2 \Rtwo{u} \right)^{-1} \right. \\ & &  \left. \left( \X - \mu_x \boldone \right) \right)\\ 
	&=&  \bm{\beta} +  \left( {\Xone}^T \bm{\Sigma}^{-1} \Xone \right)^{-1} {\Xone}^T \bm{\Sigma}^{-1}  \beta_z \left[ \mu_z \boldone + \rho \sigma_c \sigma_z  \left(\sigma_c^2 \bm{I}  + \sigma_u^2 \Rtwo{u} \Rtwo{c}^{-1} \right)^{-1} \right. \\ & & \left. \left( \X - \mu_x \boldone \right) \right]\\ 
	&=&  \bm{\beta} +   \beta_z \left[ \mu_z \begin{bmatrix}
		1 \\
		0 \\
	\end{bmatrix} + \rho \sigma_c \sigma_z  \left( {\Xone}^T \bm{\Sigma}^{-1} \Xone \right)^{-1} {\Xone}^T \bm{\Sigma}^{-1} \left(\sigma_c^2 \bm{I}  + \sigma_u^2 \Rtwo{u} \Rtwo{c}^{-1} \right)^{-1} \right. \\ & & \left. \left( \X - \mu_x \boldone \right) \right]\\ 
	&=&  \bm{\beta} +   \beta_z \left[ \mu_z \begin{bmatrix}
		1 \\
		0 \\
	\end{bmatrix} + \rho \frac{\sigma_z}{\sigma_c}  \left( {\Xone}^T \bm{\Sigma}^{-1} \Xone \right)^{-1} {\Xone}^T \bm{\Sigma}^{-1} \bm{K} \left( \X - \mu_x \boldone \right) \right]\\ 
\end{eqnarray*}
where, $\bm{K}= p_c \left(p_c \bm{I}  + (1-p_c) \Rtwo{u} \Rtwo{c}^{-1} \right)^{-1} $, $p_c=\frac{\sigma_c^2}{\sigma_c^2 + \sigma_u^2}$ , and $\bm{\Sigma} = \beta_z^2 \sigma_z^2 \Rtwo{c} + \sigma^2 \bm{I}$. This calculation of the variance is done conditional on $\X$ to mirror the results in \citet{Pac_spatialconf}.

We restrict our attention to the second element, and note: 
\begin{eqnarray*}
\textrm{E} \left(  \hat{\beta}_X^{S}  | \Xone \right) &=& \beta_x +   \beta_z \rho \frac{\sigma_z}{\sigma_c} \left[ \left( {\Xone}^T \bm{\Sigma}^{-1} \Xone \right)^{-1} {\Xone}^T \bm{\Sigma}^{-1} \bm{K} \left( \X - \mu_x \boldone \right) \right]_2
\end{eqnarray*}

\subsection{Derivation of \cref{eq:obsolsfixed}} \label{app:obsolsfixed} 
\begin{eqnarray*}
	\textrm{E} \left(  \hat{\bm{\beta}}^{NS}  \right)&=&  \left( {\xone}^T \xone \right)^{-1} {\xone}^T  \left( \beta_0 \boldone + \beta_x \x + \beta_z  \z \right) \\ 
	&=& 	\bm{\beta}  +  \beta_z  \left( {\xone}^T \xone \right)^{-1} {\xone}^T  \z \\
	&=&  	\bm{\beta}  + \beta_z \left( \begin{bmatrix}
		{\boldone}^T  \boldone &  {\boldone}^T \x\\
		{\x}^T   \boldone & {\x}^T   \x \\
	\end{bmatrix} \right)^{-1} {\xone}^T \z  \\ 
	&=& 	\bm{\beta}  +  \frac{\beta_z}{	{\boldone}^T  \boldone {\x}^T   \x -   {\boldone}^T \x	{\x}^T  \boldone  }   \left( \begin{bmatrix}
		{\x}^T   \x 	{\boldone}^T  \z  - {\boldone}^T   \x 	{\x}^T \z\\
		 {\boldone}^T  \boldone  {\x}^T  \z -	{\x}^T   \boldone 	{\boldone}^T  \z\\
	\end{bmatrix} \right) \\ 
	&=& \bm{\beta}  +  \frac{\beta_z}{	\euclnorm{\boldone}^2 \euclnorm{\x}^2 -   \left[ \euclmetric{\x}{\boldone} \right]^2  }  \begin{bmatrix}
		\euclnorm{\x}^2 \euclmetric{\z}{\boldone} - \euclmetric{\x}{\boldone} 	\euclmetric{\x}{\z}\\
		 \euclnorm{\boldone}^2  \euclmetric{\x}{\z} -	\euclmetric{\x}{\boldone} \euclmetric{\z}{\boldone}\\
	\end{bmatrix} . \\ 
\end{eqnarray*}

Restricting our attention to the second element:
\begin{eqnarray*}
 \textrm{E} \left(  \hat{\bm{\beta}}_x^{NS}  \right) &=& \beta_x +
 \frac{\beta_z}{	\euclnorm{\boldone}^2 \euclnorm{\x}^2 -   \left[ \euclmetric{\x}{\boldone}	{\x}^T \right]^2  } \left( \euclnorm{\boldone}^2  \euclmetric{\x}{\z} -	\euclmetric{\x}{\boldone} \euclmetric{\z}{\boldone} \right)
\end{eqnarray*}

Thus, $\bias{\hat{\bm{\beta}}_x^{NS}} = \frac{\beta_z}{	\euclnorm{\boldone}^2 \euclnorm{\x}^2 -   \left[ \euclmetric{\x}{\boldone} \right]^2  } \left( \euclnorm{\boldone}^2  \euclmetric{\x}{\z} -	\euclmetric{\x}{\boldone} \euclmetric{\z}{\boldone} \right)$

\subsection{Proof of \cref{eq:obsglsfixed}} \label{app:obsglsfixed} 
\begin{eqnarray*}
	\textrm{E} \left(  \hat{\bm{\beta}}^g  \right)&=&  \left( {\xone}^T \Sigmahinv \xone \right)^{-1} {\xone}^T \Sigmahinv \left( \beta_0 \boldone + \beta_x \x + \beta_z  \z \right) \\ 
	&=& 	\bm{\beta}  +  \beta_z  \left( {\xone}^T \Sigmahinv \xone \right)^{-1} {\xone}^T \Sigmahinv \z \\
	&=&  	\bm{\beta}  + \beta_z \left( \begin{bmatrix}
		{\boldone}^T \Sigmahinv  \boldone &  {\boldone}^T \Sigmahinv  \x\\
		{\x}^T \Sigmahinv \boldone & {\x}^T \Sigmahinv  \x \\
	\end{bmatrix} \right)^{-1} {\xone}^T \Sigmahinv \z  \\ 
	&=& 	\bm{\beta}  + \frac{\beta_z}{	{\boldone}^T \Sigmahinv  \boldone {\x}^T \Sigmahinv  \x -   {\boldone}^T \Sigmahinv  \x 	{\x}^T \Sigmahinv  \boldone  }   \left( \begin{bmatrix}
		{\x}^T \Sigmahinv  \x & - {\boldone}^T \Sigmahinv  \x\\
		-	{\x}^T \Sigmahinv \boldone & {\boldone}^T \Sigmahinv \boldone \\
	\end{bmatrix} \right) {\xone}^T \Sigmahinv \z  \\ 
	&=& \bm{\beta} + \frac{\beta_z}{	{\boldone}^T \Sigmahinv  \boldone {\x}^T \Sigmahinv  \x -   {\boldone}^T \Sigmahinv  \x 	{\x}^T \Sigmahinv  \boldone  }   \begin{bmatrix}
		{\x}^T \Sigmahinv  \x & - {\boldone}^T \Sigmahinv \x\\
		-	{\x}^T \Sigmahinv  \boldone & {\boldone}^T \Sigmahinv \boldone \\
	\end{bmatrix}  \begin{bmatrix}
		{\boldone}^T \Sigmahinv \z   \\
		{\x}^T \Sigmahinv \z   \\
	\end{bmatrix}   \\ 
	&=& 	\bm{\beta}  +  \frac{\beta_z}{	{\boldone}^T \Sigmahinv  \boldone {\x}^T \Sigmahinv  \x -   {\boldone}^T \Sigmahinv  \x	{\x}^T \Sigmahinv  \boldone  }   \begin{bmatrix}
		{\x}^T \Sigmahinv  \x 	{\boldone}^T \Sigmahinv \z  - {\boldone}^T \Sigmahinv  \x 	{\x}^T \Sigmahinv \z\\
	  {\boldone}^T \Sigmahinv \boldone  	{\x}^T \Sigmahinv \z 	-	{\x}^T \Sigmahinv  \boldone 	{\boldone}^T \Sigmahinv \z \\
	\end{bmatrix} \\ 
	&=& 	\bm{\beta}  +  \frac{\beta_z}{ \sigmainvnorm{\boldone}^2 \sigmainvnorm{\x}^2 -  \left[ \sigmainvdot{\x}{\boldone} \right]^2  }   \begin{bmatrix}
		\sigmainvnorm{\x}^2 \sigmainvdot{\z}{\boldone}  - \sigmainvdot{\x}{\boldone} \sigmainvdot{\x}{\z}\\
	  \sigmainvnorm{\boldone}^2 \sigmainvdot{\x}{\z} - 	\sigmainvdot{\x}{\boldone}	\sigmainvdot{\z}{\boldone} \\
	\end{bmatrix}.
\end{eqnarray*}

Restricting our attention to the second element: 
\begin{eqnarray*}
\textrm{E} \left(  \hat{\bm{\beta}}_x^g  \right) = 	\beta_x  +  \frac{\beta_z}{ \sigmainvnorm{\boldone}^2 \sigmainvnorm{\x}^2 -  \left[ \sigmainvdot{\x}{\boldone} \right]^2  } \left( 
	  \sigmainvnorm{\boldone}^2 \sigmainvdot{\x}{\z} - 	\sigmainvdot{\x}{\boldone}	\sigmainvdot{\z}{\boldone} \right).
\end{eqnarray*}

\subsection{Derivation of \cref{eq:obsadjglsfixed}} \label{app:obsadjglsfixed}

We assume that $\rb$ are known residuals from spatial model fit with $\x$ as the response variable.  We assume the model used to obtain these residuals gave an estimate of $\hat{\x} = \hat{\bm S}_x \x$ so that $\rb = \left[\bm{I} - \hat{\bm S}_x\right]\x$ for some hat matrix $\hat{\bm S}_x$.  We consider the special case when the spatial model for $\x$ is fit with only an intercept in \cref{subapp:nulladj}. For the following, we denote $\rone = \left[ \boldone ~ \rb \right]$.
\begin{eqnarray}
	\textrm{E} \left(  \hat{\bm{\beta}}^{S+}  \right)&=&  \left( {\rone}^T \bm{\Sigmah}^{-1} \rone \right)^{-1} {\rone}^T \bm{\Sigmah}^{-1} \left( \beta_0 \boldone + \beta_x \x + \beta_z  \z \right) \nonumber \\ 
	&=&  \left( {\rone}^T \bm{\Sigmah}^{-1} \rone \right)^{-1} {\rone}^T \bm{\Sigmah}^{-1} \left( \beta_0 \boldone + \beta_x \x  \right)  +  \left( {\rone}^T \bm{\Sigmah}^{-1} \rone \right)^{-1} {\rone}^T \bm{\Sigmah}^{-1} \left( \beta_z  \z \right) \nonumber \\ &=&
	\textrm{A}^{S+} \left( \rb, \x \right) + \textrm{B}^{S+} \left( \rb, \x \right) \label{eq:adj_bias}
\end{eqnarray}

The first term $\textrm{A}^{S+} \left( \rb, \x \right)$ is no longer simply $\bm{\beta}$ as it was in \cref{app:obsglsfixed}. Thus, we consider the terms $\textrm{A}^{S+} \left( \rb, \x \right)$ and $\textrm{B}^{S+} \left( \rb, \x \right)$ separately.

\paragraph{Simplifying $\textrm{A}^{S+} \left( \rb, \x \right)$} In this sub-section, we focus on the first term of \eqref{eq:adj_bias}.  
\begin{eqnarray*}
\textrm{A}^{S+} \left( \rb, \x \right) &=& \left( {\rone}^T \bm{\Sigmah}^{-1} \rone \right)^{-1} {\rone}^T \bm{\Sigmah}^{-1} \left( \beta_0 \boldone + \beta_x \x  \right)\\
&=&  \left[ \left( \begin{bmatrix}
 	{\boldone}^T \bm{\Sigmah}^{-1}  \boldone &  {\boldone}^T \bm{\Sigmah}^{-1}  \rb \\
 	{\rb}^T \bm{\Sigmah}^{-1}  \boldone & {\rb}^T \bm{\Sigma}^{-1}  \rb \\
 \end{bmatrix} \right)^{-1} {\rone}^T \bm{\Sigmah}^{-1} \boldone \beta_0 +  \right. \\  & & \left. \left( \begin{bmatrix}
 {\boldone}^T \bm{\Sigmah}^{-1}  \boldone &  {\boldone}^T \bm{\Sigmah}^{-1}  \rb \\
 {\rb}^T \bm{\Sigmah}^{-1}  \boldone & {\rb}^T \bm{\Sigma}^{-1}  \rb \\
\end{bmatrix} \right)^{-1} {\rone}^T \bm{\Sigmah}^{-1} \x \beta_x \right] \nonumber \\   
&=& \frac{1}{	{\boldone}^T \bm{\Sigmah}^{-1}  \boldone {\rb}^T \bm{\Sigmah}^{-1}  \rb -   {\boldone}^T \bm{\Sigmah}^{-1}  \rb	{\rb}^T \bm{\Sigmah}^{-1}  \boldone  }  \times \\ & & \left[  \begin{bmatrix}
	{\rb}^T \bm{\Sigmah}^{-1}  \rb 	{\boldone}^T \bm{\Sigmah}^{-1} \boldone - {\boldone}^T \bm{\Sigmah}^{-1}  \rb 	{\rb}^T \bm{\Sigmah}^{-1} \boldone \\
	{\boldone}^T \bm{\Sigmah} ^{-1} \boldone  	{\rb}^T \bm{\Sigmah}^{-1} \boldone 	-	{\rb}^T \bm{\Sigmah}^{-1}  \boldone 	{\boldone}^T \bm{\Sigmah}^{-1} \boldone  \\
\end{bmatrix}  \beta_0 +  \right.\nonumber \\
& & \left.   \begin{bmatrix}
	{\rb}^T \bm{\Sigmah}^{-1}  \rb 	{\boldone}^T \bm{\Sigmah}^{-1} \x - {\boldone}^T \bm{\Sigmah}^{-1}  \rb 	{\rb}^T \bm{\Sigmah}^{-1} \x \\
	{\boldone}^T \bm{\Sigmah} ^{-1} \boldone 	{\rb}^T \bm{\Sigmah}^{-1} \x 	-	{\rb}^T \bm{\Sigmah}^{-1}  \boldone 	{\boldone}^T \bm{\Sigmah}^{-1} \x \\ 
\end{bmatrix}  \beta_x \right] \nonumber \\
&=& \frac{1}{ \sigmainvnorm{\boldone}^2 \sigmainvnorm{\rb}^2 -  \left[ \sigmainvdot{\rb}{\boldone} \right]^2  }  \times \\ & & \left[  \begin{bmatrix}
	\sigmainvnorm{\boldone}^2 \sigmainvnorm{\rb}^2 -  \left[ \sigmainvdot{\rb}{\boldone} \right]^2 \\
	0  \\
\end{bmatrix}  \beta_0 +  \right.\nonumber \\
& & \left.   \begin{bmatrix}
	\sigmainvnorm{\rb}^2  \sigmainvdot{\x}{\boldone} - \sigmainvdot{\rb}{\boldone} \sigmainvdot{\rb}{\x} \\
	\sigmainvnorm{\boldone}^2 	\sigmainvdot{\rb}{\x} 	-	\sigmainvdot{\rb}{\boldone} \sigmainvdot{\x}{\boldone} \\ 
\end{bmatrix}  \beta_x \right] \label{adj_bias_a}
\end{eqnarray*}

Now, we restrict our attention to the second component: 
\begin{eqnarray} \label{adj_bias_a2}
\textrm{A}^{S+}_2 \left( \rb, \x \right) &=& \frac{ \beta_x \left( \sigmainvnorm{\boldone}^2 	\sigmainvdot{\rb}{\x} 	-	\sigmainvdot{\rb}{\boldone} \sigmainvdot{\x}{\boldone}\right) } {  \sigmainvnorm{\boldone}^2 \sigmainvnorm{\rb}^2 -  \left[ \sigmainvdot{\rb}{\boldone} \right]^2 }.
\end{eqnarray}

\paragraph{Simplifying $\textrm{B}^{S+} \left( \rb, \x \right)$}
In this sub-section, we focus on the second term of \eqref{eq:adj_bias}. We note that the term is equivalent to the second term of \cref{app:obsglsfixed} with $\rb$ replacing $\x$. Therefore, we can borrow the result there to note: 

\begin{eqnarray*}
\textrm{B}^{S+} \left( \rb, \x \right) &=& \left( {\rone}^T \bm{\Sigmah}^{-1} \rone \right)^{-1} {\rone}^T \bm{\Sigmah}^{-1} \left( \beta_z  \z \right)   \\
&=&   \frac{\beta_z}{ \sigmainvnorm{\boldone}^2 \sigmainvnorm{\rb}^2 -  \left[ \sigmainvdot{\rb}{\boldone} \right]^2  }   \begin{bmatrix}
		\sigmainvnorm{\rb}^2 \sigmainvdot{\z}{\boldone}  - \sigmainvdot{\rb}{\boldone} \sigmainvdot{\rb}{\z}\\
	  \sigmainvnorm{\boldone}^2 \sigmainvdot{\rb}{\z} - 	\sigmainvdot{\rb}{\boldone}	\sigmainvdot{\z}{\boldone} \\
	\end{bmatrix}. 
\end{eqnarray*}

Restricting our attention to the second component, we can further simplify as follows:
\begin{eqnarray} \label{adj_bias_b2}
\textrm{B}^{S+}_2 \left( \rb, \x \right) &=& \frac{ \beta_z \left( \sigmainvnorm{\boldone}^2 \sigmainvdot{\rb}{\z} - 	\sigmainvdot{\rb}{\boldone}	\sigmainvdot{\z}{\boldone} \right) } { \sigmainvnorm{\boldone}^2 \sigmainvnorm{\rb}^2 -  \left[ \sigmainvdot{\rb}{\boldone} \right]^2}.
\end{eqnarray}

\paragraph{Bias for $\beta_x$}
Combining our results from \eqref{adj_bias_a2} and \eqref{adj_bias_b2}:

\begin{eqnarray*}
 	\textrm{E} \left(  \hat{\bm{\beta}}^{S+}  \right) - \beta_x &=& A^{S+}_2 \left( \rb, \x \right) - \beta_x + B^{S+}_2 \left( \rb, \x \right) =  \\ 
 & & \frac{ \beta_x \left( \sigmainvnorm{\boldone}^2 	\sigmainvdot{\rb}{\x} 	-	\sigmainvdot{\rb}{\boldone} \sigmainvdot{\x}{\boldone}\right) } {  \sigmainvnorm{\boldone}^2 \sigmainvnorm{\rb}^2 -  \left[ \sigmainvdot{\rb}{\boldone} \right]^2 } - \beta_x + \\
 & & 	 \frac{ \beta_z \left( \sigmainvnorm{\boldone}^2 \sigmainvdot{\x}{\z} - 	\sigmainvdot{\x}{\boldone}	\sigmainvdot{\z}{\boldone} \right) } { \sigmainvnorm{\boldone}^2 \sigmainvnorm{\x}^2 -  \left[ \sigmainvdot{\x}{\boldone} \right]^2}\\
 & = & \beta_x \left(\frac{ \left( \sigmainvnorm{\boldone}^2 	\sigmainvdot{\rb}{\x} 	-	\sigmainvdot{\rb}{\boldone} \sigmainvdot{\x}{\boldone}\right) } {  \sigmainvnorm{\boldone}^2 \sigmainvnorm{\rb}^2 -  \left[ \sigmainvdot{\rb}{\boldone} \right]^2 } - 1 \right) + \\
    & & 	 \frac{ \beta_z \left( \sigmainvnorm{\boldone}^2 \sigmainvdot{\x}{\z} - 	\sigmainvdot{\x}{\boldone}	\sigmainvdot{\z}{\boldone} \right) } { \sigmainvnorm{\boldone}^2 \sigmainvnorm{\x}^2 -  \left[ \sigmainvdot{\x}{\boldone} \right]^2}\\
  &=&  A_2^{S+*}(\rb, \x) + B^{S+}_2(\rb, \x).
\end{eqnarray*}

\subsubsection{Special Case of \cref{eq:obsadjglsfixed}}\label{subapp:nulladj} We assume that $\rb$ are the known residuals from a model of the form \eqref{eq:genericSpatial} with $\x$ as the response and only an intercept. We note that this means $\rb = \left[ \bm{I} - \boldone \left( \boldone^T \sigmahatgen{\x} \boldone \right)^{-1} \boldone^T \sigmahatgen{\x} \right] \x$. 

To further simplify \eqref{eq:adj_bias}, we note that $\rb = [ \x - \alpha \boldone ]$ with  $\alpha=\frac{ \sigmainvdotx{\x}{\boldone}}{\sigmainvnormx{\boldone}^2}\geq 0$. 

The second component of $A^{S+}(\rb, \x)$ in \eqref{adj_bias_a2} can be simplified further:
\begin{eqnarray} \label{adj_bias_a}
\textrm{A}^{S+}_2 \left( \rb, \x \right) &=& \frac{ \beta_x \left( \sigmainvnorm{\boldone}^2 	\sigmainvdot{\rb}{\x} 	-	\sigmainvdot{\rb}{\boldone} \sigmainvdot{\x}{\boldone}\right) } {  \sigmainvnorm{\boldone}^2 \sigmainvnorm{\rb}^2 -  \left[ \sigmainvdot{\rb}{\boldone} \right]^2 } \nonumber \\ 
&=& \frac{ \beta_x \left( \sigmainvnorm{\boldone}^2 	\sigmainvdot{\rb}{\x} 	-	\sigmainvdot{\rb}{\boldone} \sigmainvdot{\x}{\boldone}\right) } {  \sigmainvnorm{\boldone}^2 \sigmainvnorm{\x}^2 -  \left[ \sigmainvdot{\x}{\boldone} \right]^2} \nonumber \\ 
&=& \frac{ \beta_x \left( \sigmainvnorm{\boldone}^2 \sigmainvnorm{\x}^2 -  \left[ \sigmainvdot{\x}{\boldone} \right]^2 \right) } {  \sigmainvnorm{\boldone}^2 \sigmainvnorm{\x}^2 -  \left[ \sigmainvdot{\x}{\boldone} \right]^2} \nonumber \\ 
&=& \beta_x.
\end{eqnarray}

The second component of $B^{S+}(\rb, \x)$ can be further simplified as follows:
\begin{eqnarray} \label{adj_bias_b}
\textrm{B}^{S+}_2 \left( \rb, \x \right) &=& \frac{ \beta_z \left( \sigmainvnorm{\boldone}^2 \sigmainvdot{\rb}{\z} - 	\sigmainvdot{\rb}{\boldone}	\sigmainvdot{\z}{\boldone} \right) } { \sigmainvnorm{\boldone}^2 \sigmainvnorm{\rb}^2 -  \left[ \sigmainvdot{\rb}{\boldone} \right]^2} \nonumber\\
&=& \frac{ \beta_z \left( \sigmainvnorm{\boldone}^2 \sigmainvdot{\rb}{\z} - 	\sigmainvdot{\rb}{\boldone}	\sigmainvdot{\z}{\boldone} \right) } { \sigmainvnorm{\boldone}^2 \sigmainvnorm{\x}^2 -  \left[ \sigmainvdot{\x}{\boldone} \right]^2} \nonumber \\
&=&  \frac{ \beta_z \left( \sigmainvnorm{\boldone}^2 \sigmainvdot{\x}{\z} - 	\sigmainvdot{\x}{\boldone}	\sigmainvdot{\z}{\boldone} \right) } { \sigmainvnorm{\boldone}^2 \sigmainvnorm{\x}^2 -  \left[ \sigmainvdot{\x}{\boldone} \right]^2} 
\end{eqnarray}

Thus, combining our results from \eqref{adj_bias_a} and \eqref{adj_bias_b}:

\begin{eqnarray*}
 	\textrm{E} \left(  \hat{\bm{\beta}}^{S+}  \right) - \beta_x = A^{S+}_2 \left( \rb, \x \right) - \beta_x + B^{S+}_2 \left( \rb, \x \right) =  \\   
 	 \frac{ \beta_z \left( \sigmainvnorm{\boldone}^2 \sigmainvdot{\x}{\z} - 	\sigmainvdot{\x}{\boldone}	\sigmainvdot{\z}{\boldone} \right) } { \sigmainvnorm{\boldone}^2 \sigmainvnorm{\x}^2 -  \left[ \sigmainvdot{\x}{\boldone} \right]^2}. 
\end{eqnarray*}

\subsection{Derivation of Bias for $\hat{\beta}_x^{gSEM}$} \label{app:obsgsemfixed}

We assume that $\rb$ and $\ry$ are known residuals from separate spatial models fit with $\x$ and $\y$ as the response variables, respectively. 
We assume the models used to obtain these residuals gave an estimates of $\hat{\x} = \hat{\bm S}_x \x$ and $\hat{\y} = \hat{\bm S}_y \y$ so that $\rb = \left[\bm{I} - \hat{\bm S}_x\right]\x$ and $\ry = \left[\bm{I} - \hat{\bm S}_y\right]\y$ for some hat matrices $\hat{\bm S}_x$ and $\hat{\bm S}_y$.  For the following, we denote $\rone = \left[ \boldone ~ \rb \right]$.

\begin{eqnarray}
	\textrm{E} \left(  \hat{\bm{\beta}}^{gSEM}  \right)&=&  \left( {\rone}^T \rone \right)^{-1} {\rone}^T \ry\\
  &=&\left( {\rone}^T \rone \right)^{-1} {\rone}^T  \left(\bm{I} - \hat{\bm{S}}_y\right) \left( \beta_0 \boldone + \beta_x \x + \beta_z  \z \right) \nonumber \\ 
	&=&  \left( {\rone}^T \rone \right)^{-1} {\rone}^T \left(\bm{I} - \bm{S}_y\right) \left( \beta_0 \boldone + \beta_x \x  \right)  +  \left( {\rone}^T \rone \right)^{-1} {\rone}^T \left(\bm{I} -\hat{\bm{S}}_y\right) \left( \beta_z  \z \right) \nonumber \\ &=&
	\textrm{A}^{gSEM} \left( \rb, \x \right) + \textrm{B}^{gSEM} \left( \rb, \x \right) .\label{eq:gsem_bias}
\end{eqnarray}

\paragraph{Simplifying $\textrm{A}^{gSEM} \left( \rb, \x \right)$} In this sub-section, we focus on the first term of \eqref{eq:gsem_bias}.  
\begin{eqnarray}
\textrm{A}^{gSEM} \left( \rb, \x \right) &=& \left( {\rone}^T \rone \right)^{-1} {\rone}^T \left(\bm{I} - \hat{\bm{S}}_y\right) \left( \beta_0 \boldone + \beta_x \x  \right) \label{gsem_bias_a} \\
&=&  \left[ \left( \begin{bmatrix}
 	{\boldone}^T  \boldone &  {\boldone}^T  \rb \\
 	{\rb}^T  \boldone & {\rb}^T  \rb \\
 \end{bmatrix} \right)^{-1} {\rone}^T \left(\bm{I} - \hat{\bm{S}}_y\right) \boldone \beta_0 \,+  \right.  \nonumber \\  & & \left. \left( \begin{bmatrix}
 {\boldone}^T  \boldone &  {\boldone}^T  \rb \\
 {\rb}^T   \boldone & {\rb}^T  \rb \\
\end{bmatrix} \right)^{-1} {\rone}^T \left(\bm{I} - \hat{\bm{S}}_y\right) \x \beta_x \right] \nonumber \\   
&=& \frac{1}{	{\boldone}^T  \boldone {\rb}^T   \rb -   {\boldone}^T   \rb	{\rb}^T   \boldone  }  \times \nonumber \\ & & \left[  \begin{bmatrix}
	{\rb}^T \rb 	{\boldone}^T \left(\bm{I} - \hat{\bm{S}}_y\right) \boldone - {\boldone}^T  \rb 	{\rb}^T \left(\bm{I} - \hat{\bm{S}}_y\right) \boldone \\
	{\boldone}^T  \boldone  	{\rb}^T \left(\bm{I} - \hat{\bm{S}}_y\right)\boldone 	-	{\rb}^T  \boldone 	{\boldone}^T \left(\bm{I} - \hat{\bm{S}}_y\right) \boldone  \\
\end{bmatrix}  \beta_0 +  \right.\nonumber \\
& & \left.   \begin{bmatrix}
	{\rb}^T  \rb 	{\boldone}^T \left(\bm{I} - \hat{\bm{S}}_y\right) \x - {\boldone}^T   \rb 	{\rb}^T \left(\bm{I} - \hat{\bm{S}}_y\right) \x \\
	{\boldone}^T \boldone 	{\rb}^T \left(\bm{I} - \hat{\bm{S}}_y\right) \x 	-	{\rb}^T   \boldone 	{\boldone}^T \left(\bm{I} - \hat{\bm{S}}_y\right) \x \\ 
\end{bmatrix}  \beta_x \right] \nonumber \\
&=& \frac{1}{ {||\boldone||}^2 {||\rb||}^2 -  \left[ \langle \rb, \boldone\rangle \right]^2  }  \times \nonumber \\ & & \left[  \begin{bmatrix}
	{||\boldone||_{(\bm{I}-\hat{\bm{S}}_y)}^2} {||\rb||}^2 -  \langle\boldone, \rb\rangle \ISdot{\rb}{\boldone}\\
	0  \\
\end{bmatrix}  \beta_0 +  \right.\nonumber \\
& & \left.   \begin{bmatrix}
	\euclnorm{\rb}^2  \ISdot{\x}{\boldone} - \euclmetric{\rb}{\boldone} \ISdot{\rb}{\x} \\
	\euclnorm{\boldone}^2 	\ISdot{\rb}{\x} 	-	\euclmetric{\rb}{\boldone} \ISdot{\x}{\boldone} \\ 
\end{bmatrix}  \beta_x \right] \nonumber
\end{eqnarray}

Now, we restrict our attention to the second component: 
\begin{eqnarray} \label{gsem_bias_a2}
\textrm{A}^{gSEM}_2 \left( \rb, \x \right) &=& \frac{ \beta_x \left( \euclnorm{\boldone}^2 	\ISdot{\rb}{\x} 	-	\euclmetric{\rb}{\boldone} \ISdot{\x}{\boldone}\right) } {  \euclnorm{\boldone}^2 \euclnorm{\rb}^2 -  \left[ \euclmetric{\rb}{\boldone} \right]^2 }.
\end{eqnarray}

\paragraph{Simplifying $\textrm{B}^{gSEM} \left( \rb, \x \right)$}
In this sub-section, we focus on the second term of \eqref{eq:gsem_bias}. Skipping over simplifications that are equivalent to simplifications in \eqref{gsem_bias_a}: 

\begin{eqnarray*}
\textrm{B}^{gSEM} \left( \rb, \x \right) &=& \left( {\rone}^T  \rone \right)^{-1} {\rone}^T (\bm{I} - \hat{\bm{S}}_y) \left( \beta_z  \z \right)   \\
&=&   \frac{\beta_z}{ \euclnorm{\boldone}^2 \euclnorm{\rb}^2 -  \left[ \euclmetric{\rb}{\boldone} \right]^2  }   \begin{bmatrix}
		\euclnorm{\rb}^2 \ISdot{\z}{\boldone}  - \euclmetric{\rb}{\boldone} \ISdot{\rb}{\z}\\
	  \euclnorm{\boldone}^2 \ISdot{\rb}{\z} - 	\euclmetric{\rb}{\boldone}	\ISdot{\z}{\boldone} \\
	\end{bmatrix}. 
\end{eqnarray*}

Note, that restricting our attention to the second component, we can further simplify as follows:
\begin{eqnarray} \label{gsem_bias_b2}
\textrm{B}^{gSEM}_2 \left( \rb, \x \right) &=& \frac{ \beta_z \left( \euclnorm{\boldone}^2 \ISdot{\rb}{\z} - 	\euclmetric{\rb}{\boldone}	\ISdot{\z}{\boldone} \right) } { \euclnorm{\boldone}^2 \euclnorm{\rb}^2 -  \left[ \euclmetric{\rb}{\boldone} \right]^2}.
\end{eqnarray}

\paragraph{Bias for $\beta_x$}
To obtain the bias for the gSEM estimate of $\beta_x$, we combine our results from \eqref{gsem_bias_a2} and \eqref{gsem_bias_b2} and subtract $\beta_x$.

\section{Additional Simulation Results}\label{app:sim}

\begin{center}
\begin{minipage}{\textwidth}
\includegraphics[width=\linewidth]{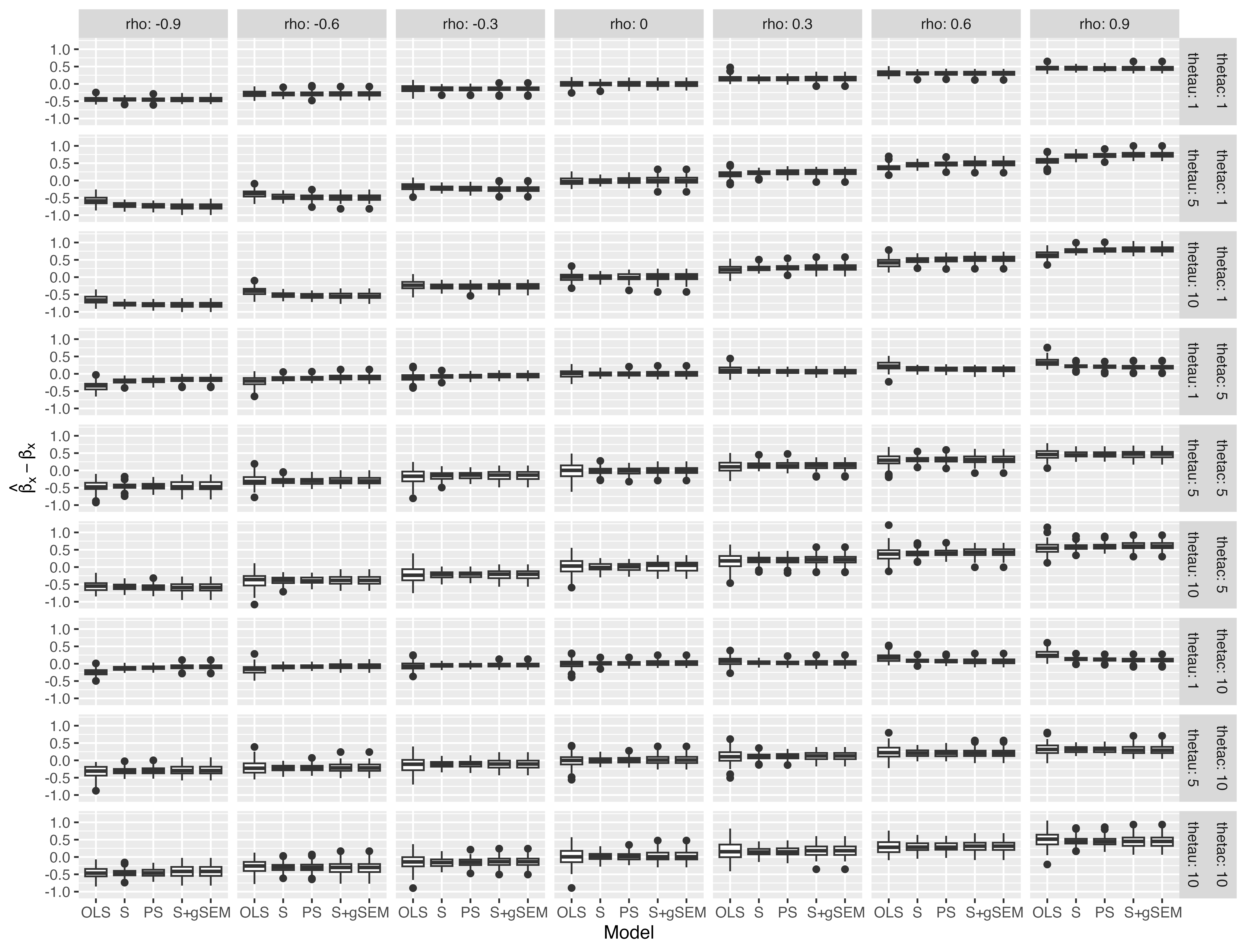}
\end{minipage}\hfill
\captionof{figure}{Boxplots of the estimate errors for different values of $(\theta_c, \theta_u, \rho)$ and five models. Plot made with the \texttt{ggplot2} package \citep{ggplot2}. }
\label{fig:bias_additional}
\end{center}

\end{document}